\documentclass[twoside,11pt]{article}

\usepackage{blindtext}

\usepackage{jmlr2e}

\usepackage{zy_notations2}

\usepackage{lastpage}
\jmlrheading{}{}{1-\pageref{LastPage}}{; Revised }{}{21-0000}{Zhanyu Wang, Arin Chang, and Jordan Awan}

\ShortHeadings{Optimal Debiased Inference on Privatized Data}{Wang, Chang, and Awan}
\firstpageno{1}

\begin{document}
\thispagestyle{empty}
\title{Optimal Debiased Inference on Privatized Data \\via Indirect Estimation and Parametric Bootstrap}

\author{\name Zhanyu Wang \email zhanyu.wang.purdue@gmail.com\\
       \addr Department of Statistics\\
       Purdue University\\
       West Lafayette, IN 47907, USA
       \AND
       \name Arin Chang \email chan1074@purdue.edu \\
       \addr Department of Statistics\\
       Purdue University\\
       West Lafayette, IN 47907, USA
       \AND
       \name Jordan Awan \email jaa557@pitt.edu \\
       \addr Department of Statistics\\
       University of Pittsburgh\\
       Pittsburgh, PA 15260, USA}

\editor{}

\maketitle

\begin{abstract}
We design a debiased parametric bootstrap framework for statistical inference from differentially private data. Existing usage of the parametric bootstrap on privatized data ignored or avoided handling possible biases introduced by the privacy mechanism, such as by clamping, a technique employed by the majority of privacy mechanisms. Ignoring these biases leads to under-coverage of confidence intervals and miscalibrated type I errors of hypothesis tests, due to the inconsistency of  parameter estimates based on the privatized data. We propose using the indirect inference method to estimate the parameter values consistently, and we use the improved estimator in parametric bootstrap for inference. To implement the indirect estimator, we present a novel simulation-based, adaptive approach along with the theory that establishes the consistency of the corresponding parametric bootstrap estimates, confidence intervals, and hypothesis tests. In particular, we prove that our adaptive indirect estimator achieves the minimum asymptotic variance among all ``well-behaved'' consistent estimators based on the released summary statistic. Our simulation studies show that our framework produces confidence intervals with well-calibrated coverage and performs hypothesis testing with the correct type I error, giving state-of-the-art performance for inference in several settings.
\end{abstract}


\section{Introduction}
\label{sec:intro}
In the age of big data, utilizing diverse data sources to train models offers significant benefits but also leaves data providers vulnerable to malicious attacks. 
To mitigate privacy concerns, \citet{dwork2006calibrating} introduced the concept of differential privacy (DP), which quantifies the level of privacy assurance offered by a data processing procedure. 
Following the advent of DP, numerous mechanisms have been developed to provide DP-guaranteed point estimates for parameters \citep{dwork2014algorithmic}. These algorithms inject additional randomness into the output to manage the tradeoff between its utility and the privacy guarantee. In addition to providing an accurate point estimator for a population parameter of interest, a central task in statistical analysis is frequentist inference—quantifying the uncertainty of a population parameter estimate derived from data via hypothesis tests (HT) or confidence intervals (CI). When applied under DP constraints, this task becomes notably challenging due to the randomness and biases introduced by privacy mechanisms.

This paper continues a line of work, which assumes that the DP mechanism and statistic are decided in advance, possibly by a separate entity; all inference must be done after-the-fact, based on knowldege of the privacy mechanism and a model for the underlying data. 
One of the most popular approaches of performing such inference for population parameters using privatized data is the parametric bootstrap (PB). \citet{du2020differentially} proposed multiple methods to construct DP CIs that differ in parameter estimation techniques, and all use PB to derive CIs through simulation. \citet{ferrando2022parametric} were the first to theoretically analyze the use of PB with DP guarantees. They validated the consistency of their CIs in two private estimation settings: exponential families and linear regression via sufficient statistic perturbation. \citet{alabi2022hypothesis} leveraged PB to conduct DP hypothesis testing specifically for linear regression. 
However, existing PB methods  typically do not account for biases introduced in the DP mechanism, such as by clamping, which can result in inaccurate inferences. As a result, they frequently produce biased estimators, leading to undercoverage of confidence intervals and miscalibrated type I error rates \citep{awan2025simulation}.

In this work, we address these challenges by proposing a novel adaptive indirect estimator which optimally post-processes DP summary statistics to produce a consistent estimate that achieves the minimum asymptotic variance among all ``well-behaved'' consistent estimators (defined in Definition \ref{def:wellbehaved}), which are a function of the released summary statistics. Our estimator is based on the method of indirect inference \citep{gourieroux1993indirect}, essentially reversing the generative process of the DP mechanism and estimating the parameter that would have most likely produced the observed privatized output. 
Based on our estimator, we develop a debiased parametric bootstrap framework to perform valid frequentist inference. We demonstrate that our approach produces well-calibrated CIs and valid HTs in several finite-sample simulation settings.

\subsection{Our contributions}
We propose a novel framework for frequentist statistical inference under differential privacy by optimally post-processing the DP-released statistic to obtain a consistent estimator of the parameter of interest. This estimator is used in a debiased inference pipeline based on indirect estimation and parametric bootstrap. Our main contributions are as follows:
\begin{itemize}
    \item We identify that the bias of existing private inference procedures primarily stem from directly using the DP output as a plug-in estimator. We address this by introducing an adaptive indirect estimator that leverages knowledge of the data-generating process and the DP mechanism. We prove that this estimator is consistent and achieves the minimum asymptotic variance among well-behaved consistent estimators which are based on the same DP summaries.
   \item We generalize the asymptotic analysis of the indirect inference method \citep{gourieroux1993indirect} to account for potential non-asymptotic normality, due to the 
   %
   DP noises, 
   and use PB to approximate the sampling distribution of the indirect estimator for better finite-sample performance. 
   \item  To address possible model mis-specification and to accommodate models that may not satisfy our assumptions, we also develop theory to justify when surrogate models can be used while maintaining the same asymptotic results.
    \item Through numerical simulations on statistical inference with private confidence sets and HTs, in the settings of location-scale normal, simple linear regression, logistic regression, and a naive Bayes classifier, we improve over state-of-the-art methods in terms of both validity and efficiency. 
\end{itemize}

While our methdology is motivated by privatized statistics, our framework can be applied in any setting where a low-dimensional summary statistic is available. Thus, our results may be of independent interest for other statistical problems.

\subsection{Related work}
\citet{jiang2004indirect} compared the indirect estimator to the generalized method of moments estimator \citep{hansen1982large}, and to other approaches using auxiliary or mis-specified models. 
\citet{gourieroux2010indirect} showed that the indirect estimator had finite sample properties superior to the generalized method of moment and the bias-corrected maximum likelihood estimator; 
\citet{kosmidis2014bias} compared indirect inference with other bias correction methods;
and \citet{guerrier2019simulation} gave a comprehensive review of the indirect inference method for bias correction, followed by \citet{guerrier2020asymptotically} and \citet{zhang2022flexible}.

Our work is also related to other simulation-based methods for DP statistical inference. 
\citet{wang2018statistical} used simulation to incorporate the classic central limit theorem and the DP mechanism to approximate the sampling distribution on privatized data. Using the sample-aggregate framework, \citet{evans2020statistically} privately estimated the proportion of their estimates affected by clamping for bias correction, and \citet{covington2025unbiased} built their DP algorithm to mitigate the effect of clamping on the sampling distribution. However, the methods by \citet{evans2020statistically} and \citet{covington2025unbiased} are restricted to specific privacy mechanisms. 
\citet{awan2025simulation} developed a simulation-based inference framework based on repro samples \citep{xie2022repro} to perform finite-sample inference on privatized data, which is closely related to our adaptive indirect estimator, but suffers from overly conservative inference and high computational cost.

For the problems with non-smooth data generating equations such as with discrete random variables, \citet{bruins2018generalized} approximated the non-smooth function using a smooth approximation, parameterized by $\lambda_n$, which converged to the original non-smooth equation when $\lambda_n\rightarrow 0$ as the sample size $n\rightarrow\infty$. 
\citet{ackerberg2009new} and \citet{sauer2021understanding} combined importance sampling with indirect inference when the likelihood function of the discrete data was smooth with respect to the parameter, which transformed the discrete optimization problem into a continuous optimization. \citet{frazier2019indirect} used a local change of variables technique to produce smooth approximations to derivatives, which may not be well-behaved in non-smooth settings. We use similar techniques as these papers in Section \ref{sec:pb:non-smooth} to extend our results to non-smooth settings as well.

 Of these related works, \citet{awan2025simulation} is most similar to ours. While \citet{awan2025simulation} offers conservative type I error guarantees, they do not have any theoretical results to support the power of their method. In contrast, we show that our method has asymptotically optimal power, and its superior performance is demonstrated in our numerical experiments.

\section{Background and motivation}
\label{sec:bg:dp}
In this section, we first illustrate the PB method and discuss the consistency of PB estimators, CIs, and HTs.
Then, we review the background of DP, including the definitions
and mechanisms used in this paper. Finally, we identify the problem of using biased estimates in PB for private inference through an HT example in the literature.

\subsection{Parametric bootstrap}\label{sec:pb:background_pb}
Let the sample dataset be $D:=(x_1,x_2,\ldots,x_n)$ with sample size $n$ and $x_i\iid F$, $i=1,\ldots,n$.

Bootstrap methods estimate the sampling distribution of a test statistic $\hat\tau(D)$ by the empirical distribution of the test statistic $\hat\tau(D_b)$, $b=1,\ldots, B$, calculated on synthetic data $D_b$ with data points sampled from $\hat F$, an estimation of the original data distribution $F$, based on $D$. 
%
Assuming that $F$ is parameterized by $\theta^*$, denoted by $F_{\theta^*}$, 
 PB resamples data by letting $\hat F:=F_{\hat \theta}$ where $\hat \theta$ is an estimate of $\theta^*$. In this paper, we focus on using PB under DP guarantees where we only observe the DP summary statistic $s$ calculated from $D$, but not $D$ itself. 
We formally define PB with $s$ below and visualize the procedure in Figure \ref{fig:diagramPB}.
\begin{definition}
    
    Let $s\in \Omega \subseteq \mathbb{R}^p$ be a random variable summarizing the observed dataset where $s\sim \mathbb{P}_{\theta^*}$, i.e., $s$ is defined on the probability space $(\Omega, \mathcal{F}, \mathbb{P}_{\theta^*})$ and following the distribution $\mathbb{P}_{\theta^*}: \mathcal{F} \rightarrow [0,1]$ parameterized by the true unknown parameter $\theta^*\in\Theta \subseteq \mathbb{R}^q$. For $d\leq q$, let $\tau:\Theta\rightarrow\mathbb{R}^d$ and let $\tau^*:=\tau(\theta^*)$ be the parameter of interest. Let
    $\hat\theta:\Omega\rightarrow\Theta$ and $\hat\tau:\Omega\rightarrow\mathbb{R}^d$ be measurable functions, and $\hat\theta(s)$ and $\hat\tau(s)$ are estimators of $\theta^*$ and $\tau^*$ respectively. The {PB estimator} of $\tau^*$ is defined by $\hat\tau(s_b)$ where $s_b | s \iid \mathbb{P}_{\hat\theta(s)}(\cdot | s)$, $b=1,\ldots,B$. We let $\mathbb{P}_{\hat\theta(s)}(\cdot | s)$ denote the Markov kernel for the law of $s_b$ given $s$, which is parameterized by $\hat\theta(s)$ in the same way as $\mathbb{P}_{\theta^*}$ parameterized by $\theta^*$. 
\end{definition}

\begin{figure}[t]
    \centering
    \includegraphics[width=0.45\textwidth, trim={0 0 280 0}, clip]{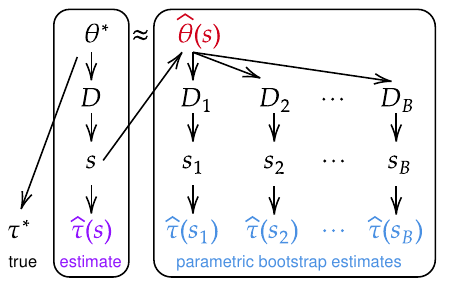}
    \caption{The illustration of parametric bootstrap. The parameter for data generation is $\theta^*$, and the parameter of interest is $\tau^*$. The left part of the figure denotes the point estimate $\hat\tau$, and the right part is the PB procedure. For non-private settings, we observe the dataset $D$, while under DP guarantees, we only observe the released statistic $s$.}
    \label{fig:diagramPB}
\end{figure}
    
The bootstrap method uses the empirical cumulative distribution function (CDF) of $\{\hat\tau(s_b)\}_{b=1}^B$ to approximate the sampling distribution of $\hat\tau(s)$ and construct confidence sets or perform HTs for $\tau^*$. In the remainder of this section, we characterize the consistency of PB estimators, confidence sets, and HTs.

\begin{definition}[PB consistency]\label{def:boot_consist}
    Let $s\sim \mathbb{P}_{\theta^*}$, $s_b\sim \mathbb{P}_{\hat\theta(s)}(\cdot | s)$, and $\hat\Sigma:\Omega \rightarrow \mathbb{R}^{d\times d}$ is measurable and non-singular almost surely.  
    Assume the PB estimator $\hat\tau(s_b)$ has the following asymptotic property:  $\|\hat\Sigma(s)^{-\frac{1}{2}}\l(\hat\tau(s) - \tau(\theta^*)\r)\|_p \xrightarrow[]{d} T_{\theta^*}$ where the CDF $F_{T_{\theta^*}}$ of $T_{\theta^*}$ is continuous, and $1\leq p\leq \infty$. Then $\hat\tau(s_b)$ is \emph{consistent} if 
    $$\mathbb{P}_{\hat\theta(s)}\l(\l\|\hat\Sigma(s_b)^{-\frac{1}{2}}\l(\hat\tau(s_b) - \tau(\hat\theta(s))\r)\r\|_p  \leq t ~\Big|~ s\r) \xrightarrow[]{P} F_{T_{\theta^*}}(t)$$ 
    for all $t\in \mathbb{R}$ with convergence in $\mathbb{P}_{\theta^*}$-probability.
\end{definition}
\begin{remark}
    Definition \ref{def:boot_consist} is a multivariate generalization of the bootstrap estimator consistency definition in Section 23.2 by \citet{van2000asymptotic}. The choice of $p$ can be tailored to the shape of the desired confidence sets and hypothesis tests; naturally we obtain hyperellipsoids for $p=2$, and hyperrectangles for $p=\infty$. Setting $\hat\Sigma(s)$ as an estimator of the covariance matrix of $\hat\tau(s)$ makes $\hat\Sigma(s)^{-\frac{1}{2}}\l(\hat\tau(s) - \tau(\theta^*)\r)$ an approximate pivot and improves performance, e.g.,  Figure \ref{fig:HT_comparison} shows the advantage of using the approximate pivot for HTs.
\end{remark}

    


\begin{definition}[Asymptotic consistency of confidence sets; \citet{van2000asymptotic}]\label{def:CI_consist}
    Let $s_n$ be a summary statistic of a dataset with sample size $n$ and $s_n\sim \mathbb{P}_{n,\theta^*}$.
    The confidence set $\hat C_n(s_n)$ for $\tau^*:=\tau(\theta^*)$ are (conservatively) \emph{asymptotically consistent} at level $1-\alpha$ if, for every $\theta^*\in\Theta$,  
    $\liminf_{n\rightarrow\infty} \mathbb{P}_{n,\theta^*}\big(\tau^* \in \hat C_n(s_n)\big) \geq 1-\alpha.$
\end{definition}

\begin{definition}[Asymptotic consistency of HTs; \citet{van2000asymptotic}]\label{def:HT_consist}
    For $s_n$ distributed as $\mathbb{P}_{n,\theta^*}$, the power function of the test between $H_0:\theta^*\in\Theta_0$ and $H_a:\theta^*\in\Theta_1$ is $\pi_n(\theta^*):=\mathbb{P}_{n,\theta^*}(T_n(s_n) \in K_n(s_n))$, i.e., we reject $H_0$ if the test statistic $T_n(s_n)$ falls into a critical region $K_n(s_n)$. A sequence of tests is called asymptotically consistent if for each $\alpha\in(0,1)$, it has a subsequence of tests with 
    $\limsup_{n\rightarrow\infty}\pi_n(\theta^*) \leq \alpha ~\mathrm{for~all}~ \theta^*\in\Theta_0$ and $\lim_{n\rightarrow \infty}\pi_n(\theta_1) = 1$ for all $\theta_1\in\Theta_1$. 
\end{definition}

With the consistent PB estimators, it is straightforward to build confidence sets and HTs that are asymptotically consistent. 
 In Lemma \ref{lem:CI_HT_consist}, for PB confidence sets, we generalize the existing result \citep[Lemma 23.3]{van2000asymptotic} on PB CIs to the multivariate scenario, and for PB HTs, we formally prove it as well, since we could not find a prior reference.

\begin{lemma}[Asymptotic consistency of PB inference]\label{lem:CI_HT_consist}
    Suppose that for every $\theta^*\in\Theta$, we have $\|\hat\Sigma(s)^{-\frac{1}{2}}\l(\hat\tau(s) - \tau(\theta^*)\r)\|_p \xrightarrow[]{d} T_{\theta^*}$ for a random variable $T_{\theta^*}$ with a continuous CDF and the PB estimator $\hat\tau(s_b)$ is consistent where $s_b\sim \mathbb{P}_{\hat\theta(s)}(\cdot | s)$. Let $\hat\xi_{1-\alpha}(s)$ be the $(1-\alpha)$-quantile of the distribution of $\l(\l\|\hat\Sigma(s_b)^{-\frac{1}{2}}\l(\hat\tau(s_b) - \tau(\hat\theta(s))\r)\r\|_p ~\Big|~ s\r)$.
    \begin{enumerate}
        \item The PB confidence sets $\l\{\hat\tau(s) - \hat\Sigma(s)^{\frac{1}{2}}\xi ~\Big|~ \|\xi\|_p \leq  \hat\xi_{1-\alpha}(s)\r\}$ are asymptotically consistent at level $1-\alpha$.  
        \item If $\hat\Sigma(s) \xrightarrow{P} 0$ and $\inf_{\theta^*\in\Theta_0}\|\tau(\theta^*) - \tau(\theta_1)\|_p > 0$ for all $\theta_1\in\Theta_1$, the sequence of PB HTs with test statistic $T_n(s):=\inf_{\theta^*\in\Theta_0}\l\|\hat\Sigma(s)^{-\frac{1}{2}}\l(\hat\tau(s) - \tau(\theta^*)\r)\r\|_p$ and critical region $K_n(s):=(\hat\xi_{1-\alpha}(s), \infty)$ for any level $\alpha$ is asymptotically consistent. 
    \end{enumerate}
\end{lemma}


From Lemma \ref{lem:CI_HT_consist}, the consistency of PB estimators $\hat\tau(s_b)$ is essential for the validity of the confidence sets and HTs based on PB.
\citet{beran1997diagnosing} and \citet{ferrando2022parametric} showed that the asymptotic equivariance of $\hat\tau$ guarantees the consistency of $\hat\tau(s_b)$ if $\hat\theta(s) - \theta^* = O_p\left(1/\sqrt{n}\right)$ and $\hat\Sigma(s)=\frac 1n I_d$ for all $s$. 
We prove Proposition \ref{prop:parambootconsistency} with a generic choice for $\hat\Sigma$.

\begin{definition}[Asymptotic equivariance; \citet{beran1997diagnosing}]\label{def:AE}
    
    Let $Q_n(\theta^*)$ be the distribution of $\sqrt{n}(\hat\tau(s) - \tau(\theta^*))$ where $s\sim \mathbb{P}_{\theta^*}$. We say that
    $\hat\tau$ is {asymptotically equivariant} if $Q_n\left(\theta^* + h_n/\sqrt{n}\right)$ converges to a limiting distribution $H(\theta^*)$ for all convergent sequences $h_n \rightarrow h$.
\end{definition}

\begin{proposition}[PB consistency]\label{prop:parambootconsistency}
Suppose $\sqrt{n}(\hat\theta(s) - \theta^*) \xrightarrow[]{d} J(\theta^*)$, $\hat\tau$ is asymptotically equivariant with continuous $H(\theta^*)$, and  $n\hat\Sigma(s) \xrightarrow[]{P} \Sigma(\theta^*)$ for all convergent sequences $h_n \rightarrow h$ and $s\sim \mathbb{P}_{\theta^* + h_n/\sqrt{n}}$. Then, the PB estimator $\hat\tau(s_b)$ is consistent.
\end{proposition}

%
Sampling $s\sim \mathbb{P}_{\theta^*}$ can often be achieved by a data generating equation $s:=s(D)$, $D:=G(\theta^*, u)$ where $D$ is the dataset, $G$ is a deterministic function named as the data generating equation, $u$ is the source of uncertainty in $D$ named as the random seed, such that $u \sim \mathbb{P}_u$ where $\mathbb{P}_u$ is known and does not depend on $\theta^*$.  The structure of a data generating equation will be used in the indirect inference simulations to isolate the effect of the parameter from the random seed, $u$.

\begin{example}
    
    Let the dataset be $D:=(x_1,\ldots,x_n)$. If $x_i\iid N(\mu,\sigma^2)$, $i=1,\ldots,n$, we can represent $D$
    as $D:=\mu + \sigma u$ where $u \sim N(0,I_{n\times n})$ and the parameter is $\theta^*:=(\mu, \sigma)$. 
\end{example}

{
\subsection{Differential privacy}
\label{sec:bg:dp}
In this paper, we use $\ep$-DP \citep{dwork2006calibrating} and Gaussian DP (GDP) \citep{dong2022gaussian} in our examples, although our results also apply to other DP notions.

A mechanism is a randomized function $M$ that takes a dataset $D$ of size $n$ as input and outputs a random variable or vector $s$. The Hamming distance between two datasets with the same sample sizes is $d(D, D')$, the number of entries in which $D$ and $D'$ differ. 

\begin{definition}[\citealp{dwork2006calibrating}]\label{def:ep}
    For $\ep>0$, $M$ satisfies $\ep$-DP if for any $D$ and $D'$ such that $d(D,D')\leq 1$, we have $P(M(D)\in S) \leq \exp(\ep) P(M(D')\in S)$ for every measurable set $S$.
\end{definition}
 Note that a DP guarantee is a property of the mechanism, which must hold for all pairs of neighboring datasets. In particular, the mechanism must be specified independent of the input dataset.

\begin{definition}[\citealp{dong2022gaussian}]\label{def:gdp}
    For $\mu>0$, $M$ satisfies $\mu$-GDP if for any $D$ and $D'$ such that $d(D,D')\leq 1$, any hypothesis test between $H_0: s\sim M(D)$ and $H_1: s\sim M(D')$ has a type II error bounded below by $\Phi(\Phi^{-1}(1-\alpha)-\mu)$, where $\alpha$ is the type I error. 
\end{definition}

Definition \ref{def:gdp} means that  $M$ is $\mu$-GDP if for any $d(D,D')\leq 1$, it is harder to distinguish $M(D)$ from $M(D')$ than to distinguish a sample drawn from either $N(0,1)$ or $N(\mu,1)$. 


 Common DP mechanisms are the Laplace and Gaussian mechanisms, which add noise scaled to the sensitivity of a statistic. 
    The $\ell_p$-sensitivity of a function $g$ is  $\Delta_p(g):=\sup\limits_{d(D,D')\leq 1}\|g(D)-g(D')\|_p$.  
%
%
%
    The Laplace mechanism on $g$ is defined as $M_{\text{L},g,b}(D):=g(D)+(\xi_1,\ldots,\xi_d)$ where $\xi_i\iid\mathrm{Laplace}(b)$ with probability density function (pdf) $p(x):=\frac{1}{2b}\ee^{-|x|/b}$, $i=1,\ldots,d$. $M_{\text{L},g,b}(D)$ satisfies $\ep$-DP if $b=\Delta_1(g)/\ep$ \citep{dwork2014algorithmic}.
    The Gaussian mechanism on $g$ is $M_{\text{G},g,\sigma}(D):=g(D)+\xi$ where $\xi\sim N(\mu=0,\Sigma=\sigma^2 I_{d\times d})$. $M_{\text{G},g,\sigma}(D)$ satisfies $\mu$-GDP if $\sigma^2={(\Delta_2(g) / \mu)^2}$ \citep{dong2022gaussian}.  It is common to clamp data to a pre-specified range in order to limit the sensitivity: for $x\in \RR$ and $L<U\in \RR$,  we define the clamp function as $[x]_L^U = \max\{\min\{x,U\},L\}$.

\begin{restatable}[Composition; \citet{dwork2010boosting, dong2022gaussian}]{proposition}{propfdpcomposition}\label{prop:fdp_composition}
    Let $M_i:\mathcal{D}\rightarrow\mathcal{Y}_i$ for $i=1,\ldots,k$, and $M:=(M_1,\ldots,M_k):\mathcal{D}\rightarrow(\mathcal{Y}_1\times\ldots\times\mathcal{Y}_k)$. 1.  If $M_i$ satisfies $\ep_i$-DP for $i=1,\ldots,k$, then $M$ satisfies $(\ep_1+\ldots+\ep_k)$-DP. 2. 
  If $M_i$ satisfies $\mu_i$-GDP for $i=1,\ldots,k$, then $M$ satisfies $\sqrt{\mu_1^2+\ldots+\mu_k^2}$-GDP.
\end{restatable} 

 The parameters $\ep$ and $\mu$ are referred to as the \emph{privacy (loss) budget}; this is because when using composition, the ``budget'' needs to be split among the different statistics of interest.

 {\bf Security without Obscurity and Invariance to Postprocessing} Two other important properties of DP are: 1) The details of the privacy mechanism can be made fully public without compromising the DP guarantee. This includes clamping parameters, sensitivity calculations, and the distribution of noise in the mechanism. 2) Applying any data-independent procedure to the output of a DP mechanism does not weaken the privacy guarantee. These two properties allow us to incorporate the privacy mechanism and a model for the underlying data to achieve valid statistical inference on the underlying parameters, without changing the DP guarantee.


\subsection{Biased DP estimators and inaccurate inference results}
\label{sec:bg:biased_estimation}
As PB only requires a point estimator, $\hat\theta(D)$, for generating $s_b$ and $\hat\tau(s_b)$, the privacy guarantee for the private statistical inference is the same as the DP point estimator because of the post-processing property \citep{dwork2014algorithmic}. However, DP mechanisms often use clamping to ensure bounded sensitivity, which leads to biased estimators and potentially inaccurate inference results.

In real-world applications, improper analysis of  DP outputs can cause severe problems. For example \citet{fredrikson2014privacy} found that in an application of Warfarin dosing ``...for privacy budgets effective at preventing attacks, patients would be exposed to increased risk of stroke, bleeding events, and mortality.'' Similarly, \citet{kenny2021use} found that ``...the [differentially private method used in the 2020 Decennial Census] systematically undercounts the population in mixed-race and mixed-partisan precincts, yielding unpredictable racial and partisan biases.'' 
These examples highlight the importance of properly accounting for the additional bias and variance introduced by a DP mechanism.

We share a motivating example from the existing literature to demonstrate the inaccurate inference results from the naive use of PB methods with a DP estimator. \citet{alabi2022hypothesis} developed DP HTs for the slope in a simple linear regression using PB, and Figure \ref{fig:clamp_issues} shows the mis-calibrated type I error reaching 67\% under significance level 0.05 of using the DP Monte Carlo tests \citep{alabi2022hypothesis}. 

\begin{figure}
    \centering
    \includegraphics[width=0.75\linewidth]{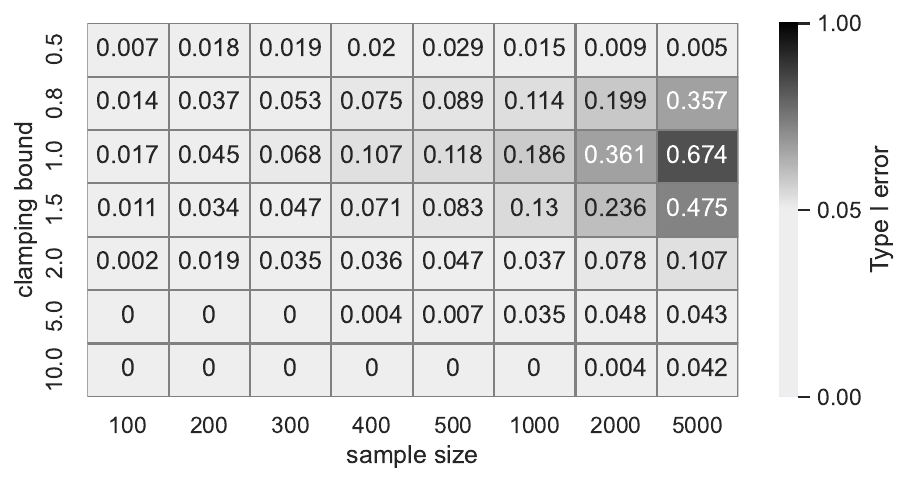}
    \caption{Type I error of private HTs on $H_0: \beta^*_1 = 0$ and $H_1: \beta^*_1 \neq 0$ with a linear regression model $Y = \beta^*_0 + X \beta^*_1 + \epsilon$ using DP PB tests \citep{alabi2022hypothesis} with nominal significance level 0.05. Results are from \citet[Figure 5]{awan2025simulation}. Simulation details are the same as in Section \ref{sec:pb:HT}.}\label{fig:clamp_issues}
\end{figure}

\begin{remark}
    Clamping is not the only reason for the bias in DP outputs. For example, the objective perturbation \citep{chaudhuri2011differentially, kifer2012private} is a DP mechanism for empirical risk minimization, and it uses a \emph{regularized} risk function which also results in a biased estimator \citep{wang2015privacy}. 
\end{remark}

}

\begin{remark}
   There are a few notable works that offer finite-sample DP inference, even involving clamping. For example, \citet{karwa2018finite} developed finite-sample DP confidence intervals for the mean of normally distributed data. Some other works offer asymptotic guarantees, sometimes even using the parametric bootstrap (e.g., \citep{alabi2022hypothesis}). In both cases, these works assume that the clamp is increasing at an appropriate rate so that the bias is asymptotically vanishing. While this can be fine for theoretical works, in practice it can be difficult to determine whether the clamp introduces a significant bias at a fixed sample size. Going forward, we propose a method that can correct for the bias due to clamping, whether or not the clamping bounds change with $n$. 
\end{remark}

\section{Debiased parametric bootstrap via indirect inference}
\label{sec:pb:method}

In this section, we first describe the traditional indirect estimator \citep{gourieroux1993indirect}, which can solve the bias issue in the clamping procedure of the DP mechanisms and give valid PB inference. Then, we propose a novel adaptive indirect inference estimator, that automatically optimizes the covariance matrix of the indirect estimator. Finally, we provide theoretical guarantees for using the (adaptive) indirect estimator in PB, showing that it has optimal asymptotic variance and also gives valid PB inference.

\subsection{Indirect estimators}

The underlying principle of the indirect estimator is to fix the ``random seeds'' for synthetic data generation and find the parameter that generates synthetic statistics most similar to the observed statistic.
We describe the indirect estimator with additional consideration of the DP mechanisms used in releasing the observed statistic. 

Similar to the data generating equation in PB, we assume that the dataset is 
$D:=G(\theta^*, u)$ where $\theta^* \in \Theta \subseteq \mathbb{R}^q$ is an unknown parameter, $u$ is the source of uncertainty following a known distribution $F_u$ that does not depend on $\theta$, and $G$ is a deterministic function. 
Let $s\in \mathbb{B} \subseteq \mathbb{R}^p$ be the released statistic calculated from $D$ by minimizing a criterion $\rho_n$, i.e., 
\[s := \argmin_{\beta\in \mathbb{B}} \rho_n(\beta; D, u_{\mathrm{DP}}),\] 
where $u_{\mathrm{DP}}\sim F_{\mathrm{DP}}$ denotes the source of uncertainty in the DP mechanisms. While $s$ is often an estimator for $\theta^*$, it could also be any set of summary statistics informative for $\theta^*$. 

\begin{remark}
    
    The optimization-form definition of $s$ includes the M-estimator as a special case which does not use $u_{\mathrm{DP}}$. It is useful for DP mechanisms such as objective perturbation \citep{chaudhuri2011differentially}, and it is also compatible with DP mechanisms like $s := \phi_n(D, u_{\mathrm{DP}})$, where $\phi_n$ is a generic function that may depend on $n$, since we can define $\rho_n(\beta; D, u_{\mathrm{DP}}):=( \beta- \phi_n(D, u_{\mathrm{DP}}))^2$.
\end{remark}

We simulate $u^r\iid F_u$ and $u_{\mathrm{DP}}^r\iid F_{\mathrm{DP}}$, $r=1,\ldots,R$, and keep them fixed. For  a candidate value of $\theta$, we generate synthetic datasets $D^r(\theta):=G(\theta, u^r)$ and calculate
\[s^r(\theta) := \argmin_{\beta\in \mathbb{B}} \rho_n(\beta; D^r(\theta), u_{\mathrm{DP}}^r).\]
We denote the sample mean and covariance of $\{s^r(\theta)\}_{r=1}^{R}$ as  $m^{R}(\theta):=\frac{1}{R}\sum_{r=1}^R s^r(\theta)$ and $S^{R}(\theta):=\frac{1}{R-1}\sum_{r=1}^R \l(s^r(\theta) - m^{R}(\theta)\r) \l(s^r(\theta) - m^{R}(\theta)\r)^\intercal$. Define $u^{[R]} := (u^1, \ldots, u^R)$, $u_{\mathrm{DP}}^{[R]} := (u_{\mathrm{DP}}^1,\cdots, u_{\mathrm{DP}}^R)$, $\|a\|^2_{\Omega}:= a^\intercal \Omega a$ for $a\in\mathbb{R}^p$ and $\Omega \in \mathbb{R}^{p\times p}$. 

\begin{definition}[Indirect estimator; \citet{gourieroux1993indirect}]\label{def:indirect_est}
    
    Let $\Omega_n\in \RR^{p\times p}$ be a sequence of positive definite matrices. 
    The indirect estimator is defined as follows,
\begin{equation}\label{eq:theta_ind}
\hat{\theta}_{\mathrm{IND}}(s, u^{[R]}, u_{\mathrm{DP}}^{[R]}) := \argmin_{\theta \in \Theta} \left\|s - m^{R}(\theta)\right\|_{{\Omega}_n}.    
\end{equation}
\end{definition}

The intuition behind indirect inference is as follows: search the parameter space for a parameter such that when simulating from that parameter, the average value of the summaries $\frac 1R \sum_{r=1}^R s^r(\theta)$  is close to the observed summary $s$. Note that indirect inference does not require an explicit estimator, only a summary statistic. See Figure \ref{fig:diagramIND} for an illustration.

\begin{figure}[t]
    \centering
    \includegraphics[width=0.75\linewidth, trim={70 25 140 15}, clip]{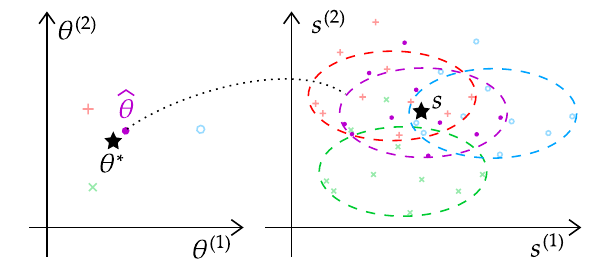}
    \caption{The illustration of the indirect estimator $\hat\theta$. The left subfigure is the space of parameter $\theta\in\mathbb{R}^2$, and the right subfigure is the space of $s\in\mathbb{R}^2$. The true parameter is $\theta^*$, and the observed statistic is $s$, where both are denoted by $\bigstar$. We let $\hat\theta$ be a minimizer of the distance between the observed $s$ and the generated statistics $s^r\iid\mathbb{P}_{\hat\theta}$, $r=1,\ldots, R$.}
    \label{fig:diagramIND}
\end{figure}

In Theorem \ref{thm:indirect_est}, we show that $\hat{\theta}_{\mathrm{IND}}$ is a consistent estimator of $\theta^*$, and the choice of $\Omega_n$ affects the asymptotic variance of $\hat{\theta}_{\mathrm{IND}}$. However, the optimal $\Omega_n$ may depend on $\theta^*$ and require additional effort to find a good estimation. For a novel and computationally efficient approach, we propose an adaptive indirect inference procedure which uses the inverse of the sample covariance matrix of $\{s^r(\theta)\}_{r=1}^R$ as an adaptive $\Omega_n$ and show its asymptotic optimality in the last part of Theorem \ref{thm:ada_indirect_est}.

\begin{definition}[Adaptive indirect estimator]\label{def:ada_indirect_est}
    The adaptive indirect estimator of $\theta^*$ is 
\begin{equation}\label{eq:theta_adi}
\hat{\theta}_{\mathrm{ADI}}(s, u^{[R]}, u_{\mathrm{DP}}^{[R]}) := \argmin_{\theta \in \Theta}\l\|s - m^{R}(\theta)\r\|_{\l(S^{R}(\theta)\r)^{-1}}.    
\end{equation}
\end{definition}

\begin{remark}
     If there exists a value $\theta$ such that $s-\frac{1}{R} \sum_{r=1}^R s^r(\theta) = 0$, then $\hat{\theta}_{\mathrm{IND}}$ and $\hat{\theta}_{\mathrm{ADI}}$ are equal, and are also equivalent to the Just Identified Indirect Inference (JINI) estimator \citep{zhang2022flexible}. When there is no such $\theta$, the choice of covariance matrix will influence the estimator.
    \citet{jiang2004indirect} proposed the adjusted estimator, $\hat\theta := \argmin_{\theta\in\Theta} \|s-b(\theta)\|_{V_n^{-1}}$, where they chose $b(\theta)$ to be the limit of $s^r(\theta)$,  the sequence of $V_n$ as a sample estimate of the covariance of $s$, and show that this choice has certain optimality properties.  On the other hand, our $\hat{\theta}_{\mathrm{ADI}}$ is fully data-driven, without the need for a covariance estimator (which would require expending additional privacy loss budget) and in next section, we show that $\hat{\theta}_{\mathrm{ADI}}$ enjoys optimal asymptotic properties.
    \end{remark}

\subsection{Parametric bootstrap with indirect estimator}
We denote the indirect estimator or the adaptive indirect estimator by $\hat\theta(s)$. For $b=1,\ldots, B$, we sample $u_b\iid F_u$ and $u_{\mathrm{DP},b}\iid F_{\mathrm{DP}}$ to generate $D_b:=G(\hat\theta(s),u_b)$ and calculate $s_b:=\argmin_{\beta\in \mathbb{B}} \rho_n(\beta; D_b, u_{\mathrm{DP},b})$.  For inference on $\tau(\theta^*)\in\mathbb{R}^d$ with consistency guarantees, we choose $\hat\tau:\Omega\rightarrow\mathbb{R}^d$ and $\hat\Sigma:\Omega\rightarrow\mathbb{R}^{d\times d}$ that satisfy the assumptions of Proposition \ref{prop:parambootconsistency}. For better performance, we can set $\hat\Sigma(s)$ as a consistent estimator of the covariance matrix of $\hat\tau(s)$. This studentization step turns $\hat\Sigma(s)^{-1/2}(\hat\tau(s)-\tau(\theta^*))$ into an approximate pivot, which reduces sensitivity to scale/units and nuisance parameters and typically improves finite-sample calibration; see Remark 2.3 for intuition and Figure \ref{fig:HT_comparison} for empirical evidence of the gains. If such an estimator is not readily available, we can use $\hat\Sigma(s):=\frac{1}{n}I_d$ by default.

We assume $(B+1)(1-\alpha) \geq 1$. Following Lemma \ref{lem:CI_HT_consist}, we define $\hat\xi_{(j)}$ as the $j$th order statistic of $\l\{\l\|\hat\Sigma(s_b)^{-\frac{1}{2}}\l(\hat\tau(s_b) - \tau(\hat\theta(s))\r)\r\|_p\r\}_{b=1}^B$. Then the $(1-\alpha)$ PB confidence set is
\[
\hat C_{1-\alpha}(s)=\l\{\hat\tau(s)-\hat\Sigma(s)^{\frac{1}{2}}\xi ~\Big|~ \|\xi\|_p \leq \hat\xi_{(\lfloor (B+1)(1-\alpha)\rfloor)}\r\}.
\]
For HTs, we compute the test statistic $T:=\inf_{\theta^*\in\Theta_0}\l\|\hat\Sigma(s)^{-\frac{1}{2}}\l(\hat\tau(s) - \tau(\theta^*)\r)\r\|_p$ and compare it to the critical region $(\hat\xi_{(\lfloor (B+1)(1-\alpha)\rfloor)},\infty)$.

The usage of PB with the indirect estimator is summarized in the appendix as Algorithm \ref{alg:indCI} and Algorithm \ref{alg:indHT} for building confidence sets and conducting HTs, respectively.


\subsection{Asymptotic properties of the indirect estimators}\label{sec:pb:asymp_theory}
This section  provides Theorem \ref{thm:indirect_est} and Theorem \ref{thm:ada_indirect_est} for the asymptotic properties of the indirect estimator and the adaptive indirect estimator, respectively, including their consistency, asymptotic distribution, and asymptotic equivariance. Theorem \ref{thm:ada_indirect_est} also shows that the adaptive indirect estimator enjoys the optimal asymptotic variance among well-behaved consistent estimators built from a given statistic. Finally, Theorem \ref{thm:indirect_inf_valid_est} and Corollary \ref{cor:indirect_inf_valid_CI_HT} provide the consistency of PB confidence sets and HTs with the (adaptive) indirect estimator.

To prove the consistency of $\hat{\theta}_{\mathrm{IND}}$, we have the following assumptions.
\\
\noindent \textbf{(A1)} {There exists $\rho_\infty(\beta; F_u, F_{\mathrm{DP}}, \theta^*)$ that is non-stochastic and continuous in $\beta$, and\\ $\sup_{\beta\in \mathbb{B}} |\rho_n(\beta; D, u_{\mathrm{DP}}) - \rho_\infty(\beta; F_u, F_{\mathrm{DP}}, \theta^*)| \xrightarrow{\mathrm{P}} 0$.}
\\
\noindent \textbf{(A2)} {
Define $b(\theta) := \argmin_{\beta\in \mathbb{B}} \rho_\infty(\beta; F_u, F_{\mathrm{DP}}, \theta)$, $\beta^* := b(\theta^*)$, $s := \argmin_{\beta\in \mathbb{B}} \rho_n(\beta; D, u_{\mathrm{DP}})$, and $\beta^*$ and $s$ are interior points in $\mathbb{B}$ for every $D$ and $u_{\mathrm{DP}}$. The mapping $b(\cdot)$ is continuous and injective from $\Theta$ to $\mathbb{B}$, i.e., $ \rho_\infty(b(\theta), F_u, F_{\mathrm{DP}}, \theta) < \rho_\infty(\beta; F_u, F_{\mathrm{DP}}, \theta)$ for all $\beta\neq b(\theta)$, and $b(\theta)\neq b(\theta')$ for all $\theta\neq \theta'$. The Jacobian matrix $\frac{\partial b}{\partial \theta^\intercal}$ exists and is full-column rank.}
\\
\noindent \textbf{(A3)} {$\sup_{\theta \in \Theta}\|s^r(\theta) - b(\theta)\| \xrightarrow{\mathrm{P}} 0$,  as $n\rightarrow \infty$}.
\\
\noindent \textbf{(A4)} {$\{\Omega_n\}_{n=1}^\infty$ is a sequence of deterministic positive definite matrices converging to a deterministic positive definite matrix $\Omega$.} 
\\
With additional assumptions \textbf{(A5, A6, A7)}, we can derive the asymptotic distribution of $\sqrt{n}(\hat{\theta}_{\mathrm{IND}} - \theta^*)$ and prove that $\hat{\theta}_{\mathrm{IND}}$ is asymptotically equivariant.
\\
\noindent \textbf{(A5)} {For every $(\theta, u^r, u_{\mathrm{DP}}^r)$, $\frac{\partial s^r(\theta)}{\partial \theta}$ and $\frac{\partial b(\theta)}{\partial\theta}$ exist and are continuous in $\theta$, and 
$\frac{\partial s^r(\theta)}{\partial \theta} \xrightarrow{\mathrm{P}} \frac{\partial b(\theta)}{\partial\theta}$. We denote $B^*\defeq\frac{\partial b(\theta^*)}{\partial \theta}$. For every $(s, u^{[R]}, u_{\mathrm{DP}}^{[R]})$, $\hat{\theta}_{\mathrm{IND}}(s, u^{[R]}, u_{\mathrm{DP}}^{[R]})$ is an interior point in $\Theta$.}
\\
\noindent \textbf{(A6)} {For every $(\beta, D, u_{\mathrm{DP}})$, $\frac{\partial \rho_n(\beta; D, u_{\mathrm{DP}})}{\partial \beta}$ exists and is continuous in $\beta$;
$\sqrt{n} (-\frac{\partial \rho_n(\beta^*; D, u_{\mathrm{DP}})}{\partial \beta}) \xrightarrow{\mathrm{d}} F_{\rho,u,\mathrm{DP}}^*$.}
\\
\noindent \textbf{(A7)} {For every $(\beta, D, u_{\mathrm{DP}})$, $\frac{\partial^2 \rho_n(\beta; D, u_{\mathrm{DP}})}{(\partial \beta)(\partial \beta^\intercal)}$ exists and is continuous in $\beta$; 
$\frac{\partial^2 \rho_n(\beta^*; D, u_{\mathrm{DP}})}{(\partial \beta)(\partial \beta^\intercal)}\xrightarrow{\mathrm{P}} \frac{\partial^2 \rho_\infty(\beta^*; F_u, F_{\mathrm{DP}}, \theta^*)}{(\partial \beta)(\partial \beta^\intercal)} =: J^*$.}

    
    \textbf{(A1, A3)} are for the consistency of $s$ and $\hat{\theta}_{\mathrm{IND}}$ as they are M-estimators \citep{van2000asymptotic}. \textbf{(A2)} is for the identifiability of $\theta^*$ using $s$. \textbf{(A4)} generalizes the $\ell_2$ norm for more efficient $\hat{\theta}_{\mathrm{IND}}$. \textbf{(A5, A6, A7)} are for the Taylor expansion to obtain asymptotic distributions of $s$ and $\hat{\theta}_{\mathrm{IND}}$ which requires us to have the true model $D=G(\theta^*, u)$ continuous in $\theta^*$ given any $u$. 

\begin{theorem}[IND asymptotics]\label{thm:indirect_est} 
    For $\hat{\theta}_{\mathrm{IND}}$ defined in \eqref{eq:theta_ind}, we have the following:
\begin{enumerate}
    \item Under \textbf{(A1--A4)}, $\hat{\theta}_{\mathrm{IND}}$ is a consistent estimator of $\theta^*$.
    \item Under \textbf{(A1--A7)}, the asymptotic distribution of $\sqrt{n}(\hat{\theta}_{\mathrm{IND}} - \theta^*)$ is the same as the one of $\left((B^*)^{\intercal} {\Omega} B^*\right)^{-1} (B^*)^{\intercal} {\Omega} (J^*)^{-1} (v_0 - \frac{1}{R}\sum_{r=1}^R v_r)$ where $v_r \iid F_{\rho,u,\mathrm{DP}}^*$ for $r=0,1,\ldots,R$.
    \item Under \textbf{(A1--A7)}, $\hat{\theta}_{\mathrm{IND}}$ and $s$ are asymptotically equivariant.
\end{enumerate}
\end{theorem}

For the asymptotic properties of $\hat{\theta}_{\mathrm{ADI}}$, we require $R\rightarrow\infty$ when $n\rightarrow\infty$, so we use $\{R_n\}_{n=1}^{\infty}$ to replace the constant $R$. We also have two additional assumptions.

\noindent \textbf{(A8)} For any $\theta \in \Theta$, we have $\lim\limits_{n\rightarrow\infty}\Var(\sqrt{n}(s^r(\theta) - b(\theta))) = \Var(\lim\limits_{n\rightarrow\infty}\sqrt{n}(s^r(\theta) - b(\theta))) =: \Sigma(\theta)\succ0$, $\Sigma(\theta)$ is continuous in $\theta$. For $\{R_n\}_{n=1}^{\infty}$ with $\lim\limits_{n\rightarrow\infty} R_n = \infty$, we have 
$\|\sqrt{n}(m^{R_n}(\theta^*) - \beta^*)\|_2=o_p(1)$, and uniformly for all $\theta$ we have $m^{R_n}(\theta) \xrightarrow{\mathrm{P}} b(\theta)$, $n S^{R_n}(\theta) \xrightarrow{\mathrm{P}} \Sigma(\theta)$.
\\
\noindent \textbf{(A9)} {For every $(s, u^{[R_n]}, u_{\mathrm{DP}}^{[R_n]})$, $\hat{\theta}_{\mathrm{ADI}}(s, u^{[R_n]}, u_{\mathrm{DP}}^{[R_n]})$ is an interior point in $\Theta$. For any $\theta_n \xrightarrow{\mathrm{P}} \theta^*$, we have $\frac{\partial m^{R_n}(\theta_n)}{\partial \theta} \xrightarrow{\mathrm{P}} B^*$ and $\frac{\partial (n S^{R_n}(\theta_n))}{\partial \theta} = O_p(1)$.
}

    \textbf{(A8)} and the conditions on $\{R_n\}_{n=1}^{\infty}$ are for the consistency of $\hat{\theta}_{\mathrm{ADI}}$.   The rate of $R_n$ does not appear in the results, so even a logarithmic rate such as $R_n=c\log(n)$ suffices. \textbf{(A9)} is for using the Taylor expansion to derive the asymptotic distribution and asymptotic equivariance property of $\hat{\theta}_{\mathrm{ADI}}$.

\begin{definition}\label{def:wellbehaved}
    The function $\psi:\mathbb{R}^p \rightarrow \mathbb{R}^q$ is a \emph{well-behaved consistent estimator} of $\theta$ if it is continuously differentiable at $\beta^*$ and $\psi(s) \xrightarrow{\mathrm{P}} \theta^*$.
\end{definition}

\begin{restatable}[ADI asymptotics]{theorem}{adaindirectestconsistent}\label{thm:ada_indirect_est}
    For $\hat{\theta}_{\mathrm{ADI}}(s, u^{[R_n]}, u_{\mathrm{DP}}^{[R_n]})$ defined in \eqref{eq:theta_adi}, where $\lim\limits_{n\rightarrow\infty} R_n = \infty$, we have the following:
    \begin{enumerate}
        \item Under \textbf{(A1--A3, A6--A8)}, $\hat{\theta}_{\mathrm{ADI}}$ is a consistent estimator of $\theta^*$.
        \item Under \textbf{(A1--A3, A5--A9)}, the asymptotic distribution of $\sqrt{n}(\hat{\theta}_{\mathrm{ADI}} - \theta^*)$ is the same as the distribution of $\left((B^*)^{\intercal} \Omega^* B^*\right)^{-1} (B^*)^{\intercal} \Omega^* (J^*)^{-1} v$ where $v \sim F_{\rho,u,\mathrm{DP}}^*$ and $\Omega^*=\Sigma(\theta^*)^{-1}=\Var[(J^*)^{-1} v]^{-1}$.
        \item Under \textbf{(A1--A3, A5--A9)}, $\hat{\theta}_{\mathrm{ADI}}$ and $s$ are asymptotically equivariant.
        \item Under \textbf{(A1--A3, A5--A9)},
        for all well-behaved consistent estimators $\psi(s)$, we have $\Var\l(\lim\limits_{n\rightarrow\infty}\sqrt{n}(\psi(s) - \theta^*)\r) \succeq \Var\l(\lim\limits_{n\rightarrow\infty}\sqrt{n}(\hat{\theta}_{\mathrm{ADI}} - \theta^*)\r) =\l((B^*)^{\intercal} \Sigma(\theta^*)^{-1} B^*\r)^{-1}$. For any choice of $\{\Omega_n\}_{n=1}^\infty$ in $\hat{\theta}_{\mathrm{IND}}$, we have\\
        $\Var\l(\lim\limits_{n\rightarrow\infty}\sqrt{n}(\hat{\theta}_{\mathrm{IND}} - \theta^*)\r) \succeq \Var\l(\lim\limits_{n\rightarrow\infty}\sqrt{n}(\hat{\theta}_{\mathrm{ADI}} - \theta^*)\r)$. \\
    \end{enumerate}
\end{restatable}

    
    By Theorem \ref{thm:ada_indirect_est}, the adaptive indirect estimator achieves the minimum asymptotic variance, $\Lambda:=\l((B^*)^{\intercal} \Sigma(\theta^*)^{-1} B^*\r)^{-1}$, compared to all indirect estimators and all well-behaved consistent estimators based on $s$.  
    Note that $\Lambda$ depends on the DP mechanism that generates $s$, and our adaptive indirect estimator is also optimal in terms of ``invertible'' post-processing on $s$: For $\phi:\mathbb{R}^p \rightarrow \mathbb{R}^p$ being continuously differentiable at $\beta^*$ and $\frac{\partial \phi(\beta^*)}{\partial s}$ being an invertible matrix, using the delta method, the adaptive indirect estimator based on $\phi(s)$ has the same asymptotic variance as $\Lambda$.
    Furthermore, we provide some interpretations of $\Lambda$ and leave it for future work to find the optimal DP mechanism that minimizes $\Lambda$.
    \begin{enumerate}
        \item If $\theta\in\mathbb{R}^1$ and $s\in\mathbb{R}^1$, $\l((B^*)^{\intercal} \Sigma(\theta^*)^{-1} B^*\r)$ is the efficacy \citep{pitman2018some} of $s$ at $\theta^*$, related to the asymptotic relative efficiency (Pitman efficiency);
        \item In general, we can interpret $B^*=\frac{\partial b(\theta^*)}{\partial \theta}$ as the signal to identify $\theta^*$ by $s$ and $\Sigma(\theta^*)$ as the noise in $s$. Then $\Lambda$ is related to the idea of the signal-to-noise ratio.
    \end{enumerate}    

\begin{remark}
    For Theorem \ref{thm:indirect_est} part 2, the variance of $(v_0 - \frac{1}{R}\sum_{r=1}^R v_r)$ is in the order of $(1+1/R)$ compared to the variance of $v_0$. To approach the optimal variance, we need to set $\Omega=\l((B^*)^{\intercal} \Sigma(\theta^*)^{-1} B^*\r)^{-1}$ and take $R$ to be an arbitrarily large constant. In Theorem \ref{thm:ada_indirect_est} part 2, this is automatically achieved without the knowledge of $\Omega$ and holds with $R_n\rightarrow\infty$. 
\end{remark}

\begin{remark}
    We use the asymptotic equivariance in Theorems \ref{thm:indirect_est} and \ref{thm:ada_indirect_est} to establish the consistency of the PB confidence sets and HTs with $\hat{\theta}:=\hat{\theta}_{\mathrm{IND}}$. While the asymptotic distribution could be used for inference as stated in the original indirect inference method \citep{gourieroux1993indirect}, it would require an estimation of $(B^*, J^*, F_{\rho,u,\mathrm{DP}}^*)$, and its finite-sample performance may be unsatisfactory. However, the asymptotic distribution is still helpful for studentization in constructing an approximate pivot for inference as shown in Theorem \ref{thm:indirect_inf_valid_est} part 3. Figure \ref{fig:HT_comparison} shows the  performance of the approximate pivot in HTs.
\end{remark}

\begin{remark}

    The first two parts of Theorem \ref{thm:indirect_est}
    are inspired by \citet[Propositions 1 and 3]{gourieroux1993indirect} while we give a more detailed and precise proof and focus on its application with DP. We generalize the asymptotic distribution of $\sqrt{n} (\frac{\partial \rho_n(\beta^*; D, u_{\mathrm{DP}})}{\partial \beta})$ from the normal distribution \citep{gourieroux1993indirect} to $F_{\rho,u,\mathrm{DP}}^*$, which ensures that the result holds in broader settings. 
    In Example \ref{ex:general_f}, we show the  value of such a generalization in DP settings with a  strong privacy guarantee.
\end{remark}

\begin{example}\label{ex:general_f}

    For $D=(x_1,\ldots,x_n)$ where $x_i\in[0,1]$, in order to estimate the population mean under $\ep$-DP, we use the Laplace mechanism which releases $s := \frac{1}{n}\sum_{i=1}^n x_i + u_{\mathrm{DP}}$ where $u_{\mathrm{DP}} \sim \mathrm{Laplace}(1/(n\ep))$. Let $\rho_n(\beta) := \frac{1}{2}\|\beta - s\|_2^2$, which indicates that $\rho_\infty(\beta) = \frac{1}{2}\|\beta - \beta^*\|_2^2$ where $\beta^* = \mathbb{E}[x_i]$. We have $\sqrt{n} (\frac{\partial \rho_n(\beta^*)}{\partial \beta})=\sqrt{n}(\beta^* - s)$.  If $x_i\iid \mathrm{Uniform}([0,1])$ and $\ep=\omega(1/\sqrt{n})$, then $F_{\rho,u,\mathrm{DP}}^*=N(0,1/12)$. However, if $\ep=c/\sqrt{n}$,   then $F_{\rho,u,\mathrm{DP}}^*$ is not normal but a convolution of $N(0,1/12)$ and $\mathrm{Laplace}(c)$. This example illustrates that in some DP settings, the asymptotic distribution may not be normal.  Such a decreasing $\ep$ could arise by the composition property of DP: with larger and richer datasets, it is common that more analyses will be computed -- under DP each analysis must have a smaller privacy loss budget in order for the total analysis to be $\ep$-DP.
\end{example}

By Proposition \ref{prop:parambootconsistency} and the asymptotic equivariance of $\hat{\theta}_{\mathrm{IND}}$ and $\hat{\theta}_{\mathrm{ADI}}$, if $\tau$ and $\hat\tau$ satisfy regularity conditions, using the delta method, we know that PB with $\hat{\theta}:=\hat{\theta}_{\mathrm{IND}}$ or $\hat{\theta}_{\mathrm{ADI}}$
is consistent. We formalize this intuition in Theorem \ref{thm:indirect_inf_valid_est} and Corollary \ref{cor:indirect_inf_valid_CI_HT}.

\begin{restatable}[PB asymptotics]{theorem}{indirectinfvalid}\label{thm:indirect_inf_valid_est}
     Let $\eta_1:\mathbb{R}^q \rightarrow \mathbb{R}^d$ and $\eta_2:\mathbb{R}^p \rightarrow \mathbb{R}^d$ be two maps defined and continuously differentiable in a neighborhood of $\theta^*$ and $s$, respectively.
    \begin{enumerate}
        \item Under \textbf{(A1--A7)}, when $n\rightarrow\infty$ and $\hat\Sigma(s_b):=\frac{1}{n}I_d$, with $\hat\theta:=\hat{\theta}_{\mathrm{IND}}$, the two types of PB estimators $\hat\tau_1(s_b):=\eta_1(\hat\theta(s_b))$ and $\hat\tau_2(s_b):=\eta_2(s_b)$ are consistent.
        \item  Under \textbf{(A1--A3, A5--A9)}, when $n\rightarrow\infty$, $R\rightarrow\infty$, and $\hat\Sigma(s_b):=\frac{1}{n}I_d$, with $\hat\theta:=\hat{\theta}_{\mathrm{ADI}}$, the two types of PB estimators $\hat\tau_1(s_b):=\eta_1(\hat{\theta}(s_b))$ and $\hat\tau_2(s_b):=\eta_2(s_b)$ are consistent.
        \item  Under \textbf{(A1--A3, A5--A9)}, when $n\rightarrow\infty$, $R\rightarrow\infty$, we let $\hat\theta:=\hat{\theta}_{\mathrm{ADI}}$, $\hat\theta_b:=\hat{\theta}(s_b)$, and 
        $\hat\Sigma(s_b):=\frac{\eta_3(\hat\theta_b)}{n}$, where $\eta_3(\hat\theta_b):=
            \l(\frac{\partial \eta_1}{\partial \theta}(\hat\theta_b)
                \l(
                    \l(\frac{\partial b}{\partial\theta}(\hat\theta_b)\r)^{\intercal} 
                    \Sigma(\hat\theta_b)^{-1} 
                    \frac{\partial b}{\partial\theta}(\hat\theta_b)
                \r)^{-1}
                \l(\frac{\partial \eta_1}{\partial \theta}(\hat\theta_b)\r)^\intercal
            \r)$,
        that is, $\hat\Sigma(s_b)$ estimates the asymptotic covariance of $\hat\tau(s_b):=\eta_1(\hat{\theta}(s_b))$.
         The PB estimators $\hat\tau(s_b)$ are consistent 
        and $T:=\|\hat\Sigma(s)^{-\frac{1}{2}}\l(\hat\tau(s) - \tau(\theta^*)\r)\|_p$ is an asymptotic pivot.
    \end{enumerate}
\end{restatable}

    

\begin{restatable}[]{corollary}{indirectinfvalidcor}\label{cor:indirect_inf_valid_CI_HT}
    Under the assumptions in Theorem \ref{thm:indirect_inf_valid_est}, the corresponding PB confidence sets are asymptotically consistent. If we further assume that  $\inf_{\theta^*\in\Theta_0}\|\tau(\theta^*) - \tau(\theta_1)\|_p > 0$ for all $\theta_1\in\Theta_1$, the corresponding PB HTs are also asymptotically consistent.
\end{restatable}

\begin{remark}
     If there exists  $\theta$ such that $s-\frac{1}{R} \sum_{r=1}^R s^r(\theta) = 0$, our asymptotic distribution is consistent with the results by \citet[Theorem 4]{zhang2022flexible} for the JINI estimator. 
    Our asymptotic distribution aligns with the optimality results by \citet{jiang2004indirect} for the adjusted estimator, while our estimator is data-driven without requirement of the knowledge of  $\Sigma(\theta)$ or  access to an estimator of this matrix. 
     For statistical inference,   \citet{jiang2004indirect}, \citet{guerrier2019simulation}, and \citet{zhang2022flexible} used the asymptotic normality of their estimators, while we use PB to build confidence sets and conduct HTs. 
\end{remark}

\subsection{Non-smooth settings and surrogate models}\label{sec:pb:non-smooth}
In Assumption \textbf{(A5)}, we require $s^r(\theta)$ to be  differentiable with respect to $\theta$ for all $(\theta, u^r, u_{\mathrm{DP}}^r)$. However, there are two important problem settings that do not satisfy this assumption:

The first type of non-smoothness is from the clamping procedure $x' := \min(\max(x,a),b)$ in DP mechanisms for bounded sensitivity.
To satisfy the assumptions on the differentiability of $x'$ when $x\in\mathbb{R}$, the original clamping procedure can be replaced with a differentiable clamping function $x':=f(x)$, $f:\mathbb{R}\rightarrow[a,b]$, e.g., a transformed sigmoid function, $x':=a+(b-a)\mathrm{sigmoid}(\frac{4}{b-a}(x-\frac{a+b}{2}))$ where $\mathrm{sigmoid}(x):=\frac{1}{1+\ee^{-x}}$, or a transformed smoothstep function, $x':=a+(b-a)\mathrm{smoothstep}(\frac{x-a}{b-a})$ where $\mathrm{smoothstep}(x):= (3x^2-2x^3)\mathbb{1}_{0\leq x\leq 1} + \mathbb{1}_{1 < x}$. Note that in our simulations, we still use the original clamping procedure, as it does not affect our usage of the \texttt{optim} function in \texttt{R} with the method \texttt{L-BFGS-B} for optimization, where the gradient is replaced by a finite-difference approximation.

The second type of non-smoothness is the discrete nature of data before using any DP mechanism. Discrete distributions such as a multinomial or Poisson distribution are common, but their corresponding data generating function is non-smooth with respect to the parameter, given the random seed, which brings difficulties in both the optimization procedure in our method and the applicability of our theoretical guarantees. 

In the remainder of this section, we extend our theory to an approximation method using a smoothed data generation,  which can be thought of as a surrogate model. For the data generating equation $D:=G(\theta^*, u)$ where $G$ may not be continuous and $u$ following a continuous distribution,  
we use a smooth function $G_{\mathrm{smooth}}$ to approximate $G$ in the indirect estimators:  let $D_{\mathrm{smooth}}^r(\theta)=G_{\mathrm{smooth}}(\theta, u^r)$, and re-define $s^r(\theta)$ in \eqref{eq:theta_ind} and \eqref{eq:theta_adi} as 
\begin{equation}\label{eq:s_smooth}
s^r(\theta) := \argmin\limits_{\beta\in \mathbb{B}} \rho_n(\beta; D_{\mathrm{smooth}}^r(\theta), u_{\mathrm{DP}}^r).
\end{equation}
To obtain the consistency and optimality results of the new $\hat{\theta}_{\mathrm{IND}}$ and $\hat{\theta}_{\mathrm{ADI}}$, we need the following assumptions in addition to the assumptions in Section \ref{sec:pb:asymp_theory}.
\\
\noindent \textbf{(A1s)} {$\sup_{\beta\in \mathbb{B}} |\rho_n(\beta; D_{\mathrm{smooth}}^r(\theta^*), u_{\mathrm{DP}}) - \rho_\infty(\beta; F_u, F_{\mathrm{DP}}, \theta^*)| \xrightarrow{\mathrm{P}} 0$.}
\\
\noindent \textbf{(A5s)} {For all convergent sequences $h_n\rightarrow h$, $s\l(\theta^* + \frac{h_n}{\sqrt{n}}\r) = s(\theta^*) + \frac{\partial b(\theta^*)}{\partial \theta}\frac{h_n}{\sqrt{n}}+o_p\l(\frac{1}{\sqrt{n}}\r)$.
}
\\
\noindent \textbf{(A6s)} {$\sqrt{n} (-\frac{\partial \rho_n(\beta^*; D_{\mathrm{smooth}}^r(\theta^*), u_{\mathrm{DP}})}{\partial \beta}) \xrightarrow{\mathrm{d}} F_{\rho,u,\mathrm{DP}}^*$.}
\\
\noindent \textbf{(A7s)} {$\frac{\partial^2 \rho_n(\beta^*; D_{\mathrm{smooth}}^r(\theta^*), u_{\mathrm{DP}})}{(\partial \beta)(\partial \beta^\intercal)}\xrightarrow{\mathrm{P}} J^*$.}

\begin{theorem}\label{thm:indirect_est_smooth} 
    For $s^r$ defined in \eqref{eq:s_smooth} and $\hat{\theta}_{\mathrm{IND}}$ defined in \eqref{eq:theta_ind}, we have the following:
\begin{enumerate}
    \item Under \textbf{(A1--A4)} and \textbf{(A1s)}, $\hat{\theta}_{\mathrm{IND}}$ is a consistent estimator of $\theta^*$.
    \item Under \textbf{(A1--A7)}  and \textbf{(A1s, A6s, A7s)}, the asymptotic distribution of $\sqrt{n}(\hat{\theta}_{\mathrm{IND}} - \theta^*)$ is the same as the one of $\left((B^*)^{\intercal} {\Omega} B^*\right)^{-1} (B^*)^{\intercal} {\Omega} (J^*)^{-1} (v_0 - \frac{1}{R}\sum_{r=1}^R v_r)$ where $v_r \iid F_{\rho,u,\mathrm{DP}}^*$ for $r=0,1,\ldots,R$.
    \item Under \textbf{(A1--A7)} and \textbf{(A1s, A5s, A6s, A7s)}, $\hat{\theta}_{\mathrm{IND}}$ and $s$ are asymptotically equivariant.
\end{enumerate}
\end{theorem}

The technical conditions \textbf{(A1s),(A5s)-(A7s)} guarantee using $D$ and $D_{\mathrm{smooth}}$ result in the same asymptotic properties.
\textbf{(A1s)} is for the consistency of $s^r(\theta^*)$. \textbf{(A6s, A7s)} are for the asymptotic distribution of $s^r(\theta^*)$. As $s$ is non-differentiable with respect to $\theta$, \textbf{(A5s)} is a finite-difference property of $s$ used to establish the asymptotic equivariance.

These additional assumptions align with the intuition that when using the surrogate model,  $s^r(\theta^*)$ should follow the same limiting distribution as $s$: \textbf{(A1, A1s, A2, A3)} are for matching mean, \textbf{(A6, A6s, A7, A7s)} are for matching variance, and \textbf{(A5, A5s)} are for matching the first-order property with respect to $\theta$. Note that we are not assuming that the true parameters of the underlying model are known; rather, the surrogate model should be designed so that at each parameter value, the mean, variance, and first-order behavior of the surrogate model match those of the true model.  As an example, when $D$ is generated by a multinomial distribution, we can use the multivariate normal distribution with the same mean and covariance matrix as $D$ to generate $D_{\mathrm{smooth}}$.

\begin{restatable}[]{theorem}{adaindirectestconsistent}\label{thm:ada_indirect_est_smooth}
    For $s^r$ defined in \eqref{eq:s_smooth} and $\hat{\theta}_{\mathrm{ADI}}(s, u^{[R_n]}, u_{\mathrm{DP}}^{[R_n]})$ defined in \eqref{eq:theta_adi}, where $\lim\limits_{n\rightarrow\infty} R_n = \infty$, we have the following:
    \begin{enumerate}
        \item Under \textbf{(A1--A3, A6--A8)} and \textbf{(A1s, A6s, A7s)}, $\hat{\theta}_{\mathrm{ADI}}$ is consistent for $\theta^*$.
        \item Under \textbf{(A1--A3, A5--A9)} and \textbf{(A1s, A6s, A7s)}, the asymptotic distribution of $\sqrt{n}(\hat{\theta}_{\mathrm{ADI}} - \theta^*)$ is the same as the distribution of $\left((B^*)^{\intercal} \Omega^* B^*\right)^{-1} (B^*)^{\intercal} \Omega^* (J^*)^{-1}v$ where $v \sim F_{\rho,u,\mathrm{DP}}^*$ and $\Omega^*=\Sigma(\theta^*)^{-1}=\Var[(J^*)^{-1} v]^{-1}$.
        \item Under \textbf{(A1--A3, A5--A9)} and \textbf{(A1s, A5s, A6s, A7s)}, $\hat{\theta}_{\mathrm{ADI}}$ and $s$ are asymptotically equivariant.
        \item Under \textbf{(A1--A3, A5--A9)} and \textbf{(A1s, A6s, A7s)}, 
        for all well-behaved consistent estimators $\psi(s)$, we have $\Var\l(\lim\limits_{n\rightarrow\infty}\sqrt{n}(\psi(s) - \theta^*)\r) \succeq \Var\l(\lim\limits_{n\rightarrow\infty}\sqrt{n}(\hat{\theta}_{\mathrm{ADI}} - \theta^*)\r) =\l((B^*)^{\intercal} \Sigma(\theta^*)^{-1} B^*\r)^{-1}$. For any choice of $\{\Omega_n\}_{n=1}^\infty$ in $\hat{\theta}_{\mathrm{IND}}$, we have \\
        $\Var\l(\lim\limits_{n\rightarrow\infty}\sqrt{n}(\hat{\theta}_{\mathrm{IND}} - \theta^*)\r) \succeq \Var\l(\lim\limits_{n\rightarrow\infty}\sqrt{n}(\hat{\theta}_{\mathrm{ADI}} - \theta^*)\r)$. \\
        
    \end{enumerate}
\end{restatable}

\begin{remark}
It is straightforward to adapt Theorem \ref{thm:indirect_inf_valid_est} to be based on $\hat{\theta}_{\mathrm{IND}}$ and $\hat{\theta}_{\mathrm{ADI}}$ using Theorems \ref{thm:indirect_est_smooth} and \ref{thm:ada_indirect_est_smooth}. Our PB procedure in Theorem \ref{thm:indirect_inf_valid_est} is still based on $s$ instead of $s^r$ in \eqref{eq:s_smooth} since the validity of PB only depends on $\hat\theta$ and $s$ being asymptotically equivariant.
\end{remark}

\section{Simulations}
\label{sec:pb:simulation}

In this section, we use simulation studies on DP statistical inference to demonstrate the performance of our methodology. We construct CIs for the population mean and variance of normal distributions,  conduct HTs with a linear regression model, and form CIs for a logistic regression model.  An additional experiment on a na\"ive Bayes log-linear model is included in Section \ref{sec:naive} in the Supplementary Materials. 
All results are computed over 1000 replicates. Code to replicate the experiments is available at \url{https://github.com/JordanAwan/debiased_dp_parametric_bootstrap/}.

\subsection{Location-scale normal}
\label{sec:pb:CI}
We construct DP CIs for the population mean and standard deviation of a normal distribution. 
Let $D:=(x_1,\ldots,x_n)$, $x_i := \mu^* + \sigma^* u_i$, $s:=(\widetilde{m}(D), \widetilde{\eta^2}(D))$, where
\[
\begin{aligned}
    &\widetilde{m}(D):= m(D) + \frac{U-L}{n\ep}u_{\mathrm{DP},1}, && \widetilde{\eta^2}(D) := \eta^2(D) + \frac{(U-L)^2}{n\ep}u_{\mathrm{DP},2},\\
    &m(D):= \frac{1}{n}\sum_{i=1}^n [x_i]_L^U, && \eta^2(D) := \frac{1}{n-1}\sum_{i=1}^n ([x_i]_L^U - m(D))^2,
\end{aligned}\]
$u:=(u_1,\ldots,u_n)\sim F_u := N(0,I_{n\times n})$ and $u_{\mathrm{DP}}:=(u_{\mathrm{DP},1}, u_{\mathrm{DP},2}) \sim F_{\mathrm{DP}} := N(0,I_{2\times 2})$.
The summary $s$ satisfies $(\sqrt{2}\ep)$-GDP, since both $\widetilde m(D)$ and $\widetilde{\eta^2}(D)$ each satisfy $\ep$-GDP, since the sensitivity of $m(D)$ is $(U-L)/n$ and the sensitivity of $\eta^2(D)$ is $(U-L)^2/n$ \citep[Theorem 26]{du2020differentially}. We will construct CIs based on $s$. We set $(\mu^*,\sigma^*;n,\ep;L,U):=(1,1;100,1;0,3)$ and   use 200 bootstrap samples and 50 synthetic indirect estimate samples.

In Table \ref{tab:pb:loc_scale_normal_multivariate}, we compare our debiased estimator, the adaptive indirect estimator, with other debiased inference methods used in PB to construct private CIs with level $(1-\alpha)=0.95$. We let $\theta^*:=(\mu^*,\sigma^*) \in \Theta \subset \mathbb{R}^2$.
For our method, we let
$\hat\tau(s):=(\hat\theta_{\mathrm{ADI}}(s))^{(i)}$, $\tau(\theta):=\theta^{(i)}$, and $\hat\sigma(s):=\frac{1}{\sqrt{n}}$, where $i=1$ and $i=2$ are for $\mu^*$ and $\sigma^*$, respectively. The other PB methods use the naive estimator $\hat\theta:=\l(\widetilde{m}(s), \sqrt{\max(0, \widetilde{\eta^2}(s))}\r)$, which is biased as shown in Figure \ref{fig:mean-std-bias}, to generate bootstrap samples. Table \ref{tab:pb:loc_scale_normal_multivariate} shows that our results have satisfactory coverage while the biased $\hat\theta$ results in under-coverage for other methods:
\begin{itemize}\setlength\itemsep{-0.25em}
    \item The naive percentile method uses the $\alpha/2$ and $(1-\alpha/2)$ percentiles of $\{\hat\theta(s_b)\}_{b=1}^B$. 
    \item The simplified $t$ method uses the $\alpha/2$ and $(1-\alpha/2)$ percentile of $\{2\hat\theta(s) - \hat\theta(s_b)\}_{b=1}^B$.
    \item \citet{ferrando2022parametric} estimates the bias of $\hat\theta(s)$ using $\l(\frac{1}{B}\sum_{b=1}^B \hat\theta(s_b) - \hat\theta(s)\r)$ and build CIs with the $\frac{\alpha}{2}$ and $1-\frac{\alpha}{2}$ percentiles of $\l\{\hat\theta(s_b) - \l(\frac{1}{B}\sum_{b=1}^B \hat\theta(s_b) - \hat\theta(s)\r)\r\}_{b=1}^B$.
    \item Efron's bias-corrected accelerated (BCa) percentile method \citep{efron1987better} is based on the idea that for the biased $\hat\theta$, there exists a transformation function $g$ such that $(\hat\phi - \phi) / \sigma_\phi + z_0 \sim N(0, 1)$ where $\hat\phi:=g(\hat\theta)$, $\phi:=g(\theta^*)$, $\sigma_\phi:=1+a\phi$, $z_0$ is a constant bias, $a$ is an acceleration constant for better performance. As the estimation of $a$ requires access to the original sensitive dataset, which is prohibited for DP guarantees, we set $a:=0$ to reduce BCa to a bias-corrected (BC) only method.
    \item The automatic percentile method \citep{diciccio1989automatic} follows the idea of BCa, but replaces $N(0,1)$ with a general distribution and uses bootstrap simulations to avoid estimating $a$ or $z_0$. 
\end{itemize}

\noindent We also compare our method with Repro \citep{awan2025simulation}, which does not use PB but another simulation-based method for inference.
Repro gives more conservative results, i.e., higher coverage and larger width, than ours since it is designed for finite-sample (non-asymptotic) inference, and offers simultaneous coverage. Furthermore, the test statistic of Mahalanobis depth, recommended by \citet{awan2025simulation}, may be suboptimal.

Figure \ref{fig:mean-std-bias} shows the (empirical) sampling distributions of the estimators where we compare their biases. The naive method is $\hat\theta(s)$, the simplified $t$ method is $\l(2\hat\theta(s) - \frac{1}{B}\sum_{b=1}^B \hat\theta(s_b)\r)$ which is also used by \citet{ferrando2022parametric} as the debiased point estimator of $\theta$  via parametric bootstrap, and the indirect estimator is $\hat\theta_{\mathrm{ADI}}(s)$.  Due to clamping, the naive and simplified $t$ estimators are biased above for the mean parameter and below for the scale parameter. The adaptive indirect estimator is the only one that does not have a significant bias.
Note that Efron's BC, automatic percentile, and Repro do not provide debiased estimators.

In the bottom of Table \ref{tab:pb:loc_scale_normal_multivariate}, we compare confidence ellipses via the PB ADI estimator against the Repro confidence regions of \citet{awan2025simulation}. We see that our method has coverage closer to the nominal level and smaller average area compared to Repro, even in this multivariate setting.

\begin{table}[t]
    \centering
    \caption{ Results of the 95\% CIs and confidence regions for parameters of $N(\mu,\sigma^2)$. Coverage and width/area are reported with standard errors in parentheses.}
    \label{tab:pb:loc_scale_normal_multivariate}
    \vspace*{6pt}
    \small
    \begin{tabular}{l cccc}
    \toprule
    & \multicolumn{2}{c}{\bf Coverage} & \multicolumn{2}{c}{\bf Average Width} \\ 
    \cmidrule(lr){2-3} \cmidrule(lr){4-5}
    \textbf{Method} & $\mu$ & $\sigma$ & $\mu$ & $\sigma$ \\ 
    \midrule
    \multicolumn{5}{l}{\textit{Univariate Inference}} \\
    PB \underline{(adaptive indirect)} & 0.959 (0.006) & 0.951 (0.007) & 0.463 (0.003) & 0.580 (0.003) \\
    PB (naive percentile)           & 0.697 (0.015) & 0.006 (0.002) & 0.311 (0.001) & 0.293 (0.001) \\
    PB (simplified $t$)              & 0.869 (0.011) & 0.817 (0.012) & 0.311 (0.001) & 0.293 (0.001) \\
    PB \citep{ferrando2022parametric} & 0.808 (0.012) & 0.371 (0.015) & 0.311 (0.001) & 0.293 (0.001) \\
    PB (Efron's BC)                  & 0.854 (0.011) & 0.042 (0.006) & 0.298 (0.001) & 0.139 (0.002) \\
    PB (automatic percentile)        & 0.865 (0.011) & 0.126 (0.010) & 0.314 (0.001) & 0.261 (0.001) \\
    Repro \citep{awan2025simulation}  & 0.989 (0.003) & 0.998 (0.001) & 0.599 (0.003) & 0.758 (0.005) \\
    
    \midrule
    \multicolumn{5}{l}{\textit{Multivariate Inference}} \\
    & \multicolumn{2}{c}{\bf Coverage $(\mu, \sigma)$} & \multicolumn{2}{c}{\bf Average Area} \\
    \cmidrule(lr){2-3} \cmidrule(lr){4-5}
    PB \underline{(adaptive indirect)} & \multicolumn{2}{c}{0.943 (0.007)} & \multicolumn{2}{c}{0.339 (0.004)} \\
    Repro (Mahalanobis) & \multicolumn{2}{c}{0.971 (0.005)} & \multicolumn{2}{c}{0.3619 (0.004)} \\
    \bottomrule
    \end{tabular}
\end{table}


\begin{figure}[t]
    \centering
    \includegraphics[width=\linewidth]{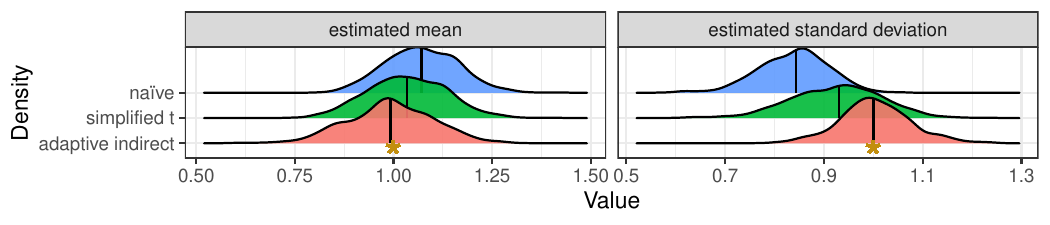}
    \caption{Comparison of the sampling distributions of different estimates. The vertical line denotes the median of each distribution, and the `$\ast$' sign denotes the true parameter value. The debiased estimator by \citet{ferrando2022parametric} is the same as the simplified $t$ estimator.}\label{fig:mean-std-bias}
\end{figure}

\subsection{Simple linear regression hypothesis testing}\label{sec:pb:HT}
We conduct private HTs for the slope coefficient in a simple linear regression model. Following 
\citet{alabi2022hypothesis}, 
for a dataset $D=\l((x_1,y_1),(x_2,y_2),\ldots,(x_n,y_n)\r)$, we assume the linear regression model $Y=X\beta_1^* + \beta_0^* + \epsilon$ where $X\sim N\l(\mu_x, \sigma_x^2\r)$ and $\epsilon\sim N\l(0, \sigma_e^2\r)$, and we will  test $H_0:\beta_1^*=0$ and $H_a:\beta_1^*\neq 0$. The parameters are $\theta := \l(\beta_1^*, \beta_0^*, \mu_x, \sigma_x, \sigma_e\r)$.  We set $\beta_0^*=-0.1$, $\mu_x = 0.5, \sigma_x = 1, \sigma_e = 1$ and we vary the true $\beta_1^*$ and sample size $n$. We use 200 bootstrap samples and 50 synthetic indirect estimate samples.     The privatized statistic $s:=\l(\tilde{x}, \tilde{y}, \widetilde{x^2}, \widetilde{y^2}, \widetilde{xy}\r)$, as defined below, satisfies $\mu$-GDP:
    \[
    \begin{aligned}[t]     
        \tilde{x} &:=\frac{1}{n} \sum_{i=1}^n [x_i ]_{-\Delta}^{\Delta}+\frac{2\Delta}{ (\mu/\sqrt{5}) n}u_1, \quad\quad
        \widetilde{x^2} :=\frac{1}{n}\sum_{i=1}^n {\l[x^2_i \r]_{0}^{\Delta^2}}+\frac{\Delta^2}{ (\mu/\sqrt{5}) n}u_2, \\
        \tilde{y} &:=\frac{1}{n}\sum_{i=1}^n { [y_i ]_{-\Delta}^{\Delta}}+\frac{2\Delta}{ (\mu/\sqrt{5}) n}u_3, \quad\quad
        \widetilde{xy} :=\frac{1}{n}\sum_{i=1}^n {[x_i y_i]_{-\Delta^2}^{\Delta^2}}+\frac{2\Delta^2}{ (\mu/\sqrt{5}) n}u_4, \\
        \widetilde{y^2} &:=\frac{1}{n}\sum_{i=1}^n{\l[y_i^2\r]_{0}^{\Delta^2}}+\frac{\Delta^2}{ (\mu/\sqrt{5}) n}u_5, \quad\quad \text{where}\quad u_i\iid N(0,1).
    \end{aligned}
    \]

    
    \citet{alabi2022hypothesis} used
    \[\hat\theta_0 := \l(0, \tilde y, \tilde x, \sqrt{\frac{n}{n-1}\max\l(0,\widetilde{x^2} - (\tilde x)^2\r)}, \sqrt{\frac{n}{n-1}\max\l(0,\widetilde{y^2} - (\tilde y)^2\r)}\r)\]
    to generate synthetic data $s_b$ and calculated the $F$-statistic, $F(s):=\frac{\l(\tilde \beta_1\r)^2 \l(n\l(\widetilde{x^2} - (\tilde x)^2\r)\r)}{\widetilde{S^2}}$ where  $\tilde\beta_1:=\frac{\widetilde{xy} - \tilde x \tilde y}{\widetilde{x^2} - (\tilde x)^2}$, $\tilde\beta_0:=\frac{\tilde y \cdot \widetilde{x^2} - \tilde x \cdot \widetilde{xy}}{\widetilde{x^2} - (\tilde x)^2}$, 
        $\widetilde{S^2}:=\frac{n\l(\widetilde{y^2} + \l(\tilde\beta_1\r)^2\widetilde{x^2} + \l(\tilde\beta_0\r)^2 - 2\tilde\beta_1 \widetilde{xy} - 2\tilde\beta_0\tilde y + 2\tilde\beta_1\tilde\beta_0\tilde x\r)}{n-2}$.
    The null hypothesis is rejected when $\widetilde{x^2} \geq (\tilde x)^2$, $\widetilde{y^2} \geq (\tilde y)^2$, $\widetilde{S^2} \geq 0$, and $F(s)$ is greater than the $(1-\alpha)$-quantile of $\{F(s_b)\}_{b=1}^B$.
    Figure \ref{fig:clamp_issues} has shown that the type I error of this test procedure is miscalibrated.



To compare our method to \citet{alabi2022hypothesis}, we set the clamping parameter $\Delta:=2$ and show the probabilities of rejecting $H_0$ in DP HTs with $1$-GDP guarantee. In Figure \ref{fig:HT_comparison},
the upper left subfigure are from the method by \citet{alabi2022hypothesis}, and there are two problems: 1) the type I error (when $\beta_1^*=0$), is larger than the significance level 0.05 when $n$ is large, and 2) the rejection probability is not always larger for larger $\beta^*_1$.
The first problem is caused by the bias in the naive estimator as in Figure \ref{fig:clamp_issues}, and the second problem may be due to the inefficiency of the private $F$-statistic.

\begin{figure}[t]
    \centering
    \includegraphics[width=\linewidth, trim={0 0 0 0},clip]{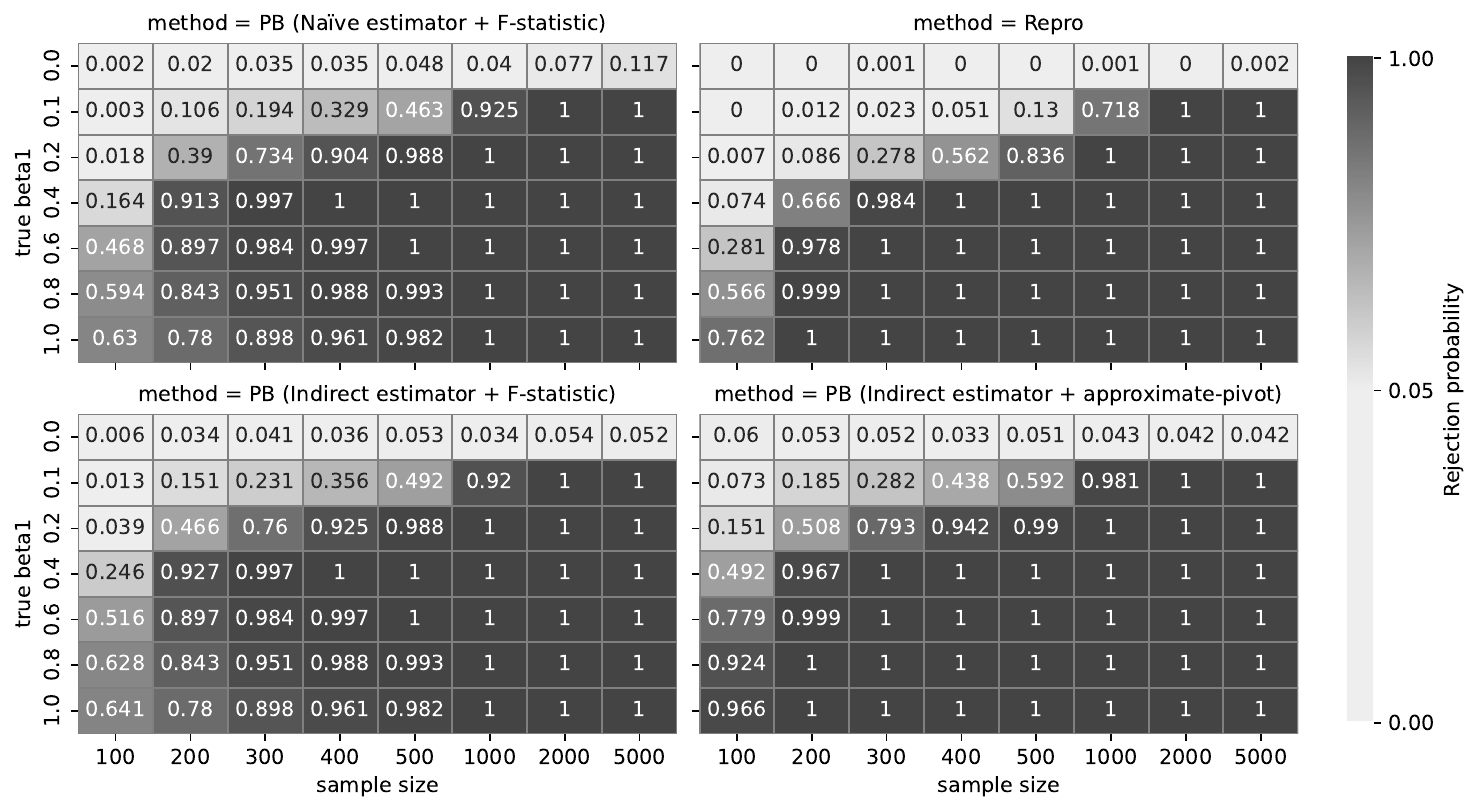}
    \caption{Comparison of the rejection probabilities on $H_0: \beta^*_1 = 0$ and $H_1: \beta^*_1 \neq 0$ for $Y = \beta^*_0 + X \beta^*_1 + \epsilon$ (significance level 0.05.)}\label{fig:HT_comparison}
    \label{fig:LR}
\end{figure}

For the results in the lower left subfigure, we replace the naive estimator with our adaptive indirect estimator, which solves the bias issue and controls the type I error. Although the power improves when using our estimator, the second problem mentioned above still remains.

As the $F$-statistic is equivalent to the $t$-statistic of $\beta_1^*$ in the non-private version of this linear model HT setting, we use the adaptive indirect estimator and its asymptotic distribution in Theorem \ref{thm:ada_indirect_est} to build a private $t$-statistic, which we also call the approximate pivot.
This approach follows Algorithm \ref{alg:indHT} with $\hat\tau(s):=(\hat\theta_{\mathrm{ADI}}(s))^{(1)}$, $\tau(\theta):=0$, and $\hat\sigma(s):=\sqrt{\reallywidehat{\Var}(\hat\tau(s))}$, where the computation of $\hat\sigma(s)$ is explained in Remark \ref{rmk:approx_pivot}. Additional details are found in Section \ref{app:linear_regresion_HT} in the appendix.
Note that in Algorithm \ref{alg:indHT}, we estimate $\theta$ in the whole parameter space rather than only under the null hypothesis, as discussed in Remark \ref{rmk:HT_estimate}, which is different from the DP Monte Carlo tests of \citet{alabi2022hypothesis}. 

From the lower-right subfigure of Figure \ref{fig:HT_comparison}, we can see that by using PB with our adaptive indirect estimator and the approximate pivot test statistic, the power improves further, and the second problem is solved, i.e., the power is higher when $\beta_1^*$ is larger. 

We also compare our method with Repro \citep{awan2025simulation} in the upper right subfigure of Figure \ref{fig:HT_comparison}, which guarantees the control of type I error in finite samples. Our method has better power than Repro while maintaining type I error control.  The reason that Repro has suboptimal performance is primarily due to two factors: 1) the Mahalanobis test statistic is likely sub-optimal and 2) repro results in a multivariate confidence set, which has simultaneous coverage, causing each confidence interval to have over-coverage.

\subsection{Logistic regression via objective perturbation}\label{s:logistic}
We adopt the experimental setup by \citet{awan2025simulation} to construct confidence intervals for a logistic regression model using the objective perturbation mechanism \citep{chaudhuri2011differentially}. We use the formulation by \citet{awan2021structure}, which adds noise to the gradient of the log-likelihood before optimizing, resulting in a non-additive noise. We assume that the predictor variables also need to be protected, so we privatize their first two moments using the optimal K-norm mechanism \citep{hardt2010geometry, awan2021structure}. See Section \ref{sec:logisticApp} in the Supplementary Materials for mechanism details.

We assume the predictor variable $x$ is naturally bounded between $[-1,1]$ and model $x$ in terms of the beta distribution: $x_i = 2z_i -1$ where $z_i
\overset{\mathrm{iid}}{\sim} \mathrm{Beta}(a^{\ast}, b^{\ast})$. A logistic regression models $y_i|x_i \sim \text{Bern}(\text{expit}(\beta^{\ast}_0+\beta^{\ast}_1x_i))$ where $\text{expit}(x) = \text{exp}(x)/(1+\text{exp}(x))$. The parameter is $\theta^{\ast} = (\beta^*_0, \beta^*_1, a^*, b^*)$, where we are interested in producing a confidence interval for $\beta^*_1$ and the other parameters are viewed as nuisance parameters. We set $a^{\ast} = b^{\ast} = 1, \beta^{\ast}_0 = 0.5, \beta^{\ast}_1 = 2, R = 200, \alpha = 0.05, n = 100, 200, 500, 1000, 2000$ and $\epsilon = 0.1, 0.3, 1, 3, 10$ in $\epsilon$-DP.  Note that since $a^*=b^*=1$, the true model is simply $x_i\iid U(-1,1)$. We use 200 parametric bootstrap samples and 200 synthetic indirect estimate samples. 

This setting has two unique challenges: 1) Since the $y$ value is discrete, we need to use a surrogate model in our indirect inference procedure. We use the normal distribution $N(p,p(1-p))$ to approximate $\mathrm{Bern}(p)$. 2) While clamping is not used in this mechanism, objective perturbation requires $\ell_2$ regularization, which causes the initial DP estimates of $\beta$ to be biased towards zero. 

\begin{figure}[t]
    \centering
    \includegraphics[width=1.0\linewidth, trim={0 0 0 0},clip]{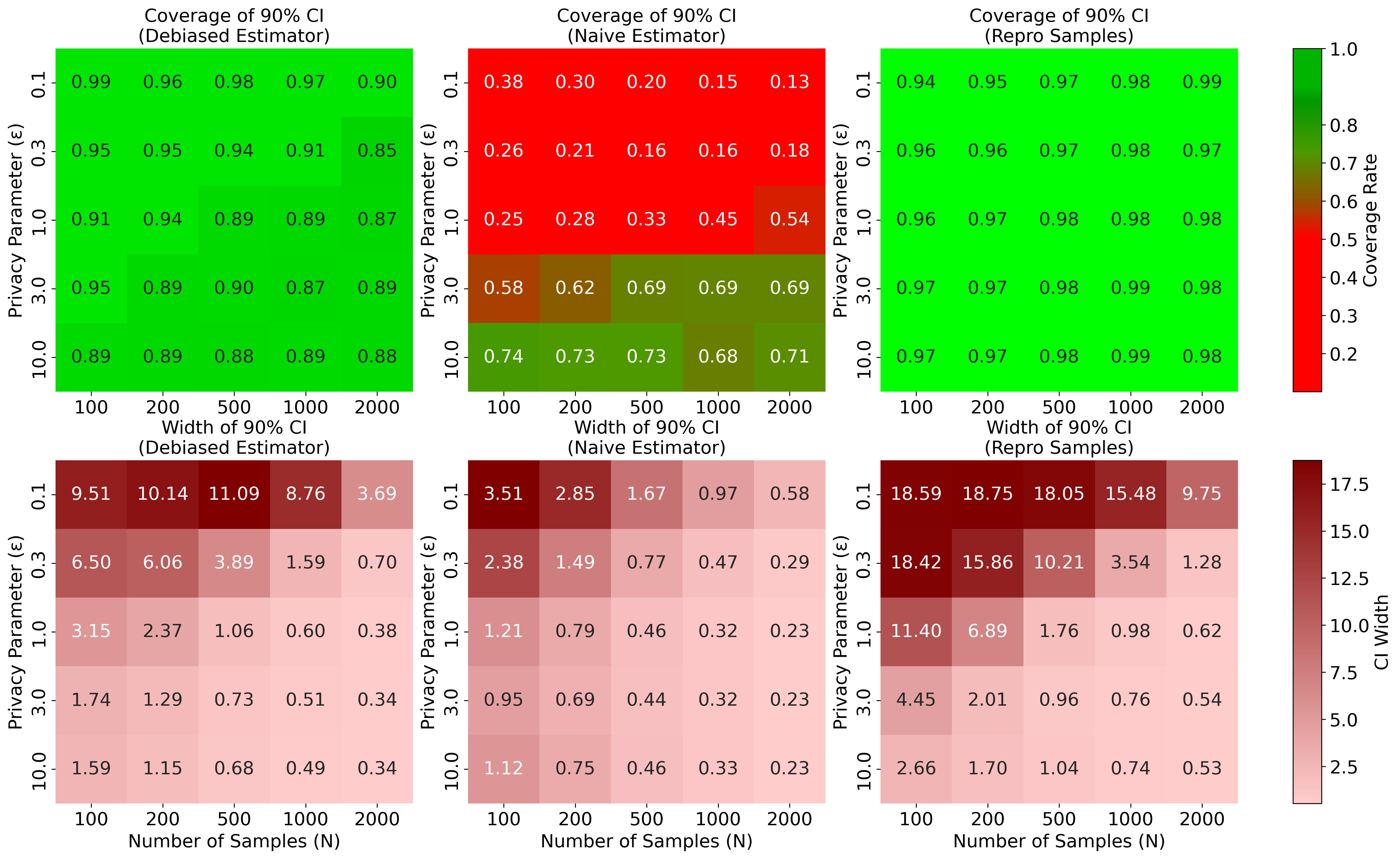}
    \caption{Comparison of the coverage and width of confidence intervals for the $\beta^*_1$ parameter in logistic regression at 90\% confidence using our proposed estimator with PB, naive plug-in estimator with PB, and Repro method.}\label{fig:logreg_comparison}
\end{figure}

In Figure 6, we compare the performance of our CIs against the Repro method from \citet{awan2025simulation} as well as a naive method which uses the DP output as the plug-in estimator in PB. 
Our method greatly improves the coverage of CIs compared to the naive method, and also achieves significantly smaller widths than the Repro method, even in cases where the coverage is similar. For smaller $\epsilon$ and $n$, our method tends to be conservative, which is preferred to undercovering. 

\section{Discussion}
\label{sec:pb:discussion}
Our novel simulation-based debiased estimator, the adaptive indirect estimator, corrects the clamping bias in DP mechanisms in a flexible and general way, which when combined with the parametric bootstrap results in consistent and powerful private statistical inference. 

While not unique to our work, one of the main limitations of the parametric bootstrap is the need for a  parametric model.  Due to the challenges of statistical inference on privatized data, parametric assumptions are common in the literature (e.g., \citealp{karwa2018finite,bernstein2018differentially,ju2022data,chen2025particle,awan2025simulation}). In Section \ref{sec:pb:non-smooth}, we show that surrogate models can be used to obtain valid inferences, illustrating some robustness to model mis-specification. Alternatively, \citet{ferrando2022parametric} combined asymptotic approximations with the parametric model to circumvent this issue in a linear regression setting. Further research could develop other methods to quantify the sensitivity to model mis-specification and/or allow for more flexible models. 


Another concern is the computational cost of optimization with respect to $\theta$ in the indirect estimator when $\theta$ is high-dimensional. Fortunately, our procedure scales well with the other parameters, as it is linear in both $R$ and $B$. If the summary $s$ can be computed in linear time, in terms of the sample size, then the procedure is linear in $n$ as well. We can roughly express the computational complexity of the ADI estimator as $O(R\cdot B\cdot S(n)\cdot p^{2.373}\cdot \mathrm{Opt}(d))$, where $S(n)$ is the time to compute $s$ on a sample of size $n$, $p^{2.373}$ is the fastest algorithm (to our knowledge) to invert a $p\times p$ matrix, which is needed for ADI \citep{alman2021refined}, and $\mathrm{Opt}(d)$ represents the time to solve the indirect inference optimization problem with a $d$-dimensional parameter, which depends on the setting.  A direction of future work is to explore gradient-based optimization tools with automatic differentiation to reduce the computational cost in the optimization step. Another direction is to identify the important dimensions of the parameter so that we can use a combination of an indirect estimator for these important dimensions and other debiased estimators for the other dimensions. Recent work related to this direction includes the composite indirect inference method \citep{gourieroux2018composite}, which used several auxiliary model criteria instead of a unified criterion $\rho_n$ to allow parallel optimization for better computational efficiency and speed.

While we aimed to make our assumptions as streamlined and interpretable as possible, the conditions remain highly technical. Future work could investigate alternative assumptions as well as special cases that are sufficient for our assumptions to hold. 

We showed in our numerical experiments that our methodology attains superior width compared to the repro methodology of \citet{awan2025simulation} while maintaining satisfactory coverage. However, \citet{awan2025simulation} has finite-sample coverage guarantees that hold under minimal assumptions, whereas our approach is asymptotic with several technical conditions.  Two main limitations in the repro methodology of \citet{awan2025simulation} are 1) there is no general prescription for an (even asymptotically) optimal test statistic and 2) by default the repro methodology results in a confidence set for all parameters, giving overly conservative coverage for the parameter of interest. One may consider whether the ideas developed in this paper could be incorporated in the repro framework to ideally obtain both provable coverage and asymptotically optimal width. 

\section*{Acknowledgements}
The authors gratefully acknowledge support from NSF awards SES-2150615 and SES-2610910.

\appendix
\section{Proofs and technical details}
\label{sec:pb:proof}
In this section, we provide proofs of the theoretical results in this paper for PB and indirect estimators, and we include the details of our simulation results.

\subsection{Proofs for Section \ref{sec:pb:background_pb}}\label{sec:pb:proof_bg}

The literature studying bootstrap used in HTs \citep{beran1986simulated, davidson1999size, davidson2006power}, mostly focused on the difference in power between a bootstrap test and one based on the true distribution, or the limit of power on a sequence of alternative hypotheses converging to the null hypothesis with rate $O(1/\sqrt{n})$. Our Lemma \ref{lem:CI_HT_consist} is simpler as it is only about the validity of PB inference.

\begin{proof}[Proof for Lemma \ref{lem:CI_HT_consist}]

    \noindent\textbf{(1) The consistency of PB confidence sets.}
    
By the following equality,
$$\mathbb{P}\l(\tau^* \in \l\{\hat\tau(s) - \hat\Sigma(s)^{\frac{1}{2}}\xi ~\Big|~ \|\xi\|_p \leq  \hat\xi_{1-\alpha}(s)\r\}\r) = \mathbb{P}\l(\|\hat\Sigma(s)^{-\frac{1}{2}}\l(\hat\tau(s) - \tau(\theta^*)\r)\|_p \leq  \hat\xi_{1-\alpha}(s)\r),$$ 
we can replace the $\frac{\hat\theta_n - \theta}{\hat\sigma_n}$ with $\|\hat\Sigma(s)^{-\frac{1}{2}}\l(\hat\tau(s) - \tau(\theta^*)\r)\|_p$ in the proof of Lemma 23.3 by \citet{van2000asymptotic} to obtain the consistency of bootstrap confidence sets.

\noindent\textbf{(2) The consistency of PB HT.}

By the duality between confidence sets and HTs, we have Corollary \ref{cor:HT_consist} for PB HTs being asymptotically of level $\alpha$.

\begin{corollary}[Asymptotic level of PB HTs]\label{cor:HT_consist}
    Suppose that for every $\theta^*\in\Theta_0$, there is a random variable $T_{\theta^*}$ with continuous CDF such that
    $\|\hat\Sigma(s)^{-\frac{1}{2}}\l(\hat\tau(s) - \tau(\theta^*)\r)\|_p \xrightarrow[]{d} T_{\theta^*}$, and the PB estimator $\hat\tau(s_b)$ is consistent where $s_b\sim \mathbb{P}_{\hat\theta(s)}(\cdot | s)$. 
    Let $\hat\xi_{1-\alpha}(s)$ be the $(1-\alpha)$-quantile of the distribution of $\l(\l\|\hat\Sigma(s_b)^{-\frac{1}{2}}\l(\hat\tau(s_b) - \tau(\hat\theta(s))\r)\r\|_p ~\Big|~ s\r)$. 
    Then the sequence of PB HTs with the test statistic $T_n(s):=\inf_{\theta^*\in\Theta_0}\|\hat\Sigma(s)^{-\frac{1}{2}}\l(\hat\tau(s) - \tau(\theta^*)\r)\|_p$ and the critical region $K_n(s):=(\hat\xi_{1-\alpha}(s), \infty)$ is asymptotically of level $\alpha$. 
\end{corollary}

    
    In Corollary \ref{cor:HT_consist}, we know that PB HTs are asymptotically of level $\alpha$. Therefore, according to Definition \ref{def:HT_consist}, we only need to prove that $\pi_n(\theta_1)\rightarrow 1$ for all $\theta_1\in\Theta_1$, i.e., for $s\sim \mathbb{P}_{\theta_1}$,
    \[
    \liminf_{n\rightarrow\infty}\mathbb{P}_{n, \theta_1}\l(\inf_{\theta^*\in\Theta_0} \|\hat\Sigma(s)^{-\frac{1}{2}}\l(\hat\tau(s) - \tau(\theta^*)\r)\|_p > \hat\xi_{1-\alpha}(s)\r) = 1.
    \]
    
    We denote the distribution of $\l(\l\|\hat\Sigma(s_b)^{-\frac{1}{2}}\l(\hat\tau(s_b) - \tau(\hat\theta(s))\r)\r\|_p ~\Big|~ s\r)$ as $F_n$ and the distribution of $T_{\theta_1}$ as $F$. As the PB estimator $\hat\tau(s_b)$ is consistent, $F_n$ converges weakly to $F$, and the quantile functions $F_n^{-1}$ converge to the quantile function $F^{-1}$ at every continuity point. Therefore, $\hat\xi_{1-\alpha}(s)=F_n^{-1}(1-\alpha)$ converges to $F^{-1}(1-\alpha)$ in probability, and to prove the result, we consider the subsequence of $\{\hat\xi_{1-\alpha}(s)\}_{n=1}^{\infty}$ that converges to $F^{-1}(1-\alpha)$ almost surely. From $\inf_{\theta\in\Theta_0}\l\|\tau(\theta_1) - \tau(\theta)\r\|_p > 0$, $\|\hat\Sigma(s)\|_p \xrightarrow{P} 0$, and $\|\hat\Sigma(s)^{-\frac{1}{2}}\l(\hat\tau(s) - \tau(\theta_1)\r)\|_p \xrightarrow{d} T_{\theta_1}$, we have 
    \(
    \inf_{\theta\in\Theta_0}\left\|\hat\Sigma(s)^{-\frac{1}{2}}\big(\tau(\theta_1)-\tau(\theta)\big)\right\|_p \xrightarrow{P} \infty.
    \)
    By the triangle inequality,
    \[
    \inf_{\theta\in\Theta_0}\left\|\hat\Sigma(s)^{-\frac{1}{2}}\big(\hat\tau(s)-\tau(\theta)\big)\right\|_p
    \ge
    \inf_{\theta\in\Theta_0}\left\|\hat\Sigma(s)^{-\frac{1}{2}}\big(\tau(\theta_1)-\tau(\theta)\big)\right\|_p
    -
    \left\|\hat\Sigma(s)^{-\frac{1}{2}}\big(\hat\tau(s)-\tau(\theta_1)\big)\right\|_p.
    \]
    The first term diverges to $\infty$ in probability, while the second term is $O_P(1)$, so the left-hand side diverges to $\infty$ in probability as well. Since $\hat\xi_{1-\alpha}(s)=O_P(1)$, it follows that
    \[
    \mathbb{P}_{n,\theta_1}\!\left(\inf_{\theta\in\Theta_0}\left\|\hat\Sigma(s)^{-\frac{1}{2}}\big(\hat\tau(s)-\tau(\theta)\big)\right\|_p > \hat\xi_{1-\alpha}(s)\right)\to 1,
    \]
    which proves the desired consistency.

\end{proof}

\begin{proof}[Proof for Proposition \ref{prop:parambootconsistency}]
    
    This proof mainly follows the proof by \citet[Theorem 1]{ferrando2022parametric}.

    For $s\sim \mathbb{P}_{\theta^*}$, let $\sqrt{n}(\hat\tau(s) - \tau(\theta^*)) \xrightarrow[]{d} V\sim H(\theta^*)$. As $n(\hat\Sigma(s)) \xrightarrow[]{P} \Sigma(\theta^*)$, by Slutsky's theorem, we know that $\|\hat\Sigma(s)^{-\frac{1}{2}}\l(\hat\tau(s) - \tau(\theta^*)\r)\|_p \xrightarrow[]{d} \lVert\Sigma(\theta^*)^{-\frac{1}{2}} V\rVert_p := T$. We denote $T\sim F_T$, and since $H(\theta^*)$ is continuous, $F_T$ is continuous as well. Therefore, from Definition \ref{def:boot_consist}, we only need to prove $\mathbb{P}_{\hat\theta(s)}\l(\l\|\hat\Sigma(s_b)^{-\frac{1}{2}}\l(\hat\tau(s_b) - \tau(\hat\theta(s))\r)\r\|_p \leq t ~\Big|~ s\r) \xrightarrow[]{P} F_T(t)$ for all $t$. 
    
    For $s\sim \mathbb{P}_{\theta^* + \frac{h_n}{\sqrt{n}}}$, we assumed $n \hat\Sigma(s) \xrightarrow[]{P} \Sigma(\theta^*)$ and we have $\sqrt{n}\left(\hat\tau(s) - \tau(\theta^* + \frac{h_n}{\sqrt{n}})\right) \xrightarrow[]{d} V$ as $\hat\tau$ is asymptotically equivariant. By Slutsky's theorem, we have 
    \[
    \mathbb{P}_{\theta^* + \frac{h_n}{\sqrt{n}}} \l(\l\|\hat\Sigma(s)^{-\frac{1}{2}}\l(\hat\tau(s) - \tau\l(\theta^* + \frac{h_n}{\sqrt{n}}\r)\r)\r\|_p
     \leq t\r) \xrightarrow[]{} F_T(t).
    \] 
    We use $\hat{h}_n := \sqrt{n}(\hat\theta(s) - \theta^*)$ which is $O_p(1)$ to replace $h_n$ above. Following \citet[Lemma 3]{ferrando2022parametric} (for easier reference, we included it as Lemma \ref{lemma:converge_in_prob} below), we have $\mathbb{P}_{\hat\theta(s)}\l(\l\|\hat\Sigma(s_b)^{-\frac{1}{2}}\l(\hat\tau(s_b) - \tau(\hat\theta(s))\r)\r\|_p \leq t ~\Big|~ s\r) \xrightarrow[]{P} F_T(t)$ where $s_b\sim \mathbb{P}_{\hat\theta(s)}$.
\end{proof}

\begin{lemma}[Lemma 3 by \citet{ferrando2022parametric}]\label{lemma:converge_in_prob}
    Suppose $g_n$ is a sequence of functions such that $g_n(h_n) \rightarrow 0$ for any fixed sequence $h_n = O(1)$. Then $g_n(\hat{h}_n)\xrightarrow[]{P} 0$ for every random sequence $\hat{h}_n = O_P (1)$.
\end{lemma}

\subsection{Proofs for Section \ref{sec:pb:asymp_theory}}

For proving Theorems \ref{thm:indirect_est} and \ref{thm:ada_indirect_est}, we recall the consistency result of M-estimators \citep{van2000asymptotic} in Proposition \ref{prop:mestconsistency}.
\begin{proposition}[Consistency of M-estimator; \citet{van2000asymptotic}]\label{prop:mestconsistency}
    
    Let $M_n$ be random functions and let $M$ be a fixed function of $\theta$ such that for every $\ep>0$, we have $\sup _{\theta \in \Theta}\left|M_n(\theta)-M(\theta)\right| \xrightarrow{\mathrm{P}} 0$ and $\sup _{\theta: \|\theta - \theta^*\| \geq \varepsilon} M(\theta)<M(\theta^*)$.
    Then any sequence of estimators $\hat{\theta}_n$ with $M_n\left(\hat{\theta}_n\right) \geq M_n(\theta^*)-o_p(1)$ converges in probability to $\theta^*$. Note that a special choice of $\hat\theta_n$ is $\hat\theta_n:=\mathrm{argmax}_\theta M_n(\theta)$.
\end{proposition}

As $\hat{\theta}_{\mathrm{IND}}$ and $\hat{\theta}_{\mathrm{ADI}}$ both depend on the observed statistic $s$, we analyze the asymptotic distribution of $s$ in the following lemma.
\begin{restatable}[Asymptotic analysis of $s$]{lemma}{betadist}\label{lem:beta_dist}
    
    Under assumptions \textbf{(A1, A2, A6, A7)}, we have $\sqrt{n} (s - \beta^*) \xrightarrow{\mathrm{d}} (J^*)^{-1} v $ where $v\sim F_{\rho,u,\mathrm{DP}}^*$.
\end{restatable}

\begin{proof}[Proof for Lemma \ref{lem:beta_dist}]
    From \textbf{(A1, A2)} and Proposition \ref{prop:mestconsistency}, we know $s$ is a consistent estimator of $\beta^*$, i.e., $s = \beta^* + o_p(1)$. 
    As $s := \mathrm{argmax}_\beta \rho_n(\beta; D, u_{\mathrm{DP}})$, by \textbf{(A6)} and $s$ is an interior point in $\mathbb{B}$, we have the first-order condition 
    \[
    \frac{\partial \rho_n}{\partial \beta}(s, D, u_{\mathrm{DP}})=0.
    \]
     By \textbf{(A7)} and $s = \beta^* + o_p(1)$, we use Taylor expansion of $\frac{\partial \rho_n}{\partial \beta}(s, D, u_{\mathrm{DP}})$ on $s$ at $\beta^*$ and obtain 
    \[
    0 = \frac{\partial \rho_n}{\partial \beta}(\beta^*, D, u_{\mathrm{DP}}) + \frac{\partial^2 \rho_n}{(\partial \beta)(\partial \beta^\intercal)}(\beta^*, D, u_{\mathrm{DP}}) (s - \beta^*) + o_P(\|s - \beta^*\|_2).
    \]
    Using \textbf{(A7)}, we have 
    \[
    \begin{aligned}
        \sqrt{n} (s - \beta^*) &= -\l(\frac{\partial^2 \rho_n}{(\partial \beta)(\partial \beta^\intercal)}(\beta^*, D, u_{\mathrm{DP}})\r)^{-1} \sqrt{n}\l(\frac{\partial \rho_n}{\partial \beta}(\beta^*, D, u_{\mathrm{DP}}) + o_P(\|s - \beta^*\|_2)\r) \\
        &= \l(J^* + o_p(1)\r)^{-1} \sqrt{n}\l(-\frac{\partial \rho_n}{\partial \beta}(\beta^*, D, u_{\mathrm{DP}})\r) + o_P(\sqrt{n}\|s - \beta^*\|_2).
    \end{aligned} 
    \]
    By \textbf{(A6)}, we have $\l(J^* + o_p(1)\r)^{-1} \sqrt{n}\l(-\frac{\partial \rho_n}{\partial \beta}(\beta^*, D, u_{\mathrm{DP}})\r) = O_P(1)$, therefore, $\sqrt{n} (s - \beta^*) = O_P(1)$. By Slutsky's theorem, we have 
    $\sqrt{n} (s - \beta^*) \xrightarrow{\mathrm{d}} (J^*)^{-1} v,\quad v\sim F_{\rho,u,\mathrm{DP}}^*.$ 
\end{proof}

\begin{proof}[Proof for Theorem \ref{thm:indirect_est}]
    
    Our proof contains three steps corresponding to the three parts of Theorem \ref{thm:indirect_est}.
\\
    \noindent\textbf{(1) The consistency of $\hat{\theta}_{\mathrm{IND}}$.}
    
    From \textbf{(A3, A4)}, we have $\|s - b(\theta^*)\| \xrightarrow{\mathrm{P}} 0$, $\|\Omega_n - \Omega\|_F \xrightarrow{} 0$. Let 
    \[
    \begin{aligned}
        M(\theta) &:= -(b(\theta^*) -b(\theta))^\intercal \Omega (b(\theta^*) -b(\theta)), \\
        M_n(\theta) &:= -\left[s-\frac{1}{R} \sum_{r=1}^R s^r\left(\theta\right)\right]^{\intercal} \Omega_n\left[s-\frac{1}{R} \sum_{r=1}^R s^r\left(\theta\right)\right].
    \end{aligned}
    \] 
    From \textbf{(A3, A4)} again, we have $\sup_{\theta \in \Theta}\left|M_n(\theta)-M(\theta)\right| \xrightarrow{\mathrm{P}} 0$.
    As $\Omega$ is positive definite by \textbf{(A4)} and $b(\cdot)$ is one-to-one by \textbf{(A2)}, we also have $\sup _{\theta: \|\theta - \theta^*\| \geq \varepsilon} M(\theta) < 0 = M(\theta^*)$ for any $\varepsilon > 0$. By Proposition \ref{prop:mestconsistency}, $\hat{\theta}_{\mathrm{IND}}$ is a consistent estimator of $\theta^*$.
\\
    \noindent\textbf{(2) The asymptotic distribution of $\sqrt{n}(\hat{\theta}_{\mathrm{IND}} - \theta^*)$.}

    From \textbf{(A5)}, we know that $\left[s-\frac{1}{R} \sum_{r=1}^R s^r(\theta)\right]^{\intercal} \Omega_n\left[s-\frac{1}{R} \sum_{r=1}^R s^r(\theta)\right]$ is differentiable in $\theta$. 
    Therefore, 
    $\hat{\theta}_{\mathrm{IND}}$ being an interior point in $\Theta$ satisfies the first-order condition
    \[
    \left[\frac{1}{R} \sum_{r=1}^R \frac{\partial s^r(\hat{\theta}_{\mathrm{IND}})}{\partial \theta}\right]^{\intercal} \Omega_n\left[s-\frac{1}{R} \sum_{r=1}^R s^r(\hat{\theta}_{\mathrm{IND}})\right]=0.
    \]    
    
    As $\hat{\theta}_{\mathrm{IND}}$ is a consistent estimator of $\theta^*$, we have $\hat{\theta}_{\mathrm{IND}} = \theta^* + o_p(1)$.
    Using Taylor expansion of $s^r(\hat{\theta}_{\mathrm{IND}})$ on $\hat{\theta}_{\mathrm{IND}}$ at $\theta^*$, we have 
    $$
    \left[\frac{1}{R} \sum_{r=1}^R \frac{\partial s^r(\hat{\theta}_{\mathrm{IND}})}{\partial \theta}\right]^{\intercal} \Omega_n\left[s-\frac{1}{R} \sum_{r=1}^R \l(s^r(\theta^*)- \frac{\partial s^r(\theta^*)}{\partial \theta} (\hat{\theta}_{\mathrm{IND}} - \theta^*)\r) + o_P(\|\hat{\theta}_{\mathrm{IND}} - \theta^*\|_2)\right]=0.
    $$
    As $\frac{\partial s^r(\theta)}{\partial \theta}$ is continuous in $\theta$ by \textbf{(A5)} and $\Omega_n = \Omega + o(1)$ by \textbf{(A4)}, we know that
    \begin{align*}
    \left[\frac{1}{R} \sum_{r=1}^R \frac{\partial s^r(\theta^*)}{\partial \theta}+o_p(1)\right]^{\intercal} ({\Omega} + o_p(1)) &\Big[s-\frac{1}{R} \sum_{r=1}^R \l(s^r(\theta^*)- \frac{\partial s^r(\theta^*)}{\partial \theta} (\hat{\theta}_{\mathrm{IND}} - \theta^*)\r) 
    \\&+ o_P(\|\hat{\theta}_{\mathrm{IND}} - \theta^*\|_2)\Big]
    \end{align*}
    is zero. 
    Using \textbf{(A5)}, we have $\frac{1}{R} \sum_{r=1}^R \frac{\partial s^r(\theta^*)}{\partial \theta}=B^*+o_p(1)$ and
    $$
    \left[B^*+o_p(1)\right]^{\intercal} \left({\Omega} + o_p(1)\right) \left[s-\frac{1}{R} \sum_{r=1}^R s^r(\theta^*)-  B^*(\hat{\theta}_{\mathrm{IND}} - \theta^*) + o_P(\|\hat{\theta}_{\mathrm{IND}} - \theta^*\|_2)\right]=0.
    $$    
    Therefore,
    \begin{equation}
    \sqrt{n}(\hat{\theta}_{\mathrm{IND}} - \theta^*) =\left(\left((B^*)^{\intercal} {\Omega} B^*\right)^{-1} (B^*)^{\intercal} {\Omega} + o_p(1) \right) \sqrt{n} \left[s-\frac{1}{R} \sum_{r=1}^R s^r(\theta^*)\right] + o_P(\sqrt{n}\|\hat{\theta}_{\mathrm{IND}} - \theta^*\|_2).\label{eq:finite_dist}    
    \end{equation}
    
    From Lemma \ref{lem:beta_dist}, we have $\sqrt{n} (s - \beta^*) \xrightarrow{\mathrm{d}} (J^*)^{-1} v $ where $v\sim F_{\rho,u,\mathrm{DP}}^*$, and
    similarly, $\sqrt{n} (s^r(\theta^*) - \beta^*) \xrightarrow{\mathrm{d}} (J^*)^{-1} v $ where $v\sim F_{\rho,u,\mathrm{DP}}^*$. Using 
    \[s-\frac{1}{R} \sum_{r=1}^R s^r(\theta^*) = (s-\beta^*)-\frac{1}{R} \sum_{r=1}^R (s^r(\theta^*) - \beta^*),
    \]
    and the independence between $s$ and $s^r$, we have 
    \[
    \sqrt{n}\l(s-\frac{1}{R} \sum_{r=1}^R s^r(\theta^*)\r) \xrightarrow{\mathrm{d}} (J^*)^{-1}\l(v_0 - \frac{1}{R}\sum_{r=1}^R v_r\r),
    \] 
    where $v_r \iid F_{\rho,u,\mathrm{DP}}^*$ for $r=0,1,\ldots,R$.

    Therefore, we have $\sqrt{n}(\hat{\theta}_{\mathrm{IND}} - \theta^*) = O_P(1)$ and more specifically,
    \[
    \sqrt{n}(\hat{\theta}_{\mathrm{IND}} - \theta^*) \xrightarrow{\mathrm{d}} \left((B^*)^{\intercal} {\Omega} B^*\right)^{-1} (B^*)^{\intercal} {\Omega} (J^*)^{-1} (v_0 - \frac{1}{R}\sum_{r=1}^R v_r),
    \] 
    where $v_r \iid F_{\rho,u,\mathrm{DP}}^*$ for $r=0,1,\ldots,R$.
\\
    \noindent\textbf{(3) The asymptotic equivariance of $s$ and $\hat{\theta}_{\mathrm{IND}}$.}

    As $s$ and $\hat{\theta}_{\mathrm{IND}}$ depend on $\theta^*$ and random seeds, when we replace $\theta^*$ by $\theta_1$, we fix the random seeds and use $s(\theta_1)$ and $\hat{\theta}_{\mathrm{IND}}(\theta_1)$ to denote the corresponding $s$ and $\hat{\theta}_{\mathrm{IND}}$, respectively. 

    Based on the proof for the asymptotic distribution above,
    we replace $\theta^*$ by $\l(\theta^*+\frac{h_n}{\sqrt{n}}\r)$ and prove that the limiting distribution of $\sqrt{n}\l[\hat{\theta}_{\mathrm{IND}}\l(\theta^*+\frac{h_n}{\sqrt{n}}\r) - \l(\theta^*+\frac{h_n}{\sqrt{n}}\r)\r]$ is the same as the limiting distribution of $\sqrt{n}\l[\hat{\theta}_{\mathrm{IND}}(\theta^*) - \theta^*\r]$ using Equation (\ref{eq:finite_dist}). 
    
    From \textbf{(A5)}, $\frac{\partial b(\theta)}{\partial\theta}$ and $\frac{\partial s(\theta)}{\partial \theta}$ are continuous in $\theta$, and $\frac{\partial s(\theta)}{\partial \theta} = \frac{\partial b(\theta)}{\partial\theta} + o_p(1)$. Therefore, the change in $\l(\left((B^*)^{\intercal} {\Omega} B^*\right)^{-1} (B^*)^{\intercal} {\Omega}\r)$ due to $\frac{h_n}{\sqrt{n}}$ is $o(1)$, and by Taylor expansion on $s\l(\theta^* + \frac{h_n}{\sqrt{n}}\r)$,  we have $s\l(\theta^* + \frac{h_n}{\sqrt{n}}\r) = s(\theta^*) + \frac{\partial s(\theta^*)}{\partial \theta} \frac{h_n}{\sqrt{n}} + o_p(\|\frac{h_n}{\sqrt{n}}\|_2)$. Note that the convergence is only for $\theta^*$ instead of all possible $\theta \in \Theta$, so uniform convergence is not needed. Therefore,
    \begin{equation}
        \begin{aligned}
            s\l(\theta^* + \frac{h_n}{\sqrt{n}}\r) = s(\theta^*) + \frac{\partial b(\theta^*)}{\partial \theta}\frac{h_n}{\sqrt{n}} + o_p\l(\frac{\|h_n\|_2}{\sqrt{n}}\r)= s(\theta^*) + \frac{\partial b(\theta^*)}{\partial \theta}\frac{h_n}{\sqrt{n}}+o_p\l(\frac{1}{\sqrt{n}}\r).
        \end{aligned} \label{eq:asymp_eqi_ind_s}        
    \end{equation}
    By Equation (\ref{eq:asymp_eqi_ind_s}) and Taylor expansion on $b\l(\theta^* + \frac{h_n}{\sqrt{n}}\r)$, we have 
    $b\l(\theta^* + \frac{h_n}{\sqrt{n}}\r) = b(\theta^*) + \frac{\partial b(\theta^*)}{\partial \theta}\frac{h_n}{\sqrt{n}}+o_p\l(\frac{1}{\sqrt{n}}\r)$. Therefore, the part of $\frac{\partial b(\theta^*)}{\partial \theta}\frac{h_n}{\sqrt{n}}$ will cancel out in the following equation
    \[
    \begin{aligned}
        \sqrt{n}\l(s\l(\theta^* + \frac{h_n}{\sqrt{n}}\r) - b\l(\theta^* + \frac{h_n}{\sqrt{n}}\r)\r) = \sqrt{n}\l(s(\theta^*) - b(\theta^*)\r) + o_p\l(1\r),
    \end{aligned}
    \]     
    which means $\sqrt{n}\l(s\l(\theta^* + \frac{h_n}{\sqrt{n}}\r) - b\l(\theta^* + \frac{h_n}{\sqrt{n}}\r)\r)$ and $\sqrt{n}\l(s(\theta^*) - b(\theta^*)\r)$ have the same limiting distribution (Lemma \ref{lem:beta_dist}), i.e., $s$ is asymptotically equivariant.
    
    From Equation (\ref{eq:asymp_eqi_ind_s}), the part of $\frac{\partial b(\theta^*)}{\partial \theta}\frac{h_n}{\sqrt{n}}$ will also cancel out in the following equation
    \[
    \sqrt{n}\l(s\l(\theta^* + \frac{h_n}{\sqrt{n}}\r) -\frac{1}{R} \sum_{r=1}^R s^r\left(\theta^* + \frac{h_n}{\sqrt{n}}\right)\r) = \sqrt{n}\l(s(\theta^*)-\frac{1}{R} \sum_{r=1}^R s^r(\theta^*)\r) + o_p\l(1\r).
    \] 
    By Equation (\ref{eq:finite_dist}), we have
    \[
    \begin{aligned}
        &\sqrt{n}\l[\hat{\theta}_{\mathrm{IND}}\l(\theta^*+\frac{h_n}{\sqrt{n}}\r) - \l(\theta^*+\frac{h_n}{\sqrt{n}}\r)\r] \\
        =& \left(\left((B^*)^{\intercal} {\Omega} B^*\right)^{-1} (B^*)^{\intercal} {\Omega} + o_p(1) \right) \sqrt{n} \left[s\l(\theta^* + \frac{h_n}{\sqrt{n}}\r) - \frac{1}{R} \sum_{r=1}^R s^r\l(\theta^* + \frac{h_n}{\sqrt{n}}\r)\right]\\
        =&\l[((B^*)^{\intercal} {\Omega} B^*)^{-1} (B^*)^{\intercal} {\Omega} + o_p(1) \r] \l\{\sqrt{n} \l[s(\theta^*)-\frac{1}{R} \sum_{r=1}^R s^r(\theta^*)\r] + o_p(1)\r\}.
    \end{aligned}
    \]
    Therefore, 
    $\sqrt{n}\l[\hat{\theta}_{\mathrm{IND}}\l(\theta^*+\frac{h_n}{\sqrt{n}}\r) - \l(\theta^*+\frac{h_n}{\sqrt{n}}\r)\r]$ and $ \sqrt{n}\l(\hat{\theta}_{\mathrm{IND}} - \theta^*\r)$ have the same limiting distribution. 
\end{proof}

\begin{remark}\label{rmk:unif_convergence}
    The uniform convergence assumptions \textbf{(A1, A3)} seem strong, and are needed to derive the consistency of the indirect estimator as an M-estimator. They can be replaced by the compactness and continuity assumptions in \textbf{(A1', A3')}.
     Note that \textbf{(A1')} immediately implies \textbf{(A1)} since $\mathbb{B}$ is compact and $\rho_\infty$ is continuous in $\beta$. By \textbf{(A1, A2)}, and Proposition \ref{prop:mestconsistency}, we have $\|s^r(\theta)-b(\theta)\|\xrightarrow{\mathrm{P}} 0$ (pointwise convergence in probability); assuming \textbf{(A3')} as well, 
    then following \citet[Corollary 2.2]{newey1991uniform}, we also have \textbf{(A3)}. Thus, \textbf{(A1',A2, A3')} imply \textbf{(A1,A2,A3)}.
    
    \noindent \textbf{(A1')} {$\mathbb{B}$ is compact, $\rho_\infty(\beta; F_u, F_{\mathrm{DP}}, \theta^*)$ is a non-stochastic and continuous function in $\beta$, and for any $\beta\in\mathbb{B}$, $|\rho_n(\beta; D, u_{\mathrm{DP}}) - \rho_\infty(\beta; F_u, F_{\mathrm{DP}}, \theta^*)| \xrightarrow{\mathrm{P}} 0$. There is $C_n$ such that $C_n = O_p(1)$ and for all $\beta, \beta' \in \mathbb{B}$, $\|\rho_n(\beta; D, u_{\mathrm{DP}}) - \rho_n(\beta'; D, u_{\mathrm{DP}})\| \leq C_n \|\beta - \beta'\|$.} \\
    \noindent \textbf{(A3')} {$\Theta$ is compact, and $b(\theta)$ is continuous in $\theta$. There is $B_n$ such that $B_n = O_p(1)$ and for all $\theta, \theta' \in \Theta$, $\|s^r(\theta) - s^r(\theta')\| \leq B_n \|\theta - \theta'\|$ (global, stochastic Lipschitz condition).} 
\end{remark}

    

\begin{proof}[Proof for Theorem \ref{thm:ada_indirect_est}]
    
    Our proof contains four steps corresponding to the four parts of Theorem \ref{thm:ada_indirect_est}. To simplify our notations, we write $\hat{\theta}_{\mathrm{ADI}}$ as $\hat\theta$. Note that the following asymptotic results are based on $n\rightarrow\infty$ with $R_n\rightarrow\infty$.
\\
    \noindent\textbf{(1) The consistency of $\hat{\theta}$.}
    From \textbf{(A3)}, we have 
    $s \xrightarrow{\mathrm{P}} b(\theta^*)$. 
    From \textbf{(A8)}, we have
    $m^{R_n}(\theta) \xrightarrow{\mathrm{P}} b(\theta)$ and $nS^{R_n}(\theta) \xrightarrow{\mathrm{P}} \Sigma(\theta)$ uniformly for every $\theta$ as well.
    Therefore, $\|s - m^{R_n}(\theta)\|_2 \xrightarrow{\mathrm{P}} \|b(\theta^*) - b(\theta)\|_2$ uniformly for all $\theta$. From \textbf{(A2)}, $\|b(\theta^*) - b(\theta)\|_2 > 0$ if $\theta \neq \theta^*$. Therefore, uniformly for $\|\theta - \theta^*\|_2 \geq \epsilon > 0$ the objective function $(s - m^{R_n}(\theta))^{\intercal} (S^{R_n}(\theta))^{-1} (s - m^{R_n}(\theta))$ diverges to infinity in probability since $\Sigma(\theta) \succ 0$ and $(s - m^{R_n}(\theta))^{\intercal} (nS^{R_n}(\theta))^{-1} (s - m^{R_n}(\theta)) \xrightarrow{\mathrm{P}} (b(\theta^*) - b(\theta))^{\intercal} (\Sigma(\theta))^{-1} (b(\theta^*) - b(\theta)) \geq 0$, i.e., for any constants $C>0$, $\epsilon>0$, and $\delta>0$, there exists $N_1$ such that 
    \begin{equation}\label{eq:adi_consistency1}
    \mathbb{P}((s - m^{R_n}(\theta))^{\intercal} (S^{R_n}(\theta))^{-1} (s - m^{R_n}(\theta)) > C) > 1-\delta/2, \quad \forall n\geq N_1, \|\theta-\theta^*\|_2 \geq \epsilon.        
    \end{equation}
    
    
    From \textbf{(A8)} and Lemma \ref{lem:beta_dist}, we have $\|\sqrt{n}(s - \beta^*)\|_2=O_p(1)$ and $\|\sqrt{n}(m^{R_n}(\theta^*) - \beta^*)\|_2=o_p(1)$, and $\|(n S^{R_n}(\theta^*))^{-1}\|_2 = O_p(1)$. Therefore, 
    $
    (s - m^{R_n}(\theta^*))^{\intercal} (S^{R_n}(\theta^*))^{-1} (s - m^{R_n}(\theta^*)) = (\sqrt{n}(s - m^{R_n}(\theta^*)))^{\intercal} (n S^{R_n}(\theta^*))^{-1} (\sqrt{n}(s - m^{R_n}(\theta^*))) = O_p(1)
    $, i.e.,  
    for any $\delta > 0$, there exist $N_2$ and $C_1$ such that 
    \begin{equation}\label{eq:adi_consistency2}
    \mathbb{P}\l((s - m^{R_n}(\theta^*))^{\intercal} (S^{R_n}(\theta^*))^{-1} (s - m^{R_n}(\theta^*)) < C_1\r) > 1- \delta/2,\quad \forall n\geq N_2.   
    \end{equation}
    
    Combining the inequalities (\ref{eq:adi_consistency1}) and (\ref{eq:adi_consistency2}), there exists $C \geq C_1, N=\max(N_1, N_2)$ such that 
    \[
    \mathbb{P}((s - m^{R_n}(\theta^*))^{\intercal} (S^{R_n}(\theta^*))^{-1} (s - m^{R_n}(\theta^*)) < (s - m^{R_n}(\theta))^{\intercal} (S^{R_n}(\theta))^{-1} (s - m^{R_n}(\theta))) > 1- \delta,
    \]
    for all $n\geq N$ and $\|\theta-\theta^*\|_2 \geq \epsilon$. 
    Therefore, the minimizer of the objective function, $\hat{\theta}$, must be within $\epsilon$ of $\theta^*$ with probability $1-\delta$, i.e., $\lim_{n\rightarrow \infty }\mathbb{P}(\|\hat{\theta} - \theta^*\|_2 \geq \epsilon) = 0$.
\\
    \noindent\textbf{(2) The asymptotic distribution of $\sqrt{n}(\hat{\theta} - \theta^*)$.}
    
    
    From \textbf{(A8)} and Lemma \ref{lem:beta_dist}, we know $\|\sqrt{n}(m^{R_n}(\theta^*) - \beta^*)\|_2=o_p(1)$ and $\sqrt{n} (s - \beta^*) \xrightarrow{\mathrm{d}} (J^*)^{-1} v$ where $v\sim F_{\rho,u,\mathrm{DP}}^*$. Therefore, we have $\sqrt{n}\l(s-m^{R_n}(\theta^*)\r) \xrightarrow{\mathrm{d}} (J^*)^{-1}v$. Also from \textbf{(A8)}, we have $\Sigma(\theta^*)^{-1}=\Var[(J^*)^{-1} v]^{-1}$. Then we relate $(\hat{\theta} - \theta^*)$ to $\l(s-m^{R_n}(\theta^*)\r)$ to derive its asymptotic distribution.

    For fixed $n$, from \textbf{(A5)}, we know that the objective function $(s - m^{R_n}(\theta))^{\intercal} (S^{R_n}(\theta))^{-1} (s - m^{R_n}(\theta))$ is differentiable in $\theta$. We write $\theta=(\theta_1,\ldots, \theta_p)$ where $\theta_i$ denotes the $i$th entry of $\theta$.
    Then, as $\hat{\theta}$ is an interior point in $\Theta$ from \textbf{(A9)}, $\hat{\theta}$ satisfies the first-order conditions with respect to $\theta_i$,
    \begin{equation}
     -2\left(\frac{\partial m^{R_n}(\hat\theta)}{\partial \theta_i}\right)^{\intercal}  (n S^{R_n}(\hat\theta))^{-1} (s - m^{R_n}(\hat\theta)) + (s - m^{R_n}(\hat\theta))^{\intercal}  \left(\frac{\partial  (n S^{R_n}(\hat\theta))^{-1}}{\partial \theta_i}\right) (s - m^{R_n}(\hat\theta)) =0, \label{eq:ada_first0}
    \end{equation}
    where we replace $(S^{R_n}(\theta))^{-1}$ with $(n S^{R_n}(\theta))^{-1}$ since $(n S^{R_n}(\theta))^{-1} \xrightarrow{\mathrm{P}} \Sigma(\theta)^{-1}$ by \textbf{(A8)}.
    
    The left side of Equation (\ref{eq:ada_first0}) can be written as 
    \begin{equation}\label{eq:ada_first}
    \left[-2\left(\frac{\partial m^{R_n}(\hat\theta)}{\partial \theta_i}\right)^{\intercal} + (s - m^{R_n}(\hat\theta))^{\intercal} (n S^{R_n}(\hat\theta))^{-1} \left(\frac{\partial (n S^{R_n}(\hat\theta))}{\partial \theta_i}\right)\right] (n S^{R_n}(\hat\theta))^{-1} (s - m^{R_n}(\hat\theta)).    
    \end{equation}
    
    By \textbf{(A9)}, $\left(\frac{\partial m^{R_n}(\hat\theta)}{\partial \theta_i}\right) = B^*_i+o_p(1)$ where $B_i^*$ is the $i$th row of $B^*$, and $(n S^{R_n}(\hat{\theta}))^{-1} = O_p(1)$ and $\frac{\partial (n S^{R_n}(\hat{\theta}))}{\partial \theta_i} = O_p(1)$. 
    Furthermore, we have $s - m^{R_n}(\hat{\theta})=(s - \beta^*) - (m^{R_n}(\hat{\theta}) - b(\hat{\theta})) - (b(\hat{\theta}) - \beta^*) =o_p(1)$ since $s - \beta^* = o_p(1)$ by \textbf{(A3)}, $m^{R_n}(\hat{\theta}) - b(\hat{\theta}) = o_p(1)$ by \textbf{(A8)}, and $b(\hat{\theta}) - \beta^* = o_p(1)$ by the continuity of $b(\cdot)$ in \textbf{(A2)} and the consistency of $\hat\theta$.     
    Therefore, 
    \[
    \left[2\left( - \frac{\partial m^{R_n}(\hat\theta)}{\partial \theta_i}\right)^{\intercal} + (s - m^{R_n}(\hat\theta))^{\intercal} (n S^{R_n}(\hat\theta))^{-1} \left(\frac{\partial (n S^{R_n}(\hat\theta))}{\partial \theta_i}\right)\right] = ( -2B^*_i+o_p(1))^{\intercal}.
    \]
    
    By \textbf{(A8)} and the continuity of $\Sigma(\theta)$, we have $(n S^{R_n}(\hat\theta))^{-1} = \Omega^*+o_p(1)$. Then, Equation (\ref{eq:ada_first}) becomes 
    \(
    ( - 2B^*+o_p(1))^{\intercal} (\Omega^*+o_p(1)) (s-m^{R_n}(\hat{\theta}))=0.
    \)
    
    As $\hat{\theta}$ is a consistent estimator of $\theta^*$, we have $\hat{\theta} = \theta^* + o_p(1)$.  We use  a Taylor expansion of $m^{R_n}(\hat\theta)$ at $\theta^*$ to obtain
    $m^{R_n}(\hat\theta) = m^{R_n}(\theta^*) +  \frac{\partial m^{R_n}(\theta^*)}{\partial \theta} (\hat{\theta} - \theta^*) + o_P(\|\hat{\theta} - \theta^*\|_2)$.
    Equation (\ref{eq:ada_first}) becomes
    $
    \left[{-}2B^*+o_p(1)\right]^{\intercal} \left({\Omega^*} + o_p(1)\right) \left[s-m^{R_n}(\theta^*)-  B^*(\hat{\theta} - \theta^*) + o_P(\|\hat{\theta} - \theta^*\|_2)\right]=0,
    $
    which can be written as
    \begin{equation}
    \sqrt{n}(\hat{\theta} - \theta^*) =\left(\left((B^*)^{\intercal} {\Omega^*} B^*\right)^{-1} (B^*)^{\intercal} {\Omega^*} + o_p(1) \right) \sqrt{n} \left[s-m^{R_n}(\theta^*)\right]+ o_P(\sqrt{n}\|\hat{\theta} - \theta^*\|_2).\label{eq:ada_finite_dist}    
    \end{equation}
    
    Therefore, $\sqrt{n}(\hat{\theta} - \theta^*) = O_P(1)$ and
    \[
    \sqrt{n}(\hat{\theta} - \theta^*) \xrightarrow{\mathrm{d}} \left((B^*)^{\intercal} \Omega^* B^*\right)^{-1} (B^*)^{\intercal} \Omega^* (J^*)^{-1}v,
    \]
    where $v \sim F_{\rho,u,\mathrm{DP}}^*$ and $\Omega^*=\Sigma(\theta^*)^{-1}=\Var[(J^*)^{-1} v]^{-1}$.
\\
    \noindent\textbf{(3) The asymptotic equivariance of $s$ and $\hat{\theta}$.}
    
    Note that this part mainly follows the corresponding part in the proof for Theorem \ref{thm:indirect_est}.
    
    Based on the proof for the asymptotic distribution above,
    we replace $\theta^*$ by $\l(\theta^*+\frac{h_n}{\sqrt{n}}\r)$ and denote the corresponding $s$ and $\hat{\theta}$ by $s\l(\theta^*+\frac{h_n}{\sqrt{n}}\r)$ and $\hat{\theta}\l(\theta^*+\frac{h_n}{\sqrt{n}}\r)$, respectively. Then, we prove that the limiting distribution of $\sqrt{n}\l[\hat{\theta}\l(\theta^*+\frac{h_n}{\sqrt{n}}\r) - \l(\theta^*+\frac{h_n}{\sqrt{n}}\r)\r]$ is the same as the limiting distribution of $\sqrt{n}\l[\hat{\theta}(\theta^*) - \theta^*\r]$ using Equation (\ref{eq:ada_finite_dist}). 
    
    From \textbf{(A5)} and \textbf{(A8)}, $\frac{\partial b(\theta)}{\partial\theta}$, $\frac{\partial s(\theta)}{\partial \theta}$, and $\Sigma(\theta)$ are continuous in $\theta$, and $\frac{\partial s(\theta)}{\partial \theta} = \frac{\partial b(\theta)}{\partial\theta} + o_p(1)$. Therefore, the change in $\l(\left((B^*)^{\intercal} {\Omega^*} B^*\right)^{-1} (B^*)^{\intercal} {\Omega^*}\r)$ due to $\frac{h_n}{\sqrt{n}}$ is $o(1)$, and by  Taylor expansion on $s\l(\theta^* + \frac{h_n}{\sqrt{n}}\r)$, we have
    \begin{equation}
        \begin{aligned}
            s\l(\theta^* + \frac{h_n}{\sqrt{n}}\r) = s(\theta^*) + \frac{\partial s(\theta^*)}{\partial \theta}\frac{h_n}{\sqrt{n}} + o_p\l(\frac{\|h_n\|_2}{\sqrt{n}}\r) = s(\theta^*) + \frac{\partial b(\theta^*)}{\partial \theta}\frac{h_n}{\sqrt{n}}+o_p\l(\frac{1}{\sqrt{n}}\r).
        \end{aligned} \label{eq:asymp_eqi_adi_s}
    \end{equation}
    
    Using another Taylor expansion on $b\l(\theta^* + \frac{h_n}{\sqrt{n}}\r) = b(\theta^*) + \frac{\partial b(\theta^*)}{\partial \theta}\frac{h_n}{\sqrt{n}}+o_p\l(\frac{1}{\sqrt{n}}\r)$, we have
    \[
    \begin{aligned}
        \sqrt{n}\l(s\l(\theta^* + \frac{h_n}{\sqrt{n}}\r) - b\l(\theta^* + \frac{h_n}{\sqrt{n}}\r)\r) = \sqrt{n}\l(s(\theta^*) - b(\theta^*)\r) + o_p\l(1\r),
    \end{aligned}
    \] 
    which means $\sqrt{n}\l(s\l(\theta^* + \frac{h_n}{\sqrt{n}}\r) - b\l(\theta^* + \frac{h_n}{\sqrt{n}}\r)\r)$ and $\sqrt{n}\l(s(\theta^*) - b(\theta^*)\r)$ have the same limiting distribution (Lemma \ref{lem:beta_dist}).
    
    Similarly, we use Taylor expansion of $m^{R_n}(\theta^* + \frac{h_n}{\sqrt{n}})$ at $\theta^*$ to obtain
     $m^{R_n}(\theta^* + \frac{h_n}{\sqrt{n}}) = m^{R_n}(\theta^*) +  \frac{\partial m^{R_n}(\theta^*)}{\partial \theta} \frac{h_n}{\sqrt{n}} + o_P(\frac{\|h_n\|_2}{\sqrt{n}})$. By \textbf{(A8)}, we have $m^{R_n}(\theta^*) - b(\theta^*) = o_P(\frac{1}{\sqrt{n}})$. Therefore,
    \[
    \begin{aligned}
        \sqrt{n}\l(m^{R_n}\left(\theta^* + \frac{h_n}{\sqrt{n}}\right) - b\l(\theta^* + \frac{h_n}{\sqrt{n}}\r)\r) = \sqrt{n}\l(m^{R_n}(\theta^*) - b(\theta^*)\r) + o_p\l(1\r),
    \end{aligned}
    \] 
    
    By Equation (\ref{eq:ada_finite_dist}), we have
    \begin{equation}
        \begin{aligned}
        &\sqrt{n}\l[\hat{\theta}\l(\theta^*+\frac{h_n}{\sqrt{n}}\r) - \l(\theta^*+\frac{h_n}{\sqrt{n}}\r)\r] \\
        =&\l[((B^*)^{\intercal} {\Omega}^* B^*)^{-1} (B^*)^{\intercal} {\Omega}^* + o_p(1)\r] \l[\sqrt{n}(s-m^{R_n}(\theta^*)) + o_p(1)\r].
        \end{aligned}\label{eq:ada_asym_eqv}    
    \end{equation}
    By comparing Equation (\ref{eq:ada_asym_eqv}) to (\ref{eq:ada_finite_dist}), we know that $\sqrt{n}\l[\hat{\theta}\l(\theta^*+\frac{h_n}{\sqrt{n}}\r) - \l(\theta^*+\frac{h_n}{\sqrt{n}}\r)\r]$ and $ \sqrt{n}(\hat{\theta} - \theta^*)$ have the same limiting distribution. 
\\
    \noindent\textbf{(4) The optimal asymptotic variance of $\hat{\theta}_{\mathrm{ADI}}$.}
    
    We let $R$ be a constant and define $\Omega_0 := \Var[(J^*)^{-1} (v_0 - \frac{1}{R}\sum_{r=1}^{R} v_r)]^{-1}=(1+\frac{1}{R})^{-1} \Omega^*$ since $v_r \iid F_{\rho,u,\mathrm{DP}}^*$ for $r=0,1,\ldots,R$. We follow the idea by \citet{gourieroux1993indirect}: From the Gauss-Markov theorem for weighted least squares, in $\left((B^*)^{\intercal} {\Omega} B^*\right)^{-1} (B^*)^{\intercal} {\Omega} (J^*)^{-1} (v_0 - \frac{1}{R}\sum_{r=1}^{R} v_r)$, by treating $B^*$ as the covariates, $(J^*)^{-1} (v_0 - \frac{1}{R}\sum_{r=1}^{R} v_r)$ as the responses, and $\Omega$ as the weights, we know that for any $\Omega$,
    $$
    \begin{aligned}
        &~ \Var\l[\lim_{n\rightarrow \infty}\sqrt{n}(\hat{\theta}_{\mathrm{IND}} - \theta^*)\r] \\
        = &~ \Var\l[\left((B^*)^{\intercal} {\Omega} B^*\right)^{-1} (B^*)^{\intercal} {\Omega} (J^*)^{-1} \l(v_0 - \frac{1}{R}\sum_{r=1}^{R} v_r\r)\r] \\
        \succeq &~ \Var\l[\left((B^*)^{\intercal} {\Omega_0} B^*\right)^{-1} (B^*)^{\intercal} {\Omega_0} (J^*)^{-1} \l(v_0 - \frac{1}{R}\sum_{r=1}^{R} v_r\r)\r] \\
        = &~ \left((B^*)^{\intercal} {\Omega_0} B^*\right)^{-1} (B^*)^{\intercal} {\Omega_0} {\Omega_0}^{-1} {\Omega_0} (B^*) \left((B^*)^{\intercal} {\Omega_0} B^*\right)^{-1} \\
        = &~\left((B^*)^{\intercal} {\Omega_0} B^*\right)^{-1} = \l(1+\frac{1}{R}\r)\left((B^*)^{\intercal} {\Omega^*} B^*\right)^{-1}.
    \end{aligned}
    $$ 
    On the other hand, we know that 
    $$
    \begin{aligned}
    &~ \Var\l[\lim_{n\rightarrow \infty}\sqrt{n}(\hat{\theta}_{\mathrm{ADI}} - \theta^*)\r] \\
    = &~ \Var\l[\left((B^*)^{\intercal} {\Omega^*} B^*\right)^{-1} (B^*)^{\intercal} {\Omega^*} (J^*)^{-1} v\r] \\ 
    = &~ \left((B^*)^{\intercal} {\Omega^*} B^*\right)^{-1} \left((B^*)^{\intercal} {\Omega^*}{\Omega^*}^{-1}{\Omega^*} B^*\right) \left((B^*)^{\intercal} {\Omega^*} B^*\right)^{-1} \\
    = &~ \left((B^*)^{\intercal} {\Omega^*} B^*\right)^{-1}.
    \end{aligned}
    $$
    Therefore,
    $$
    \Var\l[\lim_{n\rightarrow \infty}\sqrt{n}(\hat{\theta}_{\mathrm{IND}} - \theta^*)\r] - \Var\l[\lim_{n\rightarrow \infty}\sqrt{n}(\hat{\theta}_{\mathrm{ADI}} - \theta^*)\r] \succeq 0.
    $$

    We adopt the idea by \citet[Proposition 1(iv)]{jiang2004indirect} to prove that $\hat{\theta}_{\mathrm{ADI}}$ has the smallest asymptotic variance among all consistent estimators $\psi(s)$ of $\theta$ where $\psi$ is continuously differentiable at $\beta^*$. 
    As $\psi(s)$ is a consistent estimator to $\theta^*$, we have $\psi(s) = \theta^* + o_P(1)$. As $\psi$ is continuous, we have $\psi(s) = \psi(\beta^*) + o_P(1)$.  Since $\psi(\beta^*)$ and $\theta^*$ are deterministic, this implies that $\psi(\beta^*) = \theta^*$ for all $\theta^*$; that is, $\psi\circ b$ is the identity function. We denote $\frac{\partial \psi(\beta^*)}{\partial s}$ by $\psi' \in \mathbb{R}^{q\times p}$ and $\frac{\partial b(\theta^*)}{\partial \theta}$ by $B^* \in \mathbb{R}^{p\times q}$. As $b$ is differentiable at $\theta^*$ and $\psi$ is differentiable at $b(\theta^*)$, we know that $\l(\frac{\partial \psi(\beta^*)}{\partial s}\r) \l(\frac{\partial b(\theta^*)}{\partial \theta}\r) = \psi' B^* = I$. 
    Using the delta method on $\psi(s)$ and Lemma \ref{lem:beta_dist}, we know that $\sqrt{n}(\psi(s) - \theta^*) \xrightarrow{{\mathrm{d}}} \psi' (J^*)^{-1} v $ where $v\sim F_{\rho,u,\mathrm{DP}}^*$. Therefore, $\Var\l[\lim_{n\rightarrow \infty}\sqrt{n}(\psi(s) - \theta^*)\r] = \psi' (\Omega^*)^{-1} (\psi')^\intercal$. As  
    \[
    \begin{aligned}
          &~ \psi' (\Omega^*)^{-1} (\psi')^\intercal - \l((B^*)^{\intercal} {\Omega^*} B^*\r)^{-1} \\ 
        = &~ \l(\psi' - \l((B^*)^{\intercal} {\Omega^*} B^*\r)^{-1} (B^*)^{\intercal} {\Omega^*}\r) (\Omega^*)^{-1} \l(\psi' - \l((B^*)^{\intercal} {\Omega^*} B^*\r)^{-1} (B^*)^{\intercal} {\Omega^*}\r)^\intercal,      
    \end{aligned}
    \]
    which is positive semidefinite, we have 
    \[
    \Var\l[\lim_{n\rightarrow \infty}\sqrt{n}(\psi(s) - \theta^*)\r] \succeq \Var\l[\lim_{n\rightarrow \infty}\sqrt{n}(\hat{\theta}_{\mathrm{ADI}} - \theta^*)\r].
    \]

\end{proof}

To prove Theorem \ref{thm:indirect_inf_valid_est}, we use the uniform delta method \citep[Theorem 3.8]{van2000asymptotic} which is included as Lemma \ref{lem:deltamethod} for easier reference.
\begin{restatable}[Uniform Delta method; \citet{van2000asymptotic}]{lemma}{deltamethod}\label{lem:deltamethod}
    
    Let $\phi:\mathbb{R}^k \rightarrow \mathbb{R}^m$ be a map defined and continuously differentiable in a neighborhood of $\theta$. Let $T_n$ be random vectors taking their values in the domain of $\phi$. If $r_n(T_n - \theta_n) \xrightarrow{\mathrm{d}} T$ for vectors $\theta_n\rightarrow \theta$ and numbers $r_n \rightarrow \infty$, then $r_n(\phi(T_n) - \phi(\theta_n)) \xrightarrow{\mathrm{d}} \phi'(\theta) T$. Moreover, the difference between $r_n(\phi(T_n) - \phi(\theta_n))$ and $\phi'(\theta)(r_n(T_n - \theta_n))$ converges to zero in probability.
\end{restatable}

\begin{proof}[Proof for Theorem \ref{thm:indirect_inf_valid_est}]
    
    In this proof, we denote $\hat{\theta}_{\mathrm{IND}}$ and $\hat{\theta}_{\mathrm{ADI}}$ as $\hat\theta$ for simplicity.
    
    By Theorems \ref{thm:indirect_est} and \ref{thm:ada_indirect_est}, $\hat\theta$ and $s$ are asymptotically equivariant, and we further prove that $\eta_1(\hat\theta)$ and $\eta_2(s)$ are asymptotically equivariant. From the asymptotic equivariance of $\hat\theta$, we know that $\sqrt{n}(\hat\theta - \theta^*)$ has the same asymptotic distribution as $\sqrt{n}(\hat\theta(\theta^* + \frac{h_n}{\sqrt{n}}) - (\theta^* + \frac{h_n}{\sqrt{n}}) )$. We denote this asymptotic distribution by $F_T$ and let $T\sim F_T$. In Lemma \ref{lem:deltamethod}, we let $r_n := \sqrt{n}$, $T_n := \hat\theta$, and $\phi := \eta_1$. Then, $\sqrt{n}(\eta_1(\hat\theta) - \eta_1(\theta^*)) \xrightarrow{\mathrm{d}} \eta_1'(\theta^*) T$ and $\sqrt{n}(\eta_1(\hat\theta(\theta^* + \frac{h_n}{\sqrt{n}})) - \eta_1(\theta^* + \frac{h_n}{\sqrt{n}})) \xrightarrow{\mathrm{d}} \eta_1'(\theta^*) T$, which means that $\eta_1(\hat\theta)$ is asymptotically equivariant. Similarly, from the asymptotic equivariance of $s$, we have that $\eta_2(s)$ is asymptotically equivariant.

    As $\hat\theta$ is asymptotically equivariant, we know $\hat\theta(s(\theta^* + \frac{h_n}{\sqrt{n}})) \xrightarrow{P} \theta^*$. Therefore, by the continuous mapping theorem \citep[Theorem 2.3]{van2000asymptotic},  $n\hat\Sigma(s(\theta^* + \frac{h_n}{\sqrt{n}})):=\eta_3(\hat\theta(s(\theta^* + \frac{h_n}{\sqrt{n}}))) \xrightarrow{P} \eta_3(\theta^*)$.
    
    Based on the above results on  $(\hat\tau,\hat\Sigma)$, using Proposition \ref{prop:parambootconsistency}, we know that the PB estimators with the choices of  $(\hat\tau,\hat\Sigma)$ mentioned in Theorem \ref{thm:indirect_inf_valid_est} are consistent.
\end{proof}

\subsection{Proofs for Section \ref{sec:pb:non-smooth}}

Similar to Lemma \ref{lem:beta_dist}, we can prove Lemma \ref{lem:beta_dist_smooth} by replacing $s$ with $s^r(\theta^*)$.
\begin{restatable}[Asymptotic analysis of $s^r(\theta^*)$]{lemma}{betadist}\label{lem:beta_dist_smooth}
    Under assumptions \textbf{(A1, A1s, A2, A6, A6s, A7, A7s)}, we have $\sqrt{n} (s^r(\theta^*) - \beta^*) \xrightarrow{\mathrm{d}} (J^*)^{-1} v $ where $v\sim F_{\rho,u,\mathrm{DP}}^*$.
\end{restatable}
\begin{proof}[Proof for Theorem \ref{thm:indirect_est_smooth}]
The proofs for parts 1 and 2 follow the proof of parts 1 and 2 in Theorem \ref{thm:indirect_est} while we use both Lemma \ref{lem:beta_dist} and Lemma \ref{lem:beta_dist_smooth}. For part 3, as $\frac{\partial s(\theta)}{\partial \theta}$ does not exist, we use the assumption \textbf{(A5s)} to obtain \eqref{eq:asymp_eqi_ind_s}, and the rest of the proof follows the proof of part 3 in Theorem \ref{thm:indirect_est}.
\end{proof}
\begin{proof}[Proof for Theorem \ref{thm:ada_indirect_est_smooth}]
The proofs for parts 1, 2, and 4  follow the proof of parts 1, 2, and 4 in Theorem \ref{thm:ada_indirect_est}. For part 3, we use the assumption \textbf{(A5s)} to obtain \eqref{eq:asymp_eqi_adi_s}, and the rest of the proof follows the proof of part 3 in Theorem \ref{thm:ada_indirect_est}.
\end{proof}

\subsection{Algorithms}

For easier reference, we summarize the indirect estimator and the adaptive indirect estimator in Algorithm \ref{alg:indest}.
The usage of parametric bootstrap with the indirect estimator is summarized in Algorithm \ref{alg:indCI} and Algorithm \ref{alg:indHT} for building CIs and conducting HTs, respectively.

\begin{algorithm}[t]
\caption{\texttt{indirectEstimator}($\Theta$, $s$, $F_u$, $F_{\mathrm{DP}}$, $G$, $\rho$, $\Omega$, $R$)}
INPUT: parameter space $\Theta \subseteq \mathbb{R}^q$, observed statistic $s\in \mathbb{B} \subseteq \mathbb{R}^p$, distributions $F_u$ and $F_{\mathrm{DP}}$ for random seeds $u$ and $u_{\mathrm{DP}}$, data generating equation $G(\theta,u)$ that $D=G(\theta,u)$ and $u\sim F_u$, objective function $\rho(\beta; D, u_{\mathrm{DP}})$ that $s=\argmin_{\beta\in \mathbb{B}} \rho(\beta; D, u_{\mathrm{DP}})$ and $u_{\mathrm{DP}}\sim F_{\mathrm{DP}}$, (optional) positive definite symmetric matrix $\Omega$, number of indirect estimator samples $R$.
\begin{algorithmic}[1]
  \setlength\itemsep{0em}
  \FOR {$r=1, \ldots, R$}
    \STATE Sample random seeds $u^r\sim F_u$ and $u_{\mathrm{DP}}^r\sim F_{\mathrm{DP}}$
    \STATE For a given $\theta$, define $s^r(\theta) := \argmin_{\beta\in \mathbb{B}} \rho(\beta; D^r(\theta), u_{\mathrm{DP}}^r)$ where $D^r(\theta):=G(\theta, u^r)$.
  \ENDFOR
  \STATE Let $u^{[R]} := (u^1, \ldots, u^R)$, $u_{\mathrm{DP}}^{[R]} := (u_{\mathrm{DP}}^1,\cdots, u_{\mathrm{DP}}^R)$.
  \IF{$\Omega$ is \texttt{NULL}}
    \STATE Let $m^{R}(\theta):=\frac{1}{R}\sum_{r=1}^R s^r(\theta)$ and  $S^{R}(\theta):=\frac{1}{R-1}\sum_{r=1}^R (s^r(\theta) - m^{R}(\theta))(s^r(\theta) - m^{R}(\theta))^{\intercal}$.
    \STATE Solve $\hat{\theta}:=\hat{\theta}_{\mathrm{ADI}}(s, u^{[R]}, u_{\mathrm{DP}}^{[R]}) := \argmin_{\theta \in \Theta}\l\|s - m^{R}(\theta)\r\|_{\l(S^{R}(\theta)\r)^{-1}}.$
  \ELSE
    \STATE Solve $\hat{\theta}:=\hat{\theta}_{\mathrm{IND}}(s, u^{[R]}, u_{\mathrm{DP}}^{[R]}) := \argmin_{\theta \in \Theta} \left\|s-\frac{1}{R} \sum_{r=1}^R s^r(\theta)\right\|_{{\Omega}}.$
  \ENDIF
\end{algorithmic}
OUTPUT: $\hat{\theta}$ 
\label{alg:indest}
\end{algorithm}

\begin{algorithm}[t]
\caption{\texttt{indirectEstimator\_PB\_CI}($\Theta$, $s$, $F_u$, $F_{\mathrm{DP}}$, $G$, $\rho$, $\Omega$, $\hat\tau$, $\hat\Sigma$, $\tau$, $B$, $R$, $\alpha$)}
INPUT: parameter space $\Theta \subseteq \mathbb{R}^q$, observed statistic $s\in \mathbb{B} \subseteq \mathbb{R}^p$, distributions $F_u$ and $F_{\mathrm{DP}}$ for random seeds $u$ and $u_{\mathrm{DP}}$, data generating equation $G(\theta,u)$ that $D=G(\theta,u)$ and $u\sim F_u$, objective function $\rho(\beta; D, u_{\mathrm{DP}})$ that $s=\argmin_{\beta\in \mathbb{B}} \rho(\beta; D, u_{\mathrm{DP}})$ and $u_{\mathrm{DP}}\sim F_{\mathrm{DP}}$, (optional) positive definite symmetric matrix $\Omega$, test statistic $\hat\tau: \mathbb{B} \rightarrow \mathbb{R}^d$,  auxiliary covariance of test statistic $\hat\Sigma: \mathbb{B} \rightarrow \mathbb{R}^{d\times d}$, parameter of interest $\tau:\Theta \rightarrow \mathbb{R}^d$, number of bootstrap samples $B$, number of indirect estimator samples $R$, significance level $\alpha$ (i.e., confidence level $(1-\alpha)$).
\begin{algorithmic}[1]
  \setlength\itemsep{0em}
  \STATE Let $\hat\theta:=$\texttt{indirectEstimator}($\Theta$, $s$, $F_u$, $F_{\mathrm{DP}}$, $G$, $\rho$, $\Omega$, $R$)  
  \FOR {$b=1, \ldots, B$}
    \STATE Sample random seeds $u_b\sim F_u$ and $u_{\mathrm{DP},b}\sim F_{\mathrm{DP}}$
    \STATE Generate parametric bootstrap sample $D_b:=G(\hat\theta,u_b)$
    \STATE Compute $s_b:=\argmin_{\beta\in \mathbb{B}} \rho(\beta; D_b, u_{\mathrm{DP},b})$
  \ENDFOR
  \STATE  Let $\hat\xi_{(j)}$ be the $j$th order statistic of $\l\{\l\|\hat\Sigma(s_b)^{-\frac{1}{2}}\l(\hat\tau(s_b) - \tau(\hat\theta)\r)\r\|_p\r\}_{b=1}^B$
\end{algorithmic}
OUTPUT:  $\hat C_{1-\alpha}(s)=\l\{\hat\tau(s)-\hat\Sigma(s)^{\frac{1}{2}}\xi ~\Big|~ \|\xi\|_p \leq \hat\xi_{(\lfloor (B+1)(1-\alpha)\rfloor)}\r\}$
\label{alg:indCI}
\end{algorithm}

\begin{algorithm}[t]
\caption{\texttt{indirectEstimator\_PB\_HT}($\Theta$, $\Theta_0$, $s$, $F_u$, $F_{\mathrm{DP}}$, $G$, $\rho$, $\Omega$, $\hat\tau$, $\hat\Sigma$, $\tau$, $B$, $R$, $\alpha$)}
INPUT: parameter space $\Theta \subseteq \mathbb{R}^q$, parameter space under null hypothesis $\Theta_0 \subseteq \Theta \subseteq \mathbb{R}^q$, observed statistic $s\in \mathbb{B} \subseteq \mathbb{R}^p$, distributions $F_u$ and $F_{\mathrm{DP}}$ for random seeds $u$ and $u_{\mathrm{DP}}$, data generating equation $G(\theta,u)$ that $D=G(\theta,u)$ and $u\sim F_u$, objective function $\rho(\beta; D, u_{\mathrm{DP}})$ that $s=\argmin_{\beta\in \mathbb{B}} \rho(\beta; D, u_{\mathrm{DP}})$ and $u_{\mathrm{DP}}\sim F_{\mathrm{DP}}$, (optional) positive definite symmetric matrix $\Omega$, test statistic $\hat\tau: \mathbb{B} \rightarrow \mathbb{R}^d$,  auxiliary covariance of test statistic  $\hat\Sigma: \mathbb{B} \rightarrow \mathbb{R}^{d\times d}$, parameter of interest $\tau:\Theta \rightarrow \mathbb{R}^d$, number of bootstrap samples $B$, number of indirect estimator samples $R$, significance level $\alpha$.
\begin{algorithmic}[1]
  \setlength\itemsep{0em}
  \STATE Compute studentized test statistic $T:=\inf_{\theta^* \in \Theta_0}\l\|\hat\Sigma(s)^{-\frac{1}{2}}\l(\hat\tau(s) - \tau(\theta^*)\r)\r\|_p$.
  \STATE Let $\hat\theta:=$\texttt{indirectEstimator}($\Theta$, $s$, $F_u$, $F_{\mathrm{DP}}$, $G$, $\rho$, $\Omega$, $R$)  
  \FOR {$b=1, \ldots, B$}
    \STATE Sample random seeds $u_b\sim F_u$ and $u_{\mathrm{DP},b}\sim F_{\mathrm{DP}}$
    \STATE Generate parametric bootstrap sample $D_b:=G(\hat\theta,u_b)$
    \STATE Compute $s_b:=\argmin_{\beta\in \mathbb{B}} \rho(\beta; D_b, u_{\mathrm{DP},b})$
    \STATE  Compute studentized test statistic $T_b:=\l\|\hat\Sigma(s_b)^{-\frac{1}{2}}\l(\hat\tau(s_b) - \tau(\hat\theta)\r)\r\|_p$
  \ENDFOR
  \STATE Compute the $p$-value where $p := \frac{1}{B+1}(1+\sum_{b=1}^B\mathbb{1}_{T_b\geq T})$
\end{algorithmic}
OUTPUT: $\mathbb{1}_{p\leq \alpha}$
\label{alg:indHT}
\end{algorithm}
 
\subsection{Hyperparameter tuning for location-scale normal}
We show the robustness of the adaptive indirect estimator by examining different settings on the number of generated samples $R$, the GDP parameter $\ep$, and the clamping parameter $U$, where the results are shown in Figures \ref{fig:CI_ADI_dist_R}, \ref{fig:CI_ADI_dist_gdp}, and \ref{fig:CI_ADI_dist_clamp}, respectively. In these figures, we denote the median of each distribution by the vertical line and denote $\mu^*$ and $\sigma^*$ by $*$. Note that we set the optimization region to $\Theta:=[-2,10]\times[10^{-6}, 10]$. In most cases, the medians match the true parameters, and the distributions are symmetric. In addition,
\begin{itemize}
    \item The number of generated samples $R$ does not significantly affect the sampling distribution of our estimates in Figure \ref{fig:CI_ADI_dist_R}. Therefore, we use $R=50$ for the experiments in this paper. 
    \item The sampling distribution is flatter when the GDP parameter $\ep$ is smaller, i.e., under stronger privacy guarantees, in Figure \ref{fig:CI_ADI_dist_gdp}. When $\ep=0.1$, the adaptive indirect estimator sometimes takes values at the boundary of $\Theta$, which may be due to the heavy tail of the sampling distribution, but the medians of the distribution are still around $\mu^*$ and $\sigma^*$, respectively.
    \item We set the clamping region to $[0,U]$ and we try $U=0.1,0.3,1,3,5$. In Figure \ref{fig:CI_ADI_dist_clamp}, $U=3$ gives the most concentrated sampling distribution. Our estimator does not perform well under $U=0.1$ since after clamping the true distribution $N(1,1)$ to $[0,0.1]$, most of the probability mass becomes the point mass at $0$ and $0.1$, and it is difficult to identify the distribution parameter before clamping.
\end{itemize}
Finally, we compare the adaptive indirect estimator to the indirect estimator with $\Omega_n=I$. We replace the adaptive indirect estimator in Figure \ref{fig:CI_ADI_dist_clamp} with the indirect estimator, and the results are shown in Figure \ref{fig:CI_IND_dist_clamp}. It turns out that the adaptive indirect estimator is the same as the indirect estimator when the clamping bound is $U=0.5,1,3,5$ if we ignore the optimization error, and the adaptive indirect estimator seems to be better than the indirect estimator under $U=0.1$.

\begin{figure}[t]
    \centering
    \includegraphics[width=0.9\linewidth]{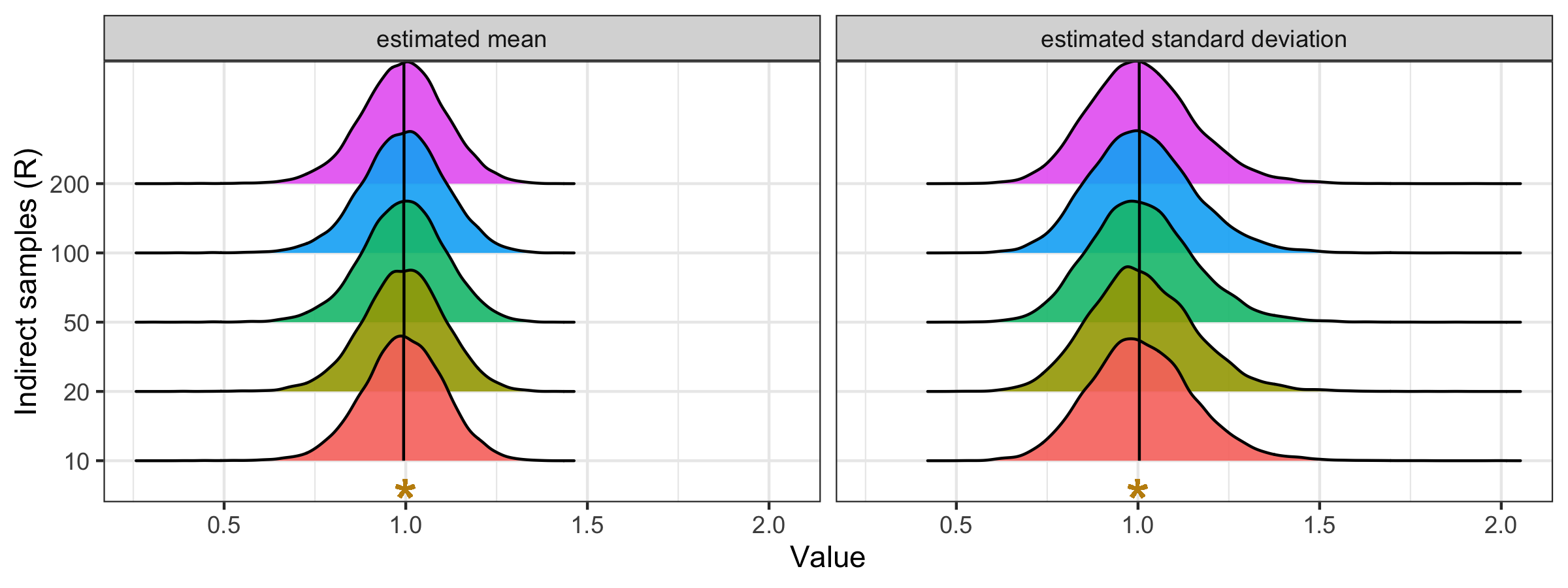}
    \caption{Comparison of the sampling distribution of the adaptive indirect estimates $\hat\theta_{\mathrm{ADI}}$ under different settings of the number of generated samples $R=10,20,50,100,200$ in the normal distribution setting in Section \ref{sec:pb:CI}.}\label{fig:CI_ADI_dist_R}
\end{figure}

\begin{figure}[t]
    \centering
    \includegraphics[width=0.9\linewidth]{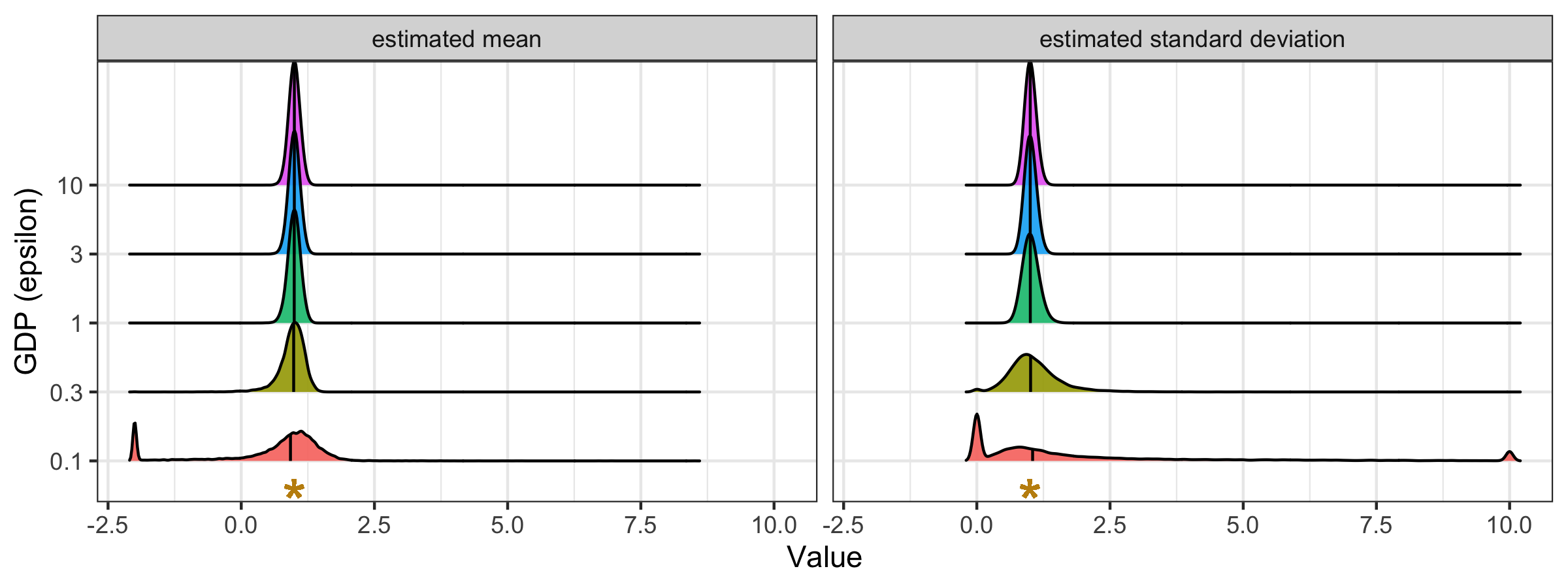}
    \caption{Comparison of the sampling distribution of the adaptive indirect estimates $\hat\theta_{\mathrm{ADI}}$ under different settings of the GDP parameter $\ep =0.1,0.3,1,3,10$ in the normal distribution setting in Section \ref{sec:pb:CI}.}\label{fig:CI_ADI_dist_gdp}
\end{figure}

\begin{figure}[t]
    \begin{minipage}{\textwidth}
        \centering
        \includegraphics[width=0.9\linewidth]{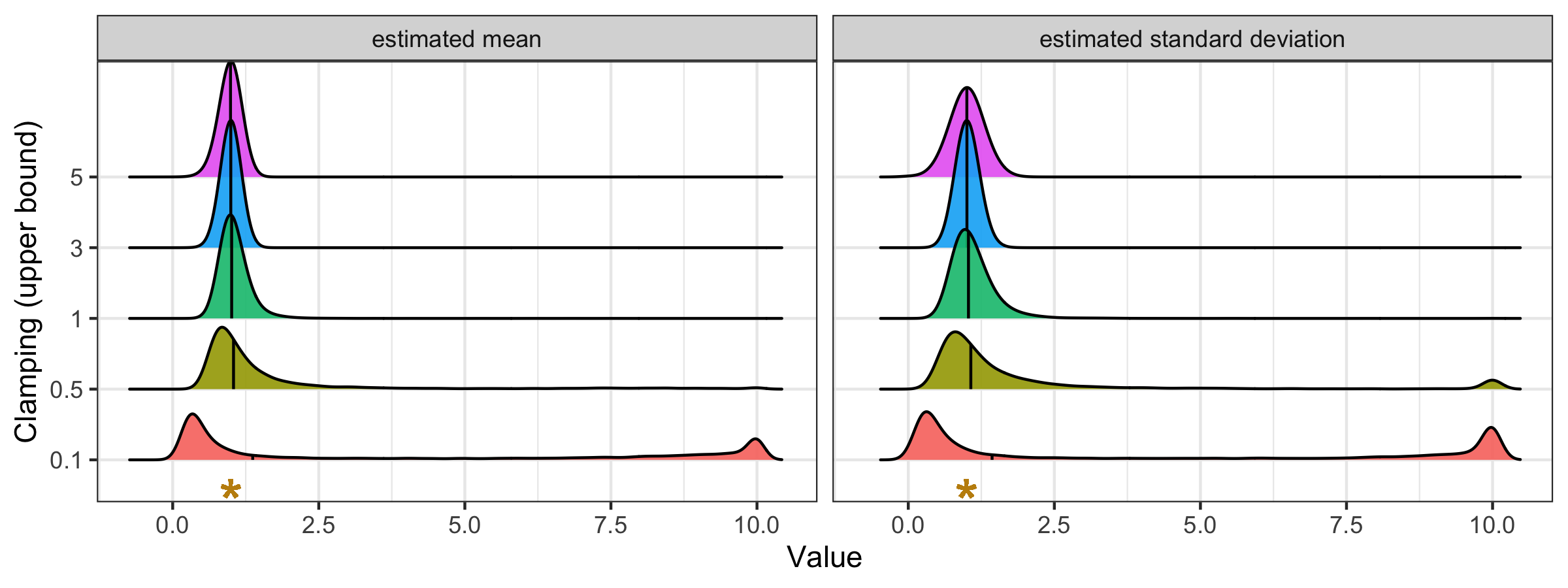}
        \caption{Comparison of the sampling distribution of the adaptive indirect estimates $\hat\theta_{\mathrm{ADI}}$ under different settings of the clamping parameter $U=0.1,0.5,1,3,5$ in the normal distribution setting in Section \ref{sec:pb:CI}.}\label{fig:CI_ADI_dist_clamp}
    \end{minipage}
    
    \vspace{20pt}
    
    \begin{minipage}{\textwidth}
        \centering
        \includegraphics[width=0.9\linewidth]{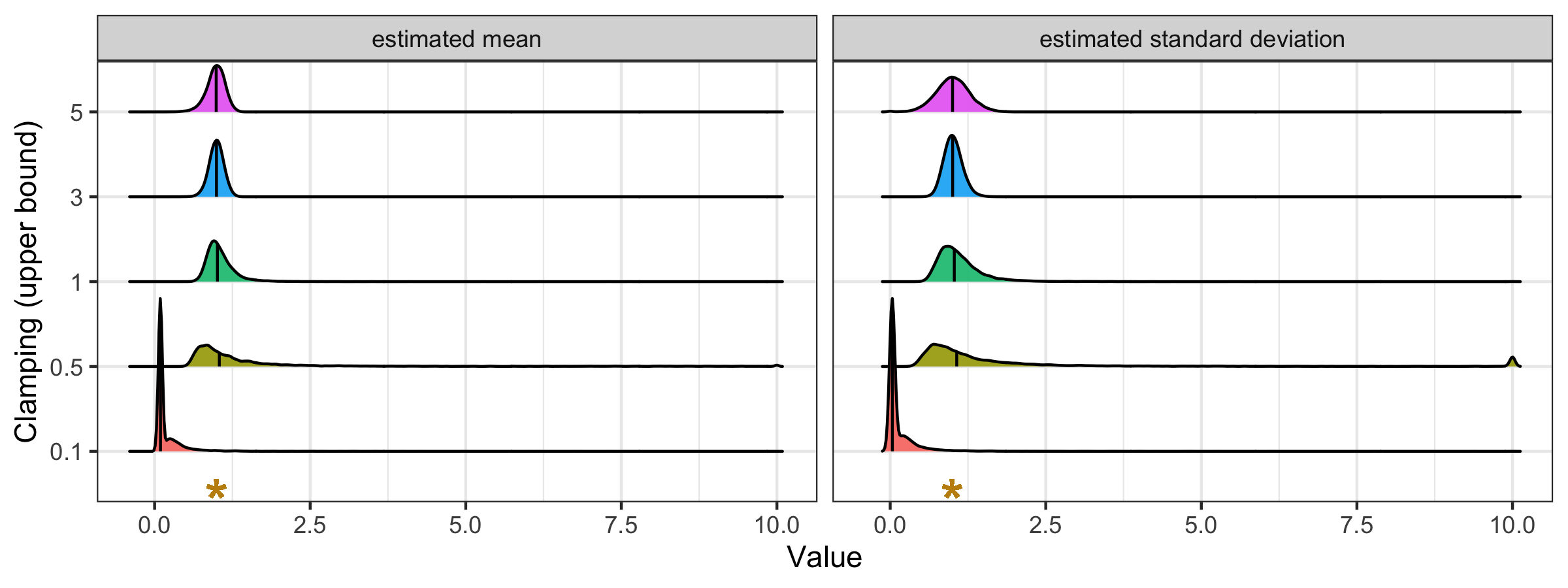}
        \caption{Comparison of the sampling distribution of the indirect estimates $\hat\theta_{\mathrm{IND}}$ under different settings of the clamping parameter $U=0.1,0.5,1,3,5$ in the normal distribution setting in Section \ref{sec:pb:CI}.}\label{fig:CI_IND_dist_clamp}
    \end{minipage}
\end{figure}
\begin{figure}[ht!]
    \centering
    \includegraphics[width=0.9\linewidth]{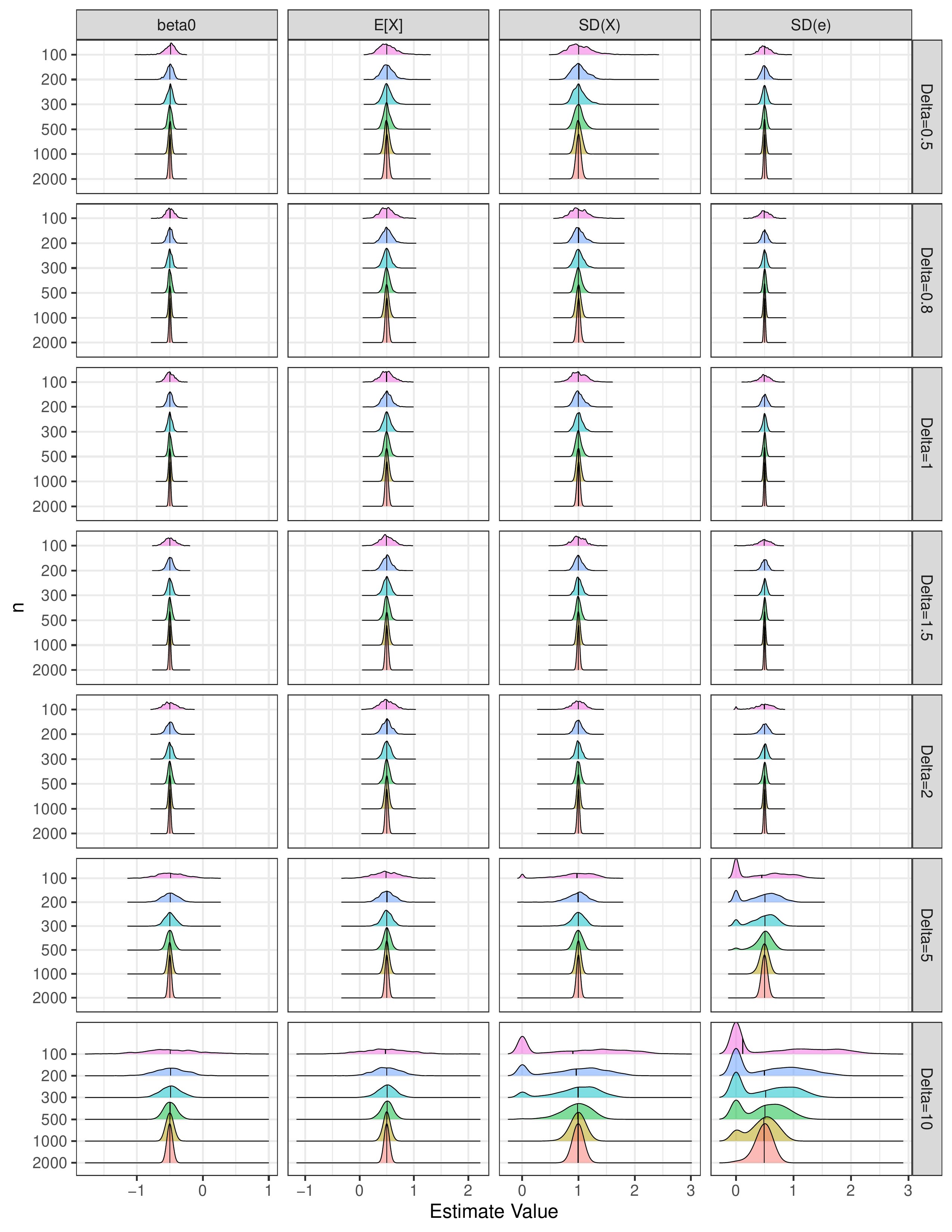}
    \caption{Comparison of the sampling distribution of the adaptive indirect estimates under different settings of the clamping parameter $\Delta$ and sample size $n$ in the hypothesis testing setting in Section \ref{sec:pb:HT}.}\label{fig:HT_ADI_dist}
\end{figure}

\subsection{Details for simple linear regression hypothesis testing}\label{app:linear_regresion_HT}

We have seen the improvement of using our adaptive indirect estimator  compared to the naive estimator in Figure \ref{fig:LR}. In Figure \ref{fig:HT_ADI_dist}, we further show the sampling distribution of our ADI estimator. The vertical line denotes the median of each distribution, and we can see that the medians match the true values and the distributions are symmetric in most cases, indicating that our method is robust across various clamping settings. 

In Remark \ref{rmk:approx_pivot}, we explain how to compute the approximate pivot for HTs in Section \ref{sec:pb:HT}.

\begin{remark}\label{rmk:approx_pivot}
    
    We denote $\hat{\theta}_{\mathrm{ADI}}$ as $\hat\theta$.
    From Theorem \ref{thm:ada_indirect_est}, we know that the asymptotic variance of $\sqrt{n}(\hat{\theta} - \theta^*)$ is $\l((B^*)^{\intercal} \Sigma(\theta^*)^{-1} B^*\r)^{-1}$. Therefore, based on assumption \textbf{(A8)}, we estimate $\Sigma(\theta^*)$ by the sample variance of $\l\{\sqrt{n}\l(s^r(\hat{\theta})\r)\r\}_{r=1}^R$ where $s^r(\hat{\theta})$ is calculated in the same way as in the indirect estimator. Based on assumptions \textbf{ (A3, A5)}, we evaluate $b(\theta^*)$ by $\mathbb{E}(s(\hat{\theta}))$\footnote{For this setting, we evaluate the expectation by numerical integral. In general, the expectation can also be evaluated by the Monte Carlo approximation.}, and we estimate $B^*:=\frac{\partial b(\theta^*)}{\partial \theta}$ by the finite difference $\frac{\mathbb{E}(s(\hat{\theta} + \delta \theta_0^i)) - \mathbb{E}(s(\hat{\theta}))}{\delta}$ where $\theta_0^i\in\mathbb{R}^q$ has all entries being 0 except for the $i$th entry being $1$, $i=1,\ldots,q$, and $\delta:=10^{-6}$. 
    We denote the two estimates of $\Sigma(\theta^*)$ and $B^*$ by 
    $\hat\Sigma(\hat\theta)$ and $\hat B(\hat\theta)$, 
    respectively. The approximate pivot is 
    $$T:=n(\hat{\theta} - \theta^*)^\intercal \l((\hat B(\hat\theta))^{\intercal} \hat\Sigma^{-1}(\hat\theta) \hat B(\hat\theta)\r)^{-1} (\hat{\theta} - \theta^*).$$ 
    As our hypothesis test is on the $\beta_1^*$, we let $\hat{\theta}^{(1)}$ be the first entry of $\hat{\theta}$ corresponding to $\beta_1^*$, and $\hat\sigma_{\beta_1}^2(\hat\theta)$ be the first element in the covariance matrix estimation $\l((\hat B(\hat\theta))^{\intercal} \hat\Sigma^{-1}(\hat\theta) \hat B(\hat\theta)\r)^{-1}$ corresponding to the variance of $\sqrt{n}\l(\hat{\theta}^{(1)}\r)$. Then, it is equivalent to use the following setting in the $\l|\frac{\hat\tau(s) - \tau(\theta^*)}{\hat\sigma(s)}\r|$ and $\l|\frac{\hat\tau(s_b) - \tau(\hat\theta)}{\hat\sigma(s_b)}\r|$ in Algorithm \ref{alg:indHT}, 
    \[
    \hat\tau(s):=\hat{\theta}^{(1)}(s),\quad \tau(\theta^*):=(\theta^*)^{(1)}=0,\quad \hat\sigma(s):=\frac{\hat\sigma_{\beta_1}(\hat{\theta}(s))}{\sqrt{n}},
    \]
    \[
    \hat\tau(s_b):=\hat{\theta}^{(1)}(s_b),\quad \tau(\hat{\theta}):=(\hat{\theta})^{(1)},\quad \hat\sigma(s_b):=\frac{\hat\sigma_{\beta_1}(\hat{\theta}(s_b))}{\sqrt{n}}.
    \]
\end{remark}

\begin{remark}\label{rmk:HT_estimate}
    When the alternative hypothesis is true, i.e., $s\sim \mathbb{P}_{\theta_1}$ where $\theta_1\in\Theta_1$, any estimator restricted to the null hypothesis, $\hat\theta_0:\Omega \rightarrow \Theta_0$, is not a consistent estimator of $\theta_1$. If we use $\hat\theta_0$ in PB, since the PB estimator $\hat\tau(s_b)$ is based on $s_b\sim \mathbb{P}_{\hat\theta_0(s)}(\cdot | s)$, the consistency of $\hat\tau(s_b)$ may not hold. Therefore, we may not have $\pi_n(\theta_1) \rightarrow 1$, and the PB HTs based on $\hat\theta_0$ are asymptotically of level $\alpha$ but not necessarily asymptotically consistent.
\end{remark}

\subsection{Logistic regression mechanism details}\label{sec:logisticApp}
The DP mechanism used in Section \ref{s:logistic} is the same as in \citet{awan2025simulation}, which uses techniques from \citet{awan2021structure}, including the sensitivity space, $K$-norm mechanisms, and objective perturbation. To aid the reader, we include the essentials in this section.

\begin{definition}[Sensitivity space: \citealp{awan2021structure}]\label{AdjacentOutput}
Let $T: \mscr X^n\rightarrow \RR^m$ be any function. The \emph{sensitivity space} of $T$ is 
\[S_T = \l\{ u \in \RR^m \middle| \begin{array}{c}\exists X,X'\in \mscr X^n \text{ s.t } d(X,X')=1\\ \text{and }u=T(X)-T(X')\end{array}\r\}.\]
\end{definition}

A set $K\subset \RR^m$ is a \emph{norm ball} if $K$ is 
\begin{inparaenum}[1)]\item convex, \item bounded,  \item absorbing: $\forall u \in \RR^m$, $\exists c>0$ such that $u\in cK$, and \item
 symmetric about zero: if $u\in K$, then $-u\in K$. \end{inparaenum} 
If $K\subset \RR^m$ is a norm ball, then $\lVert u \rVert_K = \inf \{c\in \RR^{\geq 0}\mid u\in cK\}$ is a norm on $\RR^m$. 

\begin{algorithm}
\caption{Sampling  $\propto \exp(-c \lVert V\rVert_\infty)$   \citep{steinke2016between}}
\scriptsize
INPUT: $c$, and dimension $m$
\begin{algorithmic}[1]
\STATE Set $U_j \iid U(-1,1)$ for $j=1,\ldots, m$
\STATE Draw $r \sim \mathrm{Gamma}(\al = m+1,\beta = c)$
\STATE Set $V = r\cdot (U_1,\ldots, U_m)^\top$
\end{algorithmic}
OUTPUT: $V$
\label{alg:SampleLinfty}
\end{algorithm}

The sensitivity of a statistic $T$ is the largest amount that $T$ changes when one entry of $T$ is modified. Geometrically, the sensitivity of $T$ is the largest radius of $S_T$ measured by the norm of interest. For a norm ball $K \subset \RR^m$, the $K$-norm sensitivity of $T$ is 
\[\Delta_K(T) = \sup_{d(X,X')=1} \lVert T(X) - T(X')\rVert_K = \sup_{u \in S_T} \lVert u \rVert_K .\]

\begin{algorithm}
\caption{Rejection Sampler for  $K$-Norm Mechanism \citep{awan2021structure}}
\label{RejectionAlgorithm}
\scriptsize
INPUT: $\ep$, statistic $T(X)$, $\ell_\infty$ sensitivity $\Delta_\infty(T)$, $K$-norm sensitivity $\Delta_K(T)$
\begin{algorithmic}[1]
\STATE Set $m = \mathrm{length} (T(X))$.
\STATE Draw $r \sim \mathrm{Gamma}(\al = m+1, \beta = \ep/\Delta_K(T))$
\STATE Draw $U_j \iid \mathrm{Uniform}(-\Delta_\infty(T), \Delta_\infty(T))$ for $j=1,\ldots, m$
\STATE Set $U = (U_1,\ldots, U_m)^\top$
\STATE If $U\in K$, set $N_K=U$, else go to 3)
\STATE Release $T(X) + r\cdot N_K$.
\end{algorithmic}
\label{alg:sampleK}
\end{algorithm}

\begin{algorithm}
\caption{$\ell_\infty$ Objective Perturbation \citep{awan2021structure}}
\scriptsize
INPUT: $X\in \mscr{ X}^n$, $\ep>0$, a convex set $\Theta \subset \RR^m$, a convex function $r: \Theta\rightarrow \RR$, a convex loss  $\hat{\mscr  L}(\ta; X) =\frac1n \sum_{i=1}^n \ell(\ta;x_i)$ defined on $\Theta$ such that $\nabla^2 \ell(\ta;x)$ is continuous in $\ta$ and $x$, $\De>0$ such that $\sup_{x,x'\in \mscr X} \sup_{\ta\in \Ta}\lVert \nabla \ell(\ta;x) - \nabla\ell(\ta;x')\rVert_\infty\leq \De$,  $\la>0$  is an upper bound on the eigenvalues of $\nabla^2\ell(\ta;x)$ for all $\ta\in \Theta$ and $x\in \mscr X$, and $q\in (0,1)$.
\begin{algorithmic}[1]
  \setlength\itemsep{0em}
  \STATE Set $\ga = \frac{\la}{\exp({\ep(1-q)})-1}$
\STATE Draw $V\in \RR^m$ from the density $f(V;\ep, \De)\propto \exp(-\frac{\ep q}{\De}\lVert V\rVert_\infty)$ using Algorithm \ref{alg:SampleLinfty}
\STATE Compute $\ta_{DP} = \arg\min_{\ta\in \Theta} \hat{\mscr L}(\ta;X) +\frac1n r(\theta)+ \frac{\ga}{2n} \ta^\top \ta + \frac{V^\top \ta}{n}$
\end{algorithmic}
OUTPUT: $\ta_{DP}$
\label{alg:ExtendedObjPert}
\end{algorithm}

The DP mechanism has two parts, each resulting in a 2-dimensional output. The first part is used to infer about the parameters $a$ and $b$, which consists of noisy versions of the first two moments: $(\sum_{i=1}^n z_i,\sum_{i=1}^n z_i^2)+N_K$, where $N_K$ is from a $K$-norm mechanism discussed in the following paragraphs. The second part uses Algorithm \ref{alg:ExtendedObjPert} to approximate $\beta_0$ and $\beta_1$.  \citet[Section 4.1]{awan2021structure} derived that $\Delta_\infty=2$ and $\lambda=m/4$, where $m$ is the number of regression coefficients (2 in our case); in \citet[Section 4.2]{awan2021structure}, a numerical experiment suggested that $q=0.85$ optimized the performance of the mechanism. Algorithm \ref{alg:ExtendedObjPert} satisfies $\ep$-DP according to \citet[Theorem 4.1]{awan2021structure}. 

\citet{awan2021structure} showed that the optimal $K$-norm mechanism for a statistic $T(X)$ uses the norm ball, which is the convex hull of the sensitivity space. For the statistic $T(X) = (\sum_{i=1}^n z_i,\sum_{i=1}^n z_i^2)$, where $z_i \in [0,1]$, sensitivity space is 
\begin{equation}\label{eq:sensitivity}
    S_T = \left\{(u_1,u_2) \middle| \begin{array}{c}
u_1^2\leq u_2\leq 1-(u_1-1)^2\\
\text{or }\quad (1+u_1)^2-1\leq u_2\leq u_1^2
\end{array}\right\},
\end{equation}
and the convex hull is 
\begin{equation}\label{eq:hull}
\mathrm{Hull}(S_T) = \left\{ (u_1,u_2) \middle| 
\begin{array}{cc}
(u_1+1)^2-1\leq u_2\leq -u_1^2,& u_1\leq -1/2\\
u_1-1/4\leq u_2\leq u_1+1/4,& -1/2<u_1\leq 1/2\\
u_1^2\leq u_2\leq 1-(u_1-1)^2,&1/2\leq u_1
\end{array}\right\}.\end{equation}
We use $K=\mathrm{Hull}(S_T)$ as the norm ball in Algorithm \ref{alg:sampleK}, which satisfies $\ep$-DP \citep{hardt2010geometry}. 

In the simulation generating the results in Figure \ref{fig:logreg_comparison} of Section \ref{s:logistic}, we let $a^*=b^*=0.5$, $\beta_0^*=0.5$, $\beta_1^*=2$, $R=200$, $\alpha=0.05$, $n=100,$ 200, 500, 1000, 2000, and $\ep=0.1$, 0.3, 1, 3, 10 in $\ep$-DP. For our method, the private statistics we release and use are $s=(\tilde{\beta}_0, \tilde{\beta}_1, \tilde T(X))$ where $(\tilde{\beta}_0$, $\tilde{\beta}_1)$ are the output of Algorithm \ref{alg:ExtendedObjPert} with $0.9\ep$-DP, $q=0.85$, $\lambda=1/2$, and $\tilde T(X)$ is the output of Algorithm \ref{alg:sampleK},  with private budget $0.1\ep$. When using Algorithm \ref{alg:SampleLinfty} and \ref{alg:sampleK}, we let $m=2$ as it is the dimension of the output. 
We compute the confidence interval within the range $[-10,10]$.

Note that in Figure \ref{fig:logreg_comparison}, despite the width of the confidence intervals produced by the ADI appearing to be significantly larger than those produced by the naive estimator, a quick comparison to Gaussian quantiles shows that the increase in width is an amount proportional to the improvement in coverage. Taking the example of epsilon = 10 and n = 100 in Figure 6, we see that the naive method yields an average interval width of 1.12 with 74\% empirical coverage, corresponding to a standard normal quantile of approximately 1.13. Our method, in contrast, achieves 89\% coverage with an average width of 1.59, which corresponds to a quantile of approximately 1.60. Taking the ratio of these quantiles ($1.60 / 1.13 \approx 1.42$) reflects the relative increase in width required to achieve the higher coverage. The ratio of our method’s width to the naive method’s width ($1.59 / 1.12 \approx 1.42$) aligns almost exactly with this theoretical scaling, indicating that the increase in width is proportional to the coverage gain. In short, the larger widths of confidence intervals produced by our method compared to the naive reflect the additional uncertainty needed to achieve valid coverage.

\subsection{Naive Bayes, Log-Linear Model Simulation}\label{sec:naive}
\begin{table}[t]
\centering
\caption{True Distribution Parameters for the Naive Bayes Classifier}
\label{tab:nb_params}
\begin{tabular}{@{}lccc@{}}
\toprule
\textbf{Parameter} & \textbf{Value ($j=0$)} & \textbf{Value ($j=1$)} & \textbf{Marginal $P(Y=i)$} \\ \midrule
$P(X_1 = j \mid Y=0)$ & 0.7470 & 0.2530 & 0.3890 \\
$P(X_2 = j \mid Y=0)$ & 0.0755 & 0.9245 & \\ \midrule
$P(X_1 = j \mid Y=1)$ & 0.4000 & 0.6000 & 0.6110 \\
$P(X_2 = j \mid Y=1)$ & 0.3460 & 0.6540 & \\ \bottomrule
\end{tabular}
\end{table}

\begin{figure}[]
    \centering
    \includegraphics[width=.9\linewidth, trim={0 0 0 0},clip]{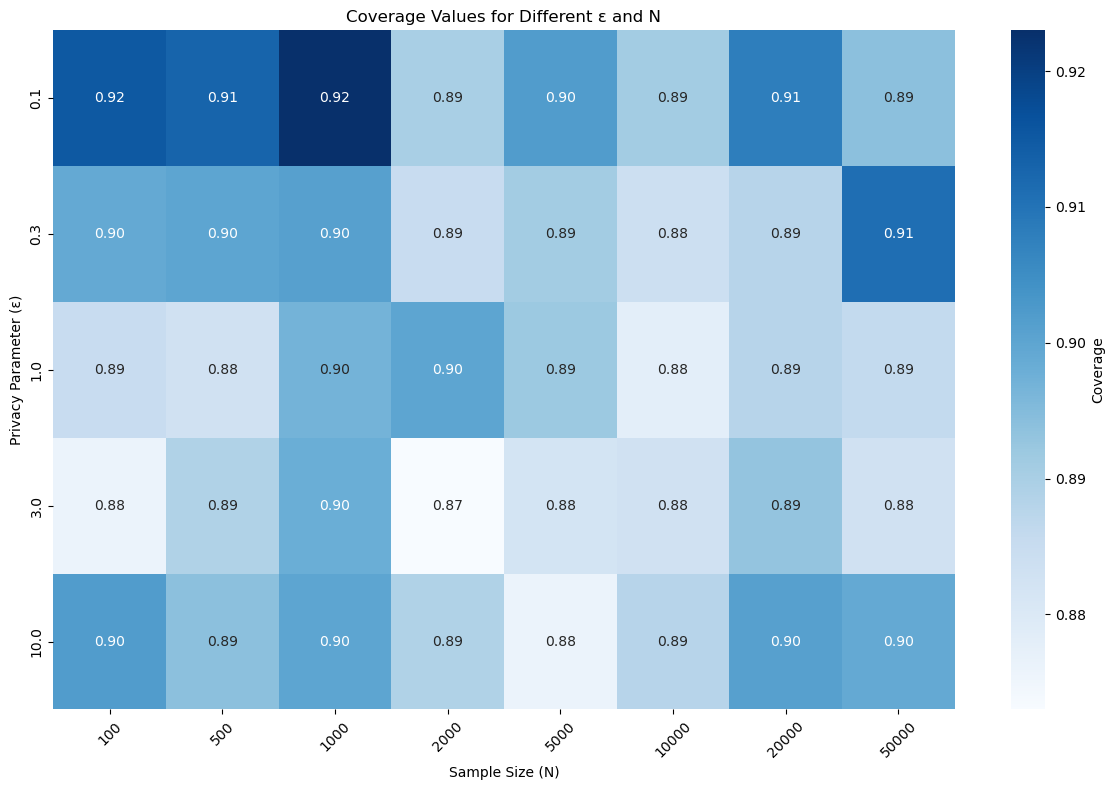}
    \includegraphics[width=.9\linewidth, trim={0 0 0 0},clip]{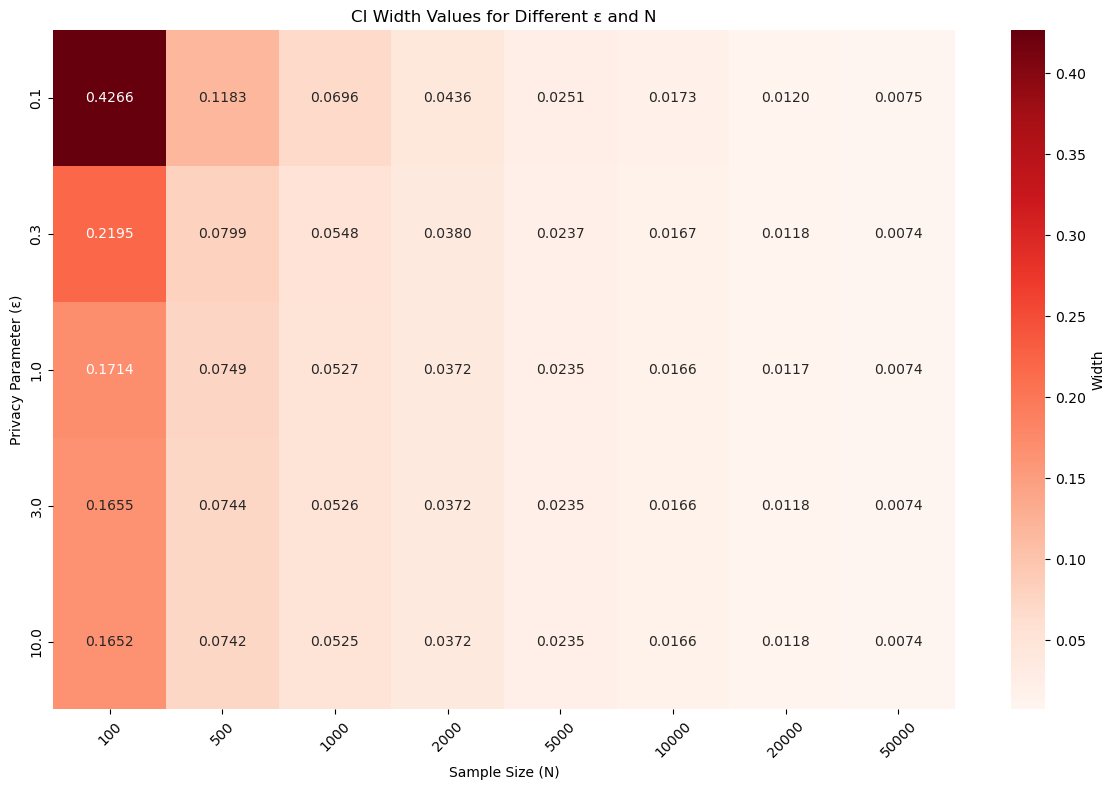}
    \caption{ Comparison of the coverage and width of confidence intervals for the true $P(Y=1)$ at 90\% confidence using our proposed estimator for the Naive Bayes classifier}\label{fig:logreg_comparison}
    \label{fig:NB}
\end{figure}

We include a simulation study based on a log-linear Naive Bayes classifier to demonstrate that our proposed estimator remains valid under model mis-specification and when surrogate models are used to approximate discrete data-generating processes, illustrating the theory developed in Section \ref{sec:pb:non-smooth}.

Similar to \citep{karwa2015private,ju2022data}, let $x = (x_1, \ldots, x_K)$ denote the categorical feature vector where each $x_k \in \{1, 2, \ldots, J_k\}$, and let $y \in \{1, 2, \ldots, I\}$ be the response class. Each input–output pair $(x, y)$ forms one record in the confidential database, with $N$ i.i.d. copies of $(x, y)$. The Naive Bayes model assumes conditional independence among features given the class label,
\[P(x \mid y) = \prod_{k=1}^{K} P(x_k \mid y),\]
with parameters $p_{ij}^k = P(x_k = j \mid y = i)$ and $p_i = P(y = i)$. The sufficient statistics of this model are the cell counts $n_{ij}^k = \#\{y = i, x_k = j\}$, which record feature–class co-occurrences.

We treat these counts as the confidential query $s$ and apply the  mechanism by adding noise $G_{ijk} \stackrel{\text{iid}}{\sim} N(0, \sqrt{2K}/\ep)$, yielding privatized noisy counts $m_{ij}^k = n_{ij}^k + G_{ijk}$ that satisfy $\ep$-GDP. Note that while \citet{karwa2015private,ju2022data} used Laplace noise to satisfy $\ep$-DP, our example uses Gaussian noise, which achieves a superior scaling in terms of $K$ \citep{dong2022gaussian}. Our goal is to construct valid confidence intervals for the parameters $p_{ij}^k$ based solely on these noisy sufficient statistics.  

While this privacy mechanism does not involve clamping, a plug-in version of the non-private estimators would involve a ratio of noisy counts; since there is noise added in the denominator, these estimators are not unbiased. Since the model is discrete, we require a surrogate model as justified by Section~\ref{sec:pb:non-smooth}: when simulating the data, we approximate the multinomial distribution of the sufficient statistics with a multivariate normal distribution. 

We use 400 parametric bootstrap samples and 40 synthetic samples for indirect inference. We record the true values of the parameters used in Table \ref{tab:nb_params} and our results are in Figure \ref{fig:NB}.  Across all settings, our proposed adaptive indirect estimator achieves coverage very close to the nominal level of  90\%, with low average width. Thus, the Naive Bayes study illustrates that our estimator can operate effectively in non-smooth or discrete settings through the use of continuous surrogate models, confirming the theoretical robustness discussed earlier.



\bibliography{all-biblatex}

\end{document}